\documentclass[a4paper,UKenglish,cleveref,thm-restate]{lipics-v2021}

\usepackage{tikz}

\usepackage{mathtools}

\usepackage{bm}

\usepackage{amsfonts}

\usepackage{amsmath}

\usepackage{amssymb}

\usepackage{amsthm}

\usepackage{wrapfig}

\usepackage{extpfeil}

\usepackage{etoolbox}%

\usepackage{xcolor}

\usepackage{booktabs}

\usepackage{cases}

\usepackage{semantic}

\usepackage{relsize}

\usepackage{pgf}

\newcommand{\equal}{\mathrel{=}}
\newcommand{\rightvdash}{\mathrel{\vdash}}
\newcommand{\nrightvdash}{\mathrel{\not{\vdash}}}

\newcommand{\lsem}{[\![}
\newcommand{\rsem}{]\!]}
\newcommand{\minus}{-}

\usepackage{xstring}

\usepackage{xspace}

\usepackage{pifont}

\usepackage{nicefrac}

\usepackage[inline,shortlabels]{enumitem}

\newcommand*{\envqed}{\hfill$\blacksquare$}
\newcommand{\exqed}{\qed}

\newcommand{\cmark}{\ding{51}}
\newcommand{\xmark}{\ding{55}}

\newcommand{\dom}[1]{\operatorname{dom}\!\left({#1}\right)}

\definecolor{yellowhighlight}{RGB}{255, 255, 64}
\definecolor{orangehighlight}{RGB}{255, 204, 153}

\newcommand{\Var}{\textsc{Var}\xspace}
\newcommand{\Val}{\textsc{Val}\xspace}
\newcommand{\PVar}{\textsc{PVar}\xspace}
\newcommand{\Act}{\textsc{Act}\xspace}

\newcommand{\basetype}[1][B]{\textsf{#1}}
\newcommand{\basetypek}[1]{\textsf{B}_{#1}}

\newcommand*{\eg}{\textit{e.g.,}\xspace}

\newcommand*{\ie}{\textit{i.e.,}\xspace}

\newcommand*{\etal}{\textit{et~al.}\xspace}
\newcommand*{\wrt}{w.r.t.\xspace}

\newcommand*{\resp}{resp.\xspace}

\newcommand{\ltsrule}[1]{{\normalfont[\textsc{#1}]}}

\newcommand{\pmsgNoVIdx}[1]{\texttt{l}_{#1}(v)}
\newcommand{\pmsgk}[2]{\texttt{#1}(#2)}

\newcommand{\psnd}[1]{\triangleleft \texttt{l}_{#1}(v_{#1})}
\newcommand{\psndk}[2]{\triangleleft \pmsgk{#1}{#2}}

\newcommand{\pstring}[1]{\texttt{"{#1}"}}

\newcommand{\prcv}[2]{\triangleright \big\{\texttt{l}_{#1}(x_{#1}).P_{#1}\big\}_{#1 \in #2}}
\newcommand{\prcvk}[2]{\triangleright \pmsgk{#1}{#2}}

\newcommand{\pmu}[1]{\mu_#1}

\newcommand{\pif}[3]{\textsf{if }#1\textsf{ then }#2\textsf{ else }#3}
\newcommand{\true}{\textsf{tt}}
\newcommand{\false}{\textsf{ff}}

\newcommand{\pnil}{\bm{0}}

\newcommand{\sndact}[2]{\triangleleft \texttt{#1}(#2)}

\newcommand{\rcvactk}[2]{\triangleright \texttt{#1}(#2)}

\newcommand{\stbra}[2]{\& \big\{?\texttt{l}_{#1}(x_{#1}:\textsf{B}_{#1})[A_{#1}].S_{#1}\big\}_{#1 \in #2}}
\newcommand{\stsel}[2]{\oplus \big\{!\texttt{l}_{#1}(x_{#1}:\textsf{B}_{#1})[A_{#1}].S_{#1}\big\}_{#1 \in #2}}
\newcommand{\stbraAi}[2]{\& \big\{?\texttt{l}_{#1}(x_{#1}:\textsf{B}_{#1})[A'_{#1}].S_{#1}\big\}_{#1 \in #2}}
\newcommand{\stselAi}[2]{\oplus \big\{!\texttt{l}_{#1}(x_{#1}:\textsf{B}_{#1})[A'_{#1}].S_{#1}\big\}_{#1 \in #2}}
\newcommand{\strec}[1]{\textsf{rec }#1}

\newcommand{\stbraNA}[2]{\& \big\{?\texttt{l}_{#1}(\textsf{B}_{#1}).S_{#1}\big\}_{#1 \in #2}}
\newcommand{\stselNA}[2]{\oplus \big\{!\texttt{l}_{#1}(\textsf{B}_{#1}).S_{#1}\big\}_{#1 \in #2}}

\newcommand{\selact}[1]{\triangleleft \texttt{#1}}
\newcommand{\selactk}[1]{\triangleleft \texttt{l}_#1}
\newcommand{\braact}[1]{\triangleright \texttt{#1}}
\newcommand{\braactk}[1]{\triangleright \texttt{l}_#1}

\newcommand{\fv}[1]{\textbf{fv}(#1)}
\newcommand{\fpv}[1]{\textbf{fpv}(#1)}

\newcommand{\revoke}{\texttt{Rvk}\xspace}
\newcommand{\auth}{\texttt{Auth}\xspace}
\newcommand{\get}{\texttt{Get}\xspace}
\newcommand{\succs}{\texttt{Succ}\xspace}
\newcommand{\fail}{\texttt{Fail}\xspace}

\newcommand{\monA}{$\textit{monitor}_{a}$\xspace}
\newcommand{\monB}{$\textit{monitor}_{b}$\xspace}
\newcommand{\server}{\textit{server}\xspace}
\newcommand{\client}{\textit{client}\xspace}

\newcommand{\lab}[1][l]{\texttt{#1}\xspace}
\newcommand{\labk}[1]{\texttt{#1}}

\newcommand{\sauth}{\textit{S}_\textit{auth}\xspace}
\newcommand{\pauth}{\textit{P}_\textit{auth}\xspace}
\newcommand{\pres}{\textit{P}_\textit{res}\xspace}
\newcommand{\mauth}{\textit{M}_\textit{auth}\xspace}
\newcommand{\sab}{\textit{S}_{ab}\xspace}
\newcommand{\pabc}{\textit{P}_{abc}\xspace}

\newcommand{\sef}{\textit{S}_{ef}\xspace}
\newcommand{\pef}{\textit{P}_{ef}\xspace}
\newcommand{\mef}{\textit{M}_{ef}\xspace}

\newcommand{\envsndOP}{\blacktriangle}
\newcommand{\envrcvOP}{\blacktriangledown}

\newcommand{\monout}[2]{\envsndOP \pmsgk{#1}{#2}}
\newcommand{\monin}[2]{\envrcvOP \big\{\texttt{l}_{#1}(x_{#1}).M_{#1}\big\}_{#1 \in #2}}

\newcommand{\monrcv}[2]{\triangleright \big\{\texttt{l}_{#1}(x_{#1}).M_{#1}\big\}_{#1 \in #2}}
\newcommand{\monsndk}[2]{\triangleleft \pmsgk{#1}{#2}}
\newcommand{\monmu}[1]{\mu_{#1}}
\newcommand{\monsndact}[2]{\triangleleft \texttt{#1}(#2)}

\newcommand{\monrcvactkNoVIdx}[1]{\triangleright \texttt{l}_{#1}(v)}

\newcommand{\moninactk}[1]{\blacktriangledown \texttt{l}_#1(v_#1)}
\newcommand{\monoutact}[2]{\blacktriangle \texttt{#1}(#2)}

\newcommand{\no}[2]{\textsf{no}^{#1}_{#2}}
\newcommand{\nops}{\no{S}{P}\xspace}
\newcommand{\noes}{\no{S}{E}\xspace}
\newcommand{\nopd}{\no{D}{P}\xspace}
\newcommand{\noed}{\no{D}{E}\xspace}

\newcommand{\instr}[2]{\langle {#1};{#2} \rangle}
\newcommand{\synth}[1]{\lsem #1 \rsem\xspace}

\newcommand{\envsnd}[2]{\envsndOP \texttt{#1}(#2)}
\newcommand{\envrcv}{\envrcvOP \big\{\texttt{l}_{i}(x_i:B_i).M_{i}\big\}_{i\in I}}

\newcommand{\alignhfill}[1]{\tag*{\llap{#1}}}

\newcommand{\accepted}[2]{\textbf{accepted}(#1,#2)}
\newcommand{\match}[2]{\textbf{match}(#1, #2)}

\newcommand{\lchannels}{\texttt{lchannels}\xspace}
\newcommand{\CM}{CM\xspace}

\newcommand{\synthbra}{\envrcvOP\mathlarger{\{}\,\lab_i(x_i:\basetype_i).\textsf{if } A_i \textsf{ then} \triangleleft \lab_i(x_i).\lsem S_i\rsem \textsf{ else no}^D_E\mathlarger{\}}_{i\in I}}
\newcommand{\synthsel}[2]{\triangleright\mathlarger{\{}\,\lab_{#1}(x_{#1}:\basetype_{#1}).\textsf{if } A_{#1} \textsf{ then} \envsndOP l_{#1}(x_{#1}).\lsem S_{#1}\rsem \textsf{ else no}^D_P\mathlarger{\}}_{#1\in #2}}
\newcommand{\condstmtinternal}[1]{\textsf{if } A_{#1}[\nicefrac{v_{#1}}{x_{#1}}] \textsf{ then} \triangleleft \lab_{#1}(v_{#1}).\lsem S_{#1}\rsem[\nicefrac{v_{#1}}{x_{#1}}] \textsf{ else no}^D_E}
\newcommand{\condstmtexternal}[1]{\textsf{if } A_{#1}[\nicefrac{v_{#1}}{x_{#1}}] \textsf{ then} \envsndOP \lab_{#1}(v_{#1}).\lsem S_{#1}\rsem[\nicefrac{v_{#1}}{x_{#1}}] \textsf{ else no}^D_P}
\newcommand{\pwrtk}[1]{\triangleleft l_{#1}(v_{#1})}

\newcommand{\isValueB}[2]{\textsf{isB}_{#1}({#2})}

\newcommand{\hlightlightyellow}[1]{{\setlength{\fboxsep}{0pt}\colorbox{lightgray}{#1}}}%

\newcommand{\stmonitor}{\texttt{STMonitor}\xspace}

\definecolor{dkgreen}{rgb}{0,0.6,0}
\definecolor{gray}{rgb}{0.5,0.5,0.5}
\definecolor{mauve}{rgb}{0.58,0,0.82}

\newcommand*{\pill}[2][white]{%
  \tikz[baseline={([yshift=-.7ex]current bounding box.center)}]{
  \node[solid, draw, semithick, fill=#1, text=black, minimum width=0.84em, minimum height=0.8em, inner xsep=0.16em, inner ysep=0, rounded corners=0.4em, font=\sffamily\scriptsize, align=center] (char) {#2};}%
}

\usetikzlibrary {
  matrix,
  shapes,
  arrows,
  shadows,
  calc,
  chains,
  decorations.pathmorphing,
  decorations.text,
  arrows.meta,
  patterns,
  fit,
  bending,
  shadows,
  shapes.geometric,
  snakes
}

\tikzset{
  label/.style={
    font=\scriptsize\itshape,
    inner sep=0.4em
  },
  point/.style={
    circle,
    fill=black,
    text width=0.3em,
    inner sep=0
  },
  state/.style={
    circle,
    text width=0.8em,
    inner sep=0.1em,
    text depth=0.08em,
    draw,
    font=\scriptsize
  },
  participant/.style={
    draw=black,
    rounded corners,
    semithick,
    font=\footnotesize,
    text height=0.2cm,
    text centered,
    anchor=base
  },
  specification/.style={
    font=\small,
    semithick,
    draw=black
  },
  synthesiser/.style={
    draw=black,
    dotted,
    semithick,
    inner sep=2pt,
    font=\footnotesize
  },
  monitor/.style={
    draw=black,
    fill=white,
    semithick,
    minimum width=0.5cm,
    minimum height=0.5cm,
    font=\footnotesize,
    text height=0.2cm,
    drop shadow=ashadow
  },
  point/.style={
    minimum size=0pt, 
    inner sep=0pt
  },
  ltspoint/.style={
    circle,
    fill=black,
    draw=white,
    line width=0.5mm,
    text width=0.25em,
    inner sep=0,
    font=\footnotesize
  },
  ashadow/.style={
    opacity=.3, 
    shadow xshift=0.6mm, 
    shadow yshift=-0.6mm 
  },
  oppashadow/.style={
    opacity=.3, 
    shadow xshift=-0.6mm, 
    shadow yshift=-0.6mm 
  },
  process/.style={
    draw=black,
    fill=white,
    semithick,
    minimum width=0.5cm,
    minimum height=0.5cm,
    font=\footnotesize,
    text height=0.2cm,
    drop shadow=ashadow
  },
  event/.style={
    draw=black,
    semithick,
    minimum width=0.4cm,
    minimum height=0.4cm,
    font=\scriptsize
  },
  trap/.style={
    trapezium, 
    trapezium angle=67.5, 
    draw,
    inner ysep=5pt, 
    outer sep=0pt,
    minimum height=1.81mm, 
    minimum width=0pt
  },
  static-process/.style={
    draw=black,
    fill=white,
    semithick,
    minimum width=0.5cm,
    minimum height=0.5cm,
    font=\footnotesize,
    text height=0.2cm
  },
  pill/.style={
    draw,
    text=black,
    fill=white,
    minimum width=0.84em, 
    minimum height=0.8em,
    inner sep=0.1em,
    outer sep=0,
    font=\scriptsize\sffamily,
    align=center,
    rounded corners=0.4em  
  }
}

\pgfdeclaredecoration{penciline}{initial}{
    \state{initial}[width=+\pgfdecoratedinputsegmentremainingdistance,auto corner on length=1mm,]{
        \pgfpathcurveto%
        {%
            \pgfqpoint{\pgfdecoratedinputsegmentremainingdistance}
                            {\pgfdecorationsegmentamplitude}
        }
        {%
        \pgfmathrand
        \pgfpointadd{\pgfqpoint{\pgfdecoratedinputsegmentremainingdistance}{0pt}}
                        {\pgfqpoint{-\pgfdecorationsegmentaspect\pgfdecoratedinputsegmentremainingdistance}%
                                        {\pgfmathresult\pgfdecorationsegmentamplitude}
                        }
        }
        {%
        \pgfpointadd{\pgfpointdecoratedinputsegmentlast}{\pgfpoint{1pt}{1pt}}
        }
    }
    \state{final}{}
} 
\newtoggle{techreport}
\toggletrue{techreport}

\usepackage[index,nomargin,%
  status=final%
]{fixme} %
\fxusetheme{colorsig}%
\fxuselayouts{pdfcmargin}%
\FXRegisterAuthor{fxAF}{anfxAF}{}%
\FXRegisterAuthor{fxAS}{anfxAS}{}%
\FXRegisterAuthor{fxCBB}{anfxCBB}{}%

\title{On the Monitorability of Session Types, in Theory and Practice (Extended Version)} %

\author{Christian Bartolo Burl\`{o}}{Gran Sasso Science Institute, L'Aquila, Italy}{christian.bartolo@gssi.it}{http://orcid.org/0000-0002-0016-086X}{}

\author{Adrian Francalanza}{Department of Computer Science, University of Malta, Msida, Malta}{adrian.francalanza@um.edu.mt}{https://orcid.org/0000-0003-3829-7391}{}

\author{Alceste Scalas}{DTU Compute, Technical University of Denmark, Kongens Lyngby, Denmark}{alcsc@dtu.dk}{https://orcid.org/0000-0002-1153-6164}{}

\authorrunning{Christian Bartolo Burl\`{o}, Adrian Francalanza, Alceste Scalas} %

\Copyright{Christian Bartolo Burl\`{o}, Adrian Francalanza, Alceste Scalas} %



\relatedversion{} %
\iftoggle{techreport}{%
  \relatedversiondetails{ECOOP 2021}{https://doi.org/10.4230/LIPIcs.ECOOP.2021.22}%
}{%
  \relatedversiondetails{Full Version}{https://arxiv.org/abs/2105.06291}
}%

\acknowledgements{%
  This work has been partly supported by:
  the project MoVeMnt (No:\,217987-051) under the Icelandic Research Fund;
  the BehAPI project funded by the EU H2020 RISE under the Marie Sk{\l}odowska-Curie action (No:\,778233);
  the EU Horizon 2020 project 830929 \textit{CyberSec4Europe};
  the Danish Industriens Fonds Cyberprogram 2020-0489 \textit{Security-by-Design in Digital Denmark}.
}%

\nolinenumbers %

\iftoggle{techreport}{%
  \hideLIPIcs  %
}{}

\EventEditors{Manu Sridharan and Anders M{\o}ller}
\EventNoEds{2}
\EventLongTitle{35th European Conference on Object-Oriented Programming (ECOOP 2021)}
\EventShortTitle{ECOOP 2021}
\EventAcronym{ECOOP}
\EventYear{2021}
\EventDate{July 12--16, 2021}
\EventLocation{Aarhus, Denmark (Virtual Conference)}
\EventLogo{}
\SeriesVolume{194}
\ArticleNo{22}

\Crefname{section}{\S}{Sections}
\Crefname{figure}{Fig.\@}{Figures}
\Crefname{definition}{Def.\@}{Definitions}

\makeatletter%
\begin{document}

\maketitle

\begin{abstract}
  Software components are expected to communicate according to predetermined protocols and APIs.
  Numerous methods have been proposed to check the correctness of communicating systems against such protocols/APIs. 
  \emph{Session types} are one such method, used both for static type-checking as well as for run-time monitoring.
  This work takes a fresh look at the run-time verification of communicating systems using session types, in theory and in practice.
  On the theoretical side, we develop a formal model of session-monitored processes.
  We then use this model to formulate and prove new results on the \emph{monitorability} of session types, defined
  in terms of \emph{soundness} (\ie whether monitors only flag ill-typed processes)
  and \emph{completeness} (\ie whether all ill-typed processes can be flagged by a monitor).
  On the practical side, we show that our monitoring theory is indeed \emph{realisable}:
  we instantiate our formal model as a Scala toolkit (called \stmonitor) for the automatic generation of session monitors.
  These executable monitors can be used as proxies to instrument communication across black-box processes written in any programming language.
  Finally, we evaluate the viability of our approach through a series of benchmarks.
\end{abstract}

\section{Introduction}\label{s:intro}

\begin{wrapfigure}
  {R}{5.5cm}
  \centering
  \begin{tikzpicture}
    \node (client) [participant]{\textsf{client}};
    \node (server) [participant, right=2cm of client]{\textsf{server}};

    \node (client-end) [point, below=4.25cm of client]{};
    \node (server-end) [point, below=4.25cm of server]{};

    \draw[-open square, semithick] (client) edge (client-end);
    \draw[-open square, semithick] (server) edge (server-end);

    \node (client-login) [point, below=0.5cm of client]{};
    \node (client-login-left) [fill=black, draw=black, minimum size=0.2pt, inner sep=0pt, left=0.3cm of client-login]{};
    \draw[dashed,->] (client-login-left) edge (client-login.east);

    \node (server-login) [point, below=0.5cm of server]{};
    
    \draw[semithick,->] (client-login) edge node[above,yshift=-2pt]{\scriptsize\texttt{Auth}(\textsf{Str},\textsf{Str})}  (server-login);

    \node (client-succ) [point, below=0.7cm of client-login]{};
    \node (server-succ) [point, below=0.7cm of server-login]{};
    \node (server-succ-right) [fill=black, draw=black, minimum size=0.2pt, inner sep=0pt, right=0.3cm of server-succ]{};
    \draw[dashed,-] (server-succ) edge (server-succ-right);

    \draw[semithick,->] (server-succ) edge node[above,yshift=-2pt]{\scriptsize\texttt{Succ}(\textsf{Str})} (client-succ);

    \node (client-get) [point, below=0.7cm of client-succ]{};
    \node (server-get) [point, below=0.7cm of server-succ]{};
    \draw[semithick,->] (client-get) edge node[above,yshift=-2pt]{\scriptsize\texttt{Get}(\textsf{Str},\textsf{Str})} node[below,align=center, yshift=-3pt]{$\mathsmaller{\vdots}$} (server-get);

    \node (client-rvk) [point, below=1cm of client-get]{};
    \node (server-rvk) [point, below=1cm of server-get]{};
    \node (client-rvk-left) [fill=black, draw=black, minimum size=0.2pt, inner sep=0pt, left=0.3cm of client-rvk]{};
    \draw[dashed,-] (client-rvk-left) edge (client-rvk);
    \draw[semithick,->] (client-rvk) edge node[above,yshift=-2pt]{\scriptsize\texttt{Rvk}(\textsf{Str})} (server-rvk);

    \node (client-fail) [point, below=0.9cm of client-rvk]{};
    \node (server-fail) [point, below=0.9cm of server-rvk]{};
    \node (server-fail-right) [fill=black, draw=black, minimum size=0.2pt, inner sep=0pt, right=0.3cm of server-fail]{};
    \draw[dashed,-] (server-fail) edge (server-fail-right);

    \draw[semithick,->] (server-fail) edge node[above,yshift=-2pt] {\scriptsize\texttt{Fail}(\textsf{Int})} (client-fail);

    \node (client-fail-left) [fill=black, draw=black, minimum size=0.2pt, inner sep=0pt, left=0.3cm of client-fail]{};
    \draw[dashed,-] (client-fail-left) edge (client-fail);

    \draw[dashed,-] (client-login-left) edge node[left,above,rotate=90,yshift=-2pt] {\scriptsize\textit{recurse}} (client-fail-left);

    \draw[dashed,-] (server-succ-right) edge node[right,above,rotate=-90,yshift=-2pt] {\scriptsize\textit{choice}} (server-fail-right);

    \node (client-1) [fill=white, minimum size=1pt, inner sep=3pt, below=0.3cm of client-rvk]{};
    \node (server-1) [fill=white, minimum size=1pt, inner sep=3pt, below=0.3cm of server-rvk]{};
  \end{tikzpicture}
   \caption[Authentication Protocol.]{Authentication Protocol.}
  \label{fig:uml-sd}
  \vspace{-3mm}
\end{wrapfigure}
Communication protocols and Application Programming Interfaces (APIs)~\cite{BroFT:14:API} 
govern the interactions between concurrent and distributed software components by exposing the functionality of a component for others to use. %
Although the order of messages exchanged and methods invoked is crucial for correct API usage,
this information is either outright omitted, or stated informally via textual descriptions~\cite{SonT:13:MC-API,Sho:ISSTA:07:AutomataAPI}.
At best, protocols and temporal API usage are described semi-formally as message sequence charts \cite{DBLP:journals/entcs/Peled02}. 
This state of affairs is conducive to conflicting interactions, which may 
manifest themselves as run-time errors, deadlocks
and livelocks.
\emph{Behavioural types}~\cite{DBLP:journals/ftpl/AnconaBB0CDGGGH16} provide a methodology to address these shortcomings, by elevating protocols and flat API descriptions to \emph{formal behavioural specifications}
with explicit \emph{sequences} and \emph{choices} of operations.
A prevalent form of behavioural types are \emph{session types}~\cite{DBLP:conf/concur/Honda93,DBLP:conf/esop/HondaVK98} which can ensure correct interactions %
that are free from communication errors, %
deadlocks and livelocks.

\begin{example}\label{eg:intro-example} 
  Consider a server that exposes the API calls \auth (authenticate), \get and \revoke (revoke). 
  The intended use of this API is to invoke \auth followed with \get and finally \revoke, as depicted in \Cref{fig:uml-sd}. 
  If authentication is successful, \auth returns a token that can be used for exclusive access to a resource with the service \get. 
  After its use, the token should be revoked with the service \revoke to allow other parties to access the resource. 
  For security reasons, the server is expected to only reply \get requests after it services an \auth request. 
  However, if the order of invocation is not respected, a client may send a \get request before an \auth request. The resulting components' interaction will be incorrect, causing an error or deadlock.
  Even worse, the server may accept the \get request and let an unauthenticated client access sensitive information.
  The protocol from the viewpoint of the client can be described as the session type:

  \smallskip\centerline{\(
  S \quad\equal\quad !\auth\,.\,\&
  \left\{
    \begin{array}{ll}
      ?\succs\,.\,!\get
      \dots
      !\revoke\,.\,S,
      \quad
      ?\texttt{Fail}\,.\,S\end{array}
  \right\}
  \)}\smallskip

Type $S$ states that the client is expected to first invoke ($!$) the service \auth and then branch ($\&$) according to the response received ($?$). 
If it receives {\succs}ess, the client can invoke \get and eventually \revoke before restarting the protocol ($S$). 
Otherwise, if  it receives \fail, the client may start following the type $S$ from the beginning and retry authentication.  
\exqed 
\end{example}

\subparagraph{Run-time monitoring of session types: promise and challenges.}
In behavioural type frameworks (including session types), the conformance between the component under scrutiny and a desired protocol is commonly checked \emph{statically}, via a type system that is tailored for the language used to develop the component. 
This avoids runtime overhead and allows for early error detection. 
However, there are cases where a (full) static analysis is not possible.  
For instance, within a distributed or collaborative system, not all system components are necessarily accessible %
for static analysis
(\eg due to 
source obfuscation). 
Components may also be implemented using different programming languages, making it infeasible to develop bespoke type-checkers for every programming language used in development. 
In these cases, post-deployment techniques such as Runtime Verification (RV)~\cite{DBLP:conf/rv/FrancalanzaAAAC17,DBLP:series/lncs/BartocciFFR18} can be used where protocol conformance 
is carried out \emph{dynamically} via \emph{monitors} \cite{DBLP:conf/tgc/ChenBDHY11,  DBLP:conf/rv/NeykovaYH13, DBLP:journals/fmsd/DemangeonHHNY15, DBLP:conf/popl/JiaGP16, DBLP:conf/cc/NeykovaY17, DBLP:journals/tcs/BocchiCDHY17, DBLP:conf/esop/GommerstadtJP18}.
Runtime monitoring of behavioural 
types comes with a set of challenges.
\begin{description}
  \item[The realisability of effective monitoring:] 
  Restrictions such as inaccessible code and license agreements 
  (regulating code modifications),  
  may restrict the ways in which software components can be instrumented,
  thus hindering a monitor's capabilities for observation and intervention.
  Moreover, the runtime overhead induced by monitors should be kept within acceptable levels.
  \item[Monitor Correctness:]
  Intuitively, a ``correct'' monitor for a session type $S$ should carry out detections that correspond to the protocol represented by $S$.
  The recent results on \emph{monitorability} help us unpack this intuition of ``correctness'' in terms of \emph{soundness} and \emph{completeness}:
  the monitor should not unnecessarily flag well-behaving code (\emph{detection soundness}~\cite{DBLP:journals/fmsd/FrancalanzaAI17,DBLP:journals/pacmpl/AcetoAFIL19}), while providing guarantees for recognising misbehaving components (\emph{detection completeness}~\cite{DBLP:conf/sefm/AcetoAFIL19,AcetoAFIL21:SOSYM}). 
\end{description} 
The aforementioned challenges are not independent of one another, and an adequate solution often needs to take both aspects into consideration.
On the one hand, monitor correctness may require computations that increase runtime overheads;
on the other hand, there are inherent limits to what can be detected at runtime
(\ie the monitorability problem~\cite{DBLP:series/lncs/BartocciFFR18}) ---
and moreover, practical implementation concerns may restrict monitoring capabilities even further (\eg 
due to the need for low overheads).
To our knowledge, the above aspects have not been fully investigated together for session types, in \emph{one unified study}:%
\begin{itemize}
  \item there is no systematic examination for the  \emph{monitorability} of session types, determining the limits of runtime monitoring when verifying session-type conformance;
  \item no previous work tackles the design of a session monitoring system that is practically \emph{realisable}, while also backed by formal detection soundness \emph{and} completeness guarantees.
\end{itemize}

\subparagraph{Contributions.}
We present the first formal analysis of the \emph{monitorability of session types},
and use it to guide the design and implementation of a practical framework
(written in Scala) for the run-time monitoring of concurrent and distributed applications.
We focus on communication protocols that can be formalised as \emph{(binary) session types}~\cite{DBLP:conf/concur/Honda93,DBLP:conf/esop/HondaVK98}
with two interacting parties (\eg a client and a server). %
Crucially, we tackle scenarios where at least one of the parties is a ``black-box'' process that may not be statically verified.
After formalising a streamlined process calculus with session types (\Cref{sec:calculus-types}),
we present our contributions:
\begin{enumerate}
  \item 
  We develop a formal model detailing how processes can be instrumented with monitors, to observe their interactions and flag violations on the offending party (\Cref{s:designing-hybrid-methodology}).
  We then design an automated synthesis procedure from session types to monitors (in this operational model) to study the monitorability of session types (\Cref{s:monitor-synthesis});
  \item 
  We carry out the first study on the \emph{monitorability} of session types, formally linking their static and run-time verification (\Cref{s:formal-results}).
  We prove that our synthesised monitors are
  \emph{detection-sound}, \ie components flagged by a monitor for session type $S$ are indeed ill-typed for $S$
  (\Cref{thrm:sound}).
  We also prove a \emph{weak detection-completeness} result (\Cref{lem:partial-monitor-completeness}) showing to what degree can our synthesised monitors detect ill-typed components.
  Importantly, we show that these limits are not specific to our synthesis procedure by proving an impossibility result: under our ``black-box'' monitoring model, session monitoring cannot be both sound and complete (\Cref{thm:impossibility}). 
  The latter results are new to the area of behavioural types;
  \item
  We show the \emph{realisability} of our model, by implementing a toolkit (called \stmonitor)
  that synthesises session monitors as executable Scala programs (\Cref{s:realisabilty}).
  \iftoggle{techreport}{%
    We provide \stmonitor as companion artifact of this work.
  }{
    We provide \stmonitor as companion artifact of this paper.
  }
  We also provide evaluation benchmarks showing that our generated Scala monitors induce limited overheads, hence their usability in practice appears promising (\Cref{s:implementation}).
\end{enumerate} 

Proofs and additional details are available
\iftoggle{techreport}{%
   in the appendices.%
}{%
  in the extended version of this paper \cite{burlo2021monitorability-techrep}.%
}

\section{Process Calculus and Session Types}
\label{s:background}
\label{sec:calculus-types}

This section 
introduces the formalism at the basis of our work:
a streamlined process calculus (\Cref{s:process-calculus})
with standard session types (\Cref{s:binary-session-types})
and typing system (\Cref{s:session-type-system}).

\subsection{Process Calculus}\label{s:process-calculus}

\subparagraph{Syntax.}
We adopt a streamlined process calculus that models
a sequential process interacting on a single communication channel,
similar to
\cite{GhilezanPPSY21,GhilezanJPSY19,SeveriD19}. 
Our process calculus %
is defined in Figure \ref{fig:process-calc-syntax-semantics}. 
The syntax assumes separate denumerable sets of \textbf{values} $v, u, w \in \Val$ (including tuples), value \textbf{variables} $x, y, z \in \Var$ and \textbf{process variables} $X, Y \in \PVar$.
We use $a,b$ to range over the set $\textsc{Val}\,\cup\,\textsc{Var}$. %
The syntax also assumes a set of %
\textbf{predicates} $A$ (used in conditionals). 
A process may communicate by sending or receiving \textbf{messages} of the form $\pmsgk{l}{v}$,
where $\texttt{l}$ is a \textbf{label}, and $v$ is the \textbf{payload value}. %
To this end, a process may perform \textbf{outputs} $\psndk{l}{a}.P$
(\ie send message $\pmsgk{l}{v}$ and continue as $P$),
or \textbf{inputs} $\prcv{i}{I}$
(\ie receive a message with label $\texttt{l}_i$ for any $i \in I$, and continue as $P_i$, with $x_i$ replaced by the message payload).
Loops are supported by the \textbf{recursion} construct $\pmu{X}.P$,
and the \textbf{process variable} $X$.
The process $\pnil$ represents a \textbf{terminated} process.
The calculus also includes a standard \textbf{conditional} construct $\pif{A}{P}{Q}$.
We assume that all recursive processes are \textbf{guarded}, \ie process variables can only occur under an input or output prefix. 
The calculus has two \textbf{binders}:
the input construct $\prcv{i}{I}$ binds the free occurrences of the (value) variables $x_i$ %
in the continuation process $P_i$, whereas the recursion construct $\mu_X.P$ binds the process variable $X$ in the continuation process $P$. 
\begin{figure}
  \small
  \textbf{Syntax}\hspace{1cm}
  {\(  
    \begin{array}{r@{\hskip 2mm}r@{\hskip 2mm}c@{\hskip 2mm}l}
      \text{Predicates} & A & \Coloneqq & \textsf{tt} \;\vert\; \textsf{ff} \;\vert\; v_{1} \equal\equal v_{2} \;\vert\; v_{1} >\equal v_{2} \;\vert\; A_{1}{\normalfont\textsf{ \&\& }}A_{2} \;\vert\; !A \;\vert\; \ldots\\[1mm]
      \text{Processes} & P,Q & \Coloneqq & \psndk{l}{a}.P \;\vert\;\prcv{i}{I} \;\vert\; \pmu{X}.P \;\vert\; X \;\vert\; \pif{A}{P}{Q} \;\vert\; \pnil\\
    \end{array}
  \)}\bigskip\\
  \textbf{Semantics}
  \hspace{5mm}
  {
    \inference[\text{[\textsc{pRec}]}]{}{\mu_{X}.P \xrightarrow{\tau} P[\nicefrac{\mu_{X}.P}{X}]}
    \hspace{1.5cm}
    \inference[\text{[\textsc{pSnd}]}]{}{\psnd{}.P \xrightarrow{\sndact{l}{v}} P}
  }\\[2mm]
  \centerline{
    \inference[\text{[\textsc{pRcv}]}]{\phantom{X}}{\prcv{i}{I} \xrightarrow{\triangleright \pmsgNoVIdx{j}} P_{j}[\nicefrac{v}{x_j}]} $j\in I$
  }\\[2mm]
  \centerline{
    \inference[\text{[\textsc{pTru}]}]{A \Downarrow \textsf{tt}}{\pif{A}{P}{Q}\xrightarrow{\tau} P}
    \hspace{1cm}
    \inference[\text{[\textsc{pFls}]}]{A \Downarrow \textsf{ff}}{\pif{A}{P}{Q}\xrightarrow{\tau} Q}
  }
  \caption[Process Calculus Syntax and Semantics]{Process Calculus Syntax and Semantics}\label{fig:process-calc-syntax-semantics}
\end{figure}
 
\subparagraph{Semantics.} The dynamic behaviour of a process is described by the transition rules in \Cref{fig:process-calc-syntax-semantics}. %
The rules take the form $P \xrightarrow{\mu} P'$, where the transition \textbf{action} $\mu$
can be either an \textbf{output action} $\sndact{l}{v}$, an \textbf{input action} $\rcvactk{l}{v}$, or a \textbf{silent action} $\tau$.
Rule \ltsrule{pRec} allows the recursive process $\pmu{X}.P$ to unfold. %
Rules \ltsrule{pSnd} and \ltsrule{pRcv} enable communication: %
\begin{itemize}
\item by \ltsrule{pSnd}, process $\psndk{l}{v}.P$ sends a message by firing action $\sndact{l}{v}$ and continuing as $P$;
\item by \ltsrule{pRcv}, process $\prcv{i}{I}$ can receive a message $\pmsgNoVIdx{j}$ ($j\in I$)
by firing action $\triangleright\pmsgNoVIdx{j}$ and continuing as $P_j$, with the payload value $v$ replacing the variable $x_j$. 
\end{itemize}
The remaining two rules \ltsrule{pTru} and \ltsrule{pFls} define the silent transitions when the predicate in the process $\pif{A}{P}{Q}$ evaluates to true ($A\Downarrow\true$) or false ($A\Downarrow\false$), respectively. 
For brevity, we often omit the trailing \textbf{0} and write $\prcvk{l}{v}.P$ for singleton input choices. 

\begin{example}[Process syntax and semantics]\label{eg:process-calculus}
  \begin{subequations}
    Recall the protocol depicted in \Cref{fig:uml-sd}.
    A corresponding client process %
    for this protocol is defined as $P_\textit{auth}$ below. 
    
    \smallskip\centerline{\(
      \small P_\textit{auth}\, \equal\, \pmu{X}.\triangleleft \texttt{Auth}(\pstring{Bob},\pstring{pwd}).P_\textit{res}
      \qquad
      \text{\small{where} }\; P_\textit{res}\, \equal\, \triangleright \big\{ \texttt{Succ}(\textit{tok}).P_\textit{succ}\,,\;\texttt{Fail}(\textit{code}).P_\textit{fail}\big\}
    \)}\smallskip

    \noindent From the rules in Figure \ref{fig:process-calc-syntax-semantics},
    the process $P_\textit{auth}$ executes as follows:
    \begin{align}
      \nonumber
      P_\textit{auth} & \xrightarrow{\tau} \big{(} \triangleleft \texttt{Auth}(\pstring{Bob},\pstring{pwd}).P_\textit{res}\big{)} [\nicefrac{P_\textit{auth}}{X}] && \quad\text{using \ltsrule{pRec}}\\
      \nonumber
      &\quad\xrightarrow{\triangleleft \texttt{Auth}(\pstring{Bob},\pstring{pwd})}  \triangleright \left\{ \begin{array}{ll} \texttt{Succ}(\textit{tok}).P_\textit{succ}[\nicefrac{P_\textit{auth}}{X}]\,,\\\texttt{Fail}(\textit{code}).P_\textit{fail}[\nicefrac{P_\textit{auth}}{X}] \end{array}\right\} && \quad\text{using \ltsrule{pSnd}}\\
      \intertext{The process performs a silent action $\tau$ %
      to unfold its recursion, and then sends a message %
      with label \texttt{Auth} and tuple \pstring{Bob},\pstring{pwd} as payload.
      If the authentication is successful, the process receives the message \texttt{Succ} including a token \textit{tok} and proceeds according to $P_\textit{succ}$ (omitted):}
      \nonumber
      &\quad\quad \xrightarrow{\triangleright \texttt{Succ}(321)} P_\textit{succ}[\nicefrac{P_\textit{auth}}{X}][\nicefrac{321}{\textit{tok}}] && \quad \text{using \ltsrule{pRcv}}
    \end{align}
    Otherwise, if the authentication is unsuccessful, the process receives the message \texttt{Fail} including an error \textit{code} from the server and proceeds according to $P_\textit{fail}$. \exqed
  \end{subequations}
\end{example} 
\subsection{Binary Session Types}\label{s:binary-session-types}

Session types 
describe the structure of interaction among processes. 
They enable the verification of 
communicating systems against a stipulated communication protocol. 
Figure \ref{fig:session-types-noa} formalises binary session types.
\begin{figure}
  \small
  \textbf{Syntax}\hspace{5mm}
  {\(
    \begin{array}{r@{\hskip 2mm}r@{\hskip 2mm}c@{\hskip 2mm}l}
      \text{Base types}& \textsf{B} & \Coloneqq & \normalfont{\textsf{Int}} \;\vert\; \normalfont{\textsf{Str}} \;\vert\; \normalfont{\textsf{Bool}} \;\vert\; \ldots \;\vert\; (\textsf{B},\textsf{B})
      \\[1mm]
      \text{Session types}& R,S & \Coloneqq & \underbrace{\stselNA{i}{I} \;\;\vert\;\; \stbraNA{i}{I}}_{\text{with $I \!\neq\! \emptyset$ and $\texttt{l}_i$ pairwise distinct}} \;\;\vert\;\; \strec{X}.S \;\;\vert\;\; X \;\;\vert\;\; \textsf{end} 
    \end{array}
    \)
  }

  \medskip

  \textbf{Dual types}\hspace{5mm}
  {\(
    \begin{array}{r@{\;\;}c@{\;\;}l@{\qquad}c}
    \overline{\stbraNA{i}{I}} &=& \oplus \big\{!\texttt{l}_{i}(\textsf{B}_{i}).\overline{S_{i}}\big\}_{i\in I} &
    \overline{\textsf{end}} \,=\, \textsf{end} \qquad {\overline{X}} \,=\, X\\[2mm]
    \overline{\stselNA{i}{I}} &=& \& \big\{?\texttt{l}_{i}(\textsf{B}_{i}).\overline{S_{i}}\big\}_{i\in I}&
    \overline{\strec{X}.S} \;=\; \strec{X}.\overline{S}
  \end{array}
  \)}
  \caption[Session Types Syntax, and Definition of Dual Types]{Session Types Syntax, and Definition of Dual Types.}\label{fig:session-types-noa}
\end{figure} %
We assume a set of standard base types $B$ which includes tuples. 
The \textbf{selection type} (or \textbf{internal choice}) $\stselNA{i}{I}$ requires a component to 
send a message $\pmsgNoVIdx{i}$ where the value $v$ has base type $\textsf{B}_i$, for some $i\in I$. 
The \textbf{branching type} (or \textbf{external choice}) $\stbraNA{i}{I}$ requires a component to receive a message of the form $\pmsgNoVIdx{i}$, where the value $v$ (\ie the message payload) is of the corresponding base type $\textsf{B}_i$ for any $i \in I$. 
The \textbf{recursive} session type $\strec{X}.S$ binds the recursion variable $X$ in $S$ (we assume guarded recursion), while \textsf{end} describes a \textbf{terminated} session. 
For brevity, we often omit $\oplus$ and $\&$ for singleton choices, as well as trailing \textsf{end}s.

A process implementing a session type $S$ can correctly interact with a process implementing the \textbf{dual type} of $S$, denoted as $\overline{S}$ (defined in \Cref{fig:session-types-noa}). 
Intuitively, the dual type of a selection 
is a branching 
type with the same choices.
Hence, every possible output from one component matches an input by the other component, and \emph{vice versa}.
Duality guarantees that the interaction between typed components is \emph{safe} (\ie only expected messages are communicated) and \emph{deadlock-free} (\ie the session 
terminates only if both components reach their end).

\begin{example}\label{eg:session-types-noa}\begin{subequations}
  The session type $S_\textit{auth}$ below formalises the first part of the protocol that the \client in \Cref{fig:uml-sd} is expected to follow (\ie the type $S$ in \Cref{eg:intro-example}). 
  \begin{align}
    \nonumber
    S_\textit{auth} \;\;\equal\;\; &\strec{Y}.!\texttt{Auth}(\textsf{Str},\textsf{Str}).
    \&\big\{?\texttt{Succ}(\textsf{Str}).S_\textit{succ}\,,\;?\texttt{Fail}(\textsf{Int}).Y \big\}
  \end{align}
  The \server 
  should follow  
  \begin{math}
    \overline{S_\textit{auth}}\ \equal\ \strec{Y}.?\texttt{Auth}(\textsf{Str},\textsf{Str}).
    \oplus\big\{!\texttt{Succ}(\textsf{Str}).\overline{S_\textit{succ}}\,,\;!\texttt{Fail}(\textsf{Int}).Y \big\}
  \end{math}, its dual.
  According to $S_\textit{auth}$, the \client 
  initiates interaction by sending a message with label $\texttt{Auth}$,
  carrying a tuple of strings (username and password) as payload.
  The \server should then
  reply with either \labk{Succ}ess (carrying a string), or \labk{Fail}ure (with an integer error code). 
  In case of \labk{Succ}ess, the \client 
  continues along 
  $S_\textit{succ}$.
  In case of \labk{Fail}ure, the session loops. \exqed
\end{subequations}
\end{example} 

\subsection{Session Typing System}\label{s:session-type-system}

Our session typing system (in \Cref{fig:session-typing-rules-noa}) is standard.
It uses two \textbf{typing environments} $\Theta$ and $\Gamma$,
where $\Theta$ is a partial mapping from process variables to session types,
while $\Gamma$ is a partial mapping from value variables to base types.
We represent them syntactically as:

{
\centerline{
  \(
  \Theta \;\Coloneqq\; \emptyset\ \big| \ \Theta,X:S \qquad\qquad \Gamma \;\Coloneqq\; \emptyset\ \big| \ \Gamma,x:\basetype
  \)
}
}
\noindent
The type system is \emph{equi-recursive} \cite{Pierce02}: when comparing two types, we consider a recursive type $\strec{X}.S$ to be equivalent to its unfolding $S[\nicefrac{\strec{X}.S}{X}]$ %
(\ie interchangeable in all contexts). 
\begin{figure}
  \small
  \textbf{Identifier Typing}\hspace{2cm}
    \inference[\text{[\textsc{tVar}]}]{\Gamma(x)\equal \basetype}{\Gamma \rightvdash x : \basetype}
    \hspace{1cm}
    \inference[\text{[\textsc{tVal}]}]{v \in \basetype}{\Gamma \rightvdash v : \basetype}
  \bigskip\\
  \textbf{Process Typing}
  \bigskip\\
  \centerline{
    \inference[\text{[\textsc{tBra}]}]{\forall i \in I \qquad \Theta\cdot\Gamma, x_i:\textsf{B}_i \rightvdash P_i:S_i}{\Theta\cdot\Gamma \rightvdash \prcv{i}{I \cup J}: \stbraNA{i}{I}}
    \hspace{1cm}
    \inference[\text{[\textsc{tRec}]}]{\Theta, X:S\cdot\Gamma\rightvdash P:S}{\Theta\cdot\Gamma \rightvdash \pmu{X}.P:S}
  }\bigskip\\
  \centerline{
    \inference[\text{[\textsc{tSel}]}]{\exists i \in I \qquad \texttt{l} \equal \texttt{l}_i \qquad \Gamma\rightvdash a:\textsf{B}_i \qquad \Theta\cdot\Gamma \rightvdash P:S_i}{\Theta\cdot\Gamma \rightvdash \psndk{l}{a}.P: \stselNA{i}{I}}
    \hspace{1cm}
    \inference[\text{[\textsc{tPVar}]}]{\Theta(X)\equal S}{\Theta\cdot\Gamma \rightvdash X :S}
  }\bigskip\\
  \centerline{
    \inference[\text{[\textsc{tIf}]}]{\Gamma\rightvdash A:\textsf{Bool}&\Theta\cdot\Gamma\rightvdash P:S & \Theta\cdot\Gamma\rightvdash Q:S}{\Theta\cdot\Gamma \rightvdash \pif{A}{P}{Q}: S}
    \hspace{1cm}
    \inference[\text{[\textsc{tNil}]}]{}{\Theta\cdot\Gamma \rightvdash \pnil: \textsf{end}}
  }
  \caption[Session Typing Rules]{Session Typing Rules.}\label{fig:session-typing-rules-noa}
\end{figure}
 
The typing judgement for values and variables %
is\; $\Gamma\rightvdash a:\basetype$, defined by rules \ltsrule{tVar} and 
\ltsrule{tVal}.
The process typing judgement, $\Theta \cdot \Gamma \rightvdash P:S$,
states that process $P$ communicates according to session type $S$, given the typing assumptions in $\Theta$ and $\Gamma$. 
In the \emph{branching rule} \ltsrule{tBra}, an input process has a branching type $\stbraNA{i}{I}$ if all the possible branches in the type are present as choices in the process, with matching labels. 
Hence, the process must have the form $\prcv{i}{I \cup J}$
(notice that if $J \neq \emptyset$, the process has more input branches than the type).
Moreover, for each matching branch, each continuation process $P_i$ (for $i \in I$) must be typed with the corresponding continuation type $S_i$, assuming that the received message payload $x_i$ has the expected type $\textsf{B}_i$.
The \emph{selection rule} \ltsrule{tSel} states that $\psndk{l}{a}.P$ follows a selection type of the form $\stselNA{i}{I}$ if there exists a possible choice in the type that matches the message $\lab(a)$. 
To match, the labels must be identical, and the type of the payload $a$ must be of the type $\basetypek{i}$ stated in the session type, and the continuation process $P$ must be of the continuation type $S_i$. 
The remaining rules are fairly standard.

\begin{remark}
  \label{remark:predicates}
Although we do not fix the boolean predicates $A$, 
we assume that:
\begin{enumerate}
  \item boolean predicates can be type-checked with standard rules;
  \item\label{assumption:bi}
    base types $\textsf{B}$ come with a predicate $\isValueB{}{v}$
    that returns $\true$ if $v$ is of type $\textsf{B}$, and $\false$ otherwise (akin to {\small\texttt{instanceof}} in Java.)  \exqed
\end{enumerate}
\end{remark}

\begin{example}\label{eg:type-system}\begin{subequations}
  Recall the process $P_\textit{auth}$ (\Cref{eg:process-calculus}) and the session type $S_\textit{auth}$ (\Cref{eg:session-types-noa}):
 
  \smallskip\centerline{\(
    P_\textit{auth}\ \equal\ \pmu{X}.\triangleleft \texttt{Auth}(\pstring{Bob},\pstring{pwd}).P_\textit{res}
    \qquad
    S_\textit{auth}\ \equal\ \strec{Y}.!\texttt{Auth}(\textsf{Str},\textsf{Str}).S_\textit{res}
  \)}\smallskip
  
  \noindent%
  One can show that $P_\textit{auth}$ type-checks with 
  $S_\textit{auth}$, \ie $\emptyset\cdot\emptyset \rightvdash P_\textit{auth}: S_\textit{auth}$. \exqed 
\end{subequations}

\end{example} 

\section{A Formal Model for Monitoring Sessions}
\label{s:designing-hybrid-methodology}

We now formalise an operational setup that enables us to verify the (binary) session types of \Cref{s:background} at runtime.
Our runtime analysis is conducted by \emph{uni-verdict} rejection monitors, whose purpose is to flag any session type violations detected (\ie \emph{violation monitors}~\cite{DBLP:journals/fmsd/FrancalanzaAI17,DBLP:journals/pacmpl/AcetoAFIL19}).

\subsection{Monitor and Instrumentation Design}
\label{s:monitor-design}
\label{s:monitor-instrumentation}

We now illustrate the design decisions behind our formal monitoring framework.  To this end, we use as a reference the client-server system outlined in \Cref{eg:intro-example}.
Consider, in particular, the scenario depicted in \Cref{fig:potential-designs:no-monitors}, where a \client is expected to interact with a \server following the prescribed protocol $S_\textit{auth}$;
the \server is \emph{trusted} and guaranteed to adhere to the dual type $\overline{S_\textit{auth}}$
(\eg because it has been statically typechecked against $\overline{S_\textit{auth}}$ using the type system in \Cref{s:session-type-system})
--- but we have limited control over the \client, which might be untyped,
hence its interactions are potentially unsafe.

Our setup should place no assumptions on the \client, largely treating it as a black box. 
In fact, we target scenarios where the \client source is inaccessible, possibly remote, interacting with the server via a generic channel of communication (\eg TCP sockets or HTTP addresses).
This precludes the possibility of weaving the monitor within the \client component. 
To achieve a model that can handle these requirements, we restrict ourselves to \emph{outline monitors}~\cite{DBLP:series/lncs/BartocciFFR18,AAttardFI:outlined-choreographed-tech:21}, which are decoupled from the process-under-scrutiny as concurrent units of code that can be more readily deployed over a black-box component; 
outline monitors are also easier to verify for correctness via compositional techniques~\cite{DBLP:conf/sp/AmorimDGHPST15,DBLP:conf/fossacs/Francalanza16,DBLP:conf/concur/Francalanza17,BasinDHKRST:2020:verifiedMon,DBLP:conf/coordination/FrancalanzaX20,Monitors:21}.

\begin{wrapfigure}[20]{r}{7cm}
  \vspace{-0.2cm}

  \begin{subfigure}{7cm}
    \begin{tikzpicture}[decoration=penciline]
      \node(client) [process, fill=gray!20] {\textit{client}};
      \node(server) [process, right=3cm of client] {\textit{server}};
      \draw[thin,double distance=2pt] (server) -- (client);

      \node(snake-start) [point, above right=0.1cm and 1.5cm of client] {};
      \node(snake-end) [point, below right=0.7cm and 1.5cm of client] {};

      \draw[preaction={draw, line width=1pt, white},decorate,decoration={snake, segment length=3mm, amplitude=.5mm, post length=0mm}] (snake-start) -- (snake-end);

      \node(s-typechecks) [below=0.14cm of server, decorate, draw, font=\footnotesize] {$\rightvdash \textit{server}:\overline{S_\textit{auth}}$};
      \node(c-typechecks) [below=0.14cm of client, decorate, draw, font=\footnotesize] {$\rightvdash \textit{client}:S_\textit{auth}$};
      \node(qmark-l) [left=0.008cm of c-typechecks, font=\small, rotate=10, inner sep=2pt] {\textbf{?}};
      \node(qmark-r) [right=0.008cm of c-typechecks, font=\small, rotate=-10, inner sep=2pt] {\textbf{?}};

    \end{tikzpicture}
    \caption{\centering No monitors.}
    \label{fig:potential-designs:no-monitors}
  \end{subfigure}

  \bigskip
  \smallskip
  \begin{subfigure}{7cm}
    \begin{tikzpicture}
      \node(client) [process, fill=gray!20] {\textit{client}};
      \node(server) [process, right=1.5cm of client] {\textit{server}};
      \node(monitor) [monitor, right=1.5cm of server]{$\textit{monitor}_a$};
      \draw (server) edge[{Arc Barb[reversed]}->] (monitor);

      \draw[thin,double distance=2pt] (server) -- (client);

      \draw [dotted] ($(monitor.north west)+(-0.15cm,0.15cm)$) rectangle ($(monitor.south east)+(0.15cm,-0.15cm)$);

      \draw[->, dashed, semithick, bend left=20] (server) edge node(auth-mon)[above, xshift=4pt, yshift=1.5pt, font=\footnotesize]{$\triangleright\texttt{Auth}(\texttt{``Bob''},\texttt{``pwd''})$} node[pill, solid] {2} (monitor);
      \draw[->, dashed, semithick, bend left=20] (client) edge node(auth)[above, xshift=-4pt, yshift=1.5pt, font=\footnotesize]{$\texttt{Auth}(\texttt{``Bob''},\texttt{``pwd''})$} node[pill, solid] {1} (server);

      \draw[->, dashed, semithick, bend right=20] (server) edge node(succ-mon)[below, yshift=-1.5pt, font=\footnotesize]{$\triangleleft\texttt{Fail}(\texttt{1})$} node[pill, solid] {4} (monitor);
      \draw[->, dashed, semithick, bend left=20] (server) edge node(succ)[below, yshift=-1.5pt, font=\footnotesize]{$\texttt{Fail}(\texttt{1})$} node[pill, solid] {3} (client);

    \end{tikzpicture}
    \caption{\centering Server side instrumentation.}
    \label{fig:potential-designs:a}
  \end{subfigure}
  
  \bigskip
  \smallskip
  \begin{subfigure}{7cm}
    \begin{tikzpicture}
      \node(client) [process, fill=gray!20] {\textit{client}};
      \node(monitor) [monitor, right=1.5cm of client]{$\textit{monitor}_b$};
      \node(server) [process, right=1.5cm of monitor] {\textit{server}};

      \draw[->, dashed, semithick, bend left=20] (monitor) edge node(auth-mon)[above, font=\footnotesize, yshift=1.5pt]{$\texttt{Auth}(\texttt{``Bob''},\texttt{``pwd''})$} node[pill, solid] {2} (server);
      \draw[->, dashed, semithick, bend left=20] (client) edge node(auth)[above, font=\footnotesize, yshift=1.5pt]{$\texttt{Auth}(\texttt{``Bob''},\texttt{``pwd''})$} node[pill, solid] {1} (monitor);

      \draw[->, dashed, semithick, bend left=20] (server) edge node(succ-mon)[below, font=\footnotesize, yshift=-1.5pt]{$\texttt{Fail}(\texttt{1})$} node[pill, solid] {3} (monitor);
      \draw[->, dashed, semithick, bend left=20] (monitor) edge node(succ)[below, font=\footnotesize, yshift=-1.5pt]{$\texttt{Fail}(\texttt{1})$} node[pill, solid] {4} (client);

      \draw[thin,double distance=2pt] (client) -- (monitor);

      \draw[thin,double distance=2pt] (monitor) -- (server);

      \draw [dotted] ($(monitor.north west)+(-0.15cm,0.15cm)$) rectangle ($(monitor.south east)+(0.15cm,-0.15cm)$);

    \end{tikzpicture}
    \caption{\centering Channel instrumentation.}
    \label{fig:potential-designs:b}
  \end{subfigure}
  \caption{Design choices for instrumentation.}
  \label{fig:potential-designs} \end{wrapfigure}
Our model focusses on the communication occurring on the channel between the \client and the \server~--- and we assume such communication to be \emph{synchronous} and \emph{reliable}.
Outline monitors can typically only analyse the externally \emph{observable} actions of a monitored component.
In our case, monitored processes follow the semantics of \Cref{fig:process-calc-syntax-semantics}, hence the only observable actions are send ($\sndact{l}{v}$) and receive ($\rcvactk{l}{v}$); $\tau$-moves are unobservable. 

We consider two potential instrumentation setup designs for an outline approach. 
In the setup in \Cref{fig:potential-designs:a}, the \server is instrumented with a \emph{sequence-recogniser monitor}~\cite{Schneider:2000,Ligatti05} (\monA). 
The \server is required to notify \monA about every send and receive action it performs---this can be achieved 
via listeners added through mechanisms such as class-loaders, agents and VM-level tracers. 
For \monA, every receive action the \server performs indicates a send action by the \client and vice-versa (\ie every send indicates a receive). 
  In \Cref{fig:potential-designs:a}, 
  The \client sends the message $\texttt{Auth}(\texttt{``Bob''}, \texttt{``pwd''})$ to the \server \pill{1}. 
  Once received, \pill{2} the \server notifies \monA with the message contents 
  and the direction of the message ($\triangleright$). 
  For \monA this indicates that the \client sent the particular message. 
  After the \server replies with the message $\texttt{Fail}(\texttt{1})$ \pill{3}, it notifies \monA with the message contents and the direction ($\triangleleft$) \pill{4}, indicating that the \client received the message. 

In the alternative setup depicted in \Cref{fig:potential-designs:b}, the monitor (\monB) is instrumented on the communication channel and acts as a \emph{proxy} (or a \emph{partial-identity monitor}~\cite{DBLP:conf/esop/GommerstadtJP18}) between the two components. 
Any communicated messages must pass through \monB in order for it to analyse them.
  In the execution of \Cref{fig:potential-designs:b} 
  the \client sends the message $\texttt{Auth}(\texttt{``Bob''}, \texttt{``pwd''})$ to \monB \pill{1}. 
  The monitor checks that its contents conform with the protocol before proceeding to forward the message to the \server \pill{2}. 
  The \server replies by sending the message $\texttt{Fail}(\texttt{1})$ to \monB \pill{3}, which forwards it straight to the \client \pill{4}. 

On the one hand, the monitor in \Cref{fig:potential-designs:a} is completely passive: it performs analysis in response to the events received. 
On the other hand, the monitor in \Cref{fig:potential-designs:b} is also responsible for \emph{forwarding} messages between the \client and the \server. 
Thus, the communication between the two components in \Cref{fig:potential-designs:b} \emph{relies} on 
\monB: should the monitor crash or terminate abruptly, the \client and the \server will stop interacting. 
Moreover, the setup in \Cref{fig:potential-designs:b} introduces additional delays when every communicated message passes through \monB; these are avoided in \Cref{fig:potential-designs:a}.  
The main drawback of the setup in \Cref{fig:potential-designs:a} is that the \server is directly exposed to an untrusted client, with additional responsibility of reporting events.
In contrast, the instrumentation in \Cref{fig:potential-designs:b} provides a layer of protection to the \server from potentially malicious interactions: if the \client sends a message that violates the protocol, \monB is able to flag the message without forwarding it to the \server. 
Moreover, the setup in \Cref{fig:potential-designs:b} provides more flexibility for reasoning on the run-time monitoring of systems where \emph{both} the \client and the \server are black boxes.
This work 
opts for the setup in \Cref{fig:potential-designs:b}. 

\subsection{A Monitor Calculus}

\Cref{fig:monitor-calculus} describes the structure and behaviour of a partial-identity monitor operating as in \Cref{fig:potential-designs:b}.
Monitors are similar to the processes defined in \Cref{fig:process-calc-syntax-semantics}, with a few key additions.
Since monitors need to interact with the environment, they also include the constructs $\envsnd{l}{a}.M$ and $\envrcv$, and rules \ltsrule{mOut} and \ltsrule{mIn}:
they are analogous to the process output and input constructs, where
interaction takes place between the
environment and the monitor instead. 
We use the terms \textbf{internal} and \textbf{external} to differentiate
between actions involving the monitored process and the environment, respectively. 

\begin{figure}
  \small
  \textbf{Syntax}\hspace{1cm}
  \(  
    \begin{array}{r@{\hskip 2mm}r@{\hskip 2mm}c@{\hskip 2mm}l}
      \text{Monitor} & M,N & \coloneqq & \monsndk{l}{a}.M \;\vert\;\monrcv{i}{I} \;\vert\;\monout{l}{a}.M \;\vert\; \monin{i}{I}\\[1mm]
      & & & \;\vert\; \monmu{X}.M \;\vert\; X \;\vert\; \pif{A}{M}{N} \;\vert\; \pnil \;\vert\; \no{}{P} \;\vert\; \no{}{E}\\
    \end{array}
  \)
  \bigskip\\
  \textbf{Semantics}
  \medskip\\
  \centerline{
    \inference[\text{[\textsc{mSnd}]}]{}{\monsndk{l}{v}.M \xrightarrow{\monsndact{l}{v}} M}
    \hspace{0.5cm}
    \inference[\text{[\textsc{mOut}]}]{}{\monout{l}{v}.M \xrightarrow{\monoutact{l}{v}} M}
    \hspace{0.5cm}
    \inference[\text{[\textsc{mRec}]}]{}{\monmu{X}.M \xrightarrow{\tau} M[\nicefrac{\monmu{X}.M}{X}]}
  }\bigskip\\
  \centerline{
    \inference[\text{[\textsc{mRcv}]}]{}{\monrcv{i}{I} \xrightarrow{\monrcvactkNoVIdx{j}} M_{j}[\nicefrac{v}{x_j}]} $j\in I$%
    \hspace{0.5cm}
    \inference[\text{[\textsc{mIn}]}]{}{\monin{i}{I} \xrightarrow{\moninactk{j}} M_{j}[\nicefrac{v_{j}}{x_j}]} $j\in I$%
  }\medskip\\
  \centerline{
    \inference[\text{[\textsc{mTru}]}]{A \Downarrow \textsf{tt}}{\pif{A}{M}{N} \xrightarrow{\tau} M}
    \hspace{1cm}
    \inference[\text{[\textsc{mFls}]}]{A \Downarrow \textsf{ff}}{\pif{A}{M}{N} \xrightarrow{\tau} N}
  }\bigskip\\
  \textbf{Violation Semantics}
  \medskip\\
  \centerline{
    \inference[\text{[\textsc{mIV}]}]{}{\monrcv{i}{I} \xrightarrow{\monrcvactkNoVIdx{{}}} \no{}{P}} $\forall i \!\in\!I: {\texttt{l} \!\neq\! \texttt{l}_i}$%
    \hspace{0.5cm}
    \inference[\text{[\textsc{mEV}]}]{}{\monin{i}{I} \xrightarrow{\moninactk{{}}} \no{}{E}} $\forall i \!\in\!I: {\texttt{l} \!\neq\! \texttt{l}_i}$%
  }\medskip\\
  \caption[Monitor Syntax and Semantics]{Monitor Syntax and Semantics}\label{fig:monitor-calculus}
\end{figure}

\begin{wrapfigure}[11]{r}{5.8cm}
  \vspace{-0.5cm}
  \begin{subfigure}{5.8cm}
    \begin{tikzpicture}[decoration=penciline]
      \node(client) [process, fill=gray!20] {\textit{client}};
      \node(monitor) [monitor, right=2cm of client]{\textit{monitor}};
      \draw[->, dashed, semithick,] (client) edge node(auth-mon)[above, font=\scriptsize, yshift=-1pt]{$\texttt{Login}(\texttt{``Bob''})$} (monitor);
      \node(static-error-internal)[draw, decorate, font=\scriptsize,below=0.2cm of monitor]{$\not{\rightvdash}\triangleleft \texttt{Login}(\texttt{``Bob''}):S_\textit{auth}$};
      \node(no-s-p)[right=0.5cm of monitor, inner sep=1.5 pt, text=red!70!black] {$\no{}{P}$};
      \draw[-stealth, dotted, semithick] (monitor) edge (no-s-p);
    \end{tikzpicture}
    \caption{\centering Internal violation.}
    \label{fig:monitor-violations:nosp}
  \end{subfigure}
  \smallskip

  \begin{subfigure}{5.8cm}
    \begin{tikzpicture}[decoration=penciline]
      \node(monitor) [monitor]{\textit{monitor}};
      \node(server) [process, right=2cm of monitor] {\textit{server}};
      \draw[->, dashed, semithick] (server) edge node(auth-mon)[above, font=\scriptsize, yshift=-1pt]{$\texttt{Res}(\texttt{227})$} (monitor);
      \node(static-error-internal)[draw, decorate, font=\scriptsize,below=0.2cm of monitor]{$\not{\rightvdash}\triangleright \texttt{Res}(\text{227}):S_\textit{auth}$};
      \node(no-s-e)[left=0.5cm of monitor, inner sep=1.5 pt, text=red!70!black] {$\no{}{E}$};
      \draw[-stealth, dotted, semithick] (monitor) edge (no-s-e);
    \end{tikzpicture}
    \caption{\centering External violation.} %
    \label{fig:monitor-violations:nose}
  \end{subfigure}
  \caption{Monitor violations.}
  \label{fig:monitor-violations}
 \end{wrapfigure}
As shown in Figure \ref{fig:monitor-violations}, monitors can reach two kinds of \textbf{rejection verdicts}, namely $\no{}{P}$ and $\no{}{E}$; 
the $P$ and $E$ tags distinguish between violations committed by the \emph{monitored process} ($P$) and the \emph{environment} ($E$). 
The rules \ltsrule{mIV} and \ltsrule{mEV} specify how the monitor reaches a verdict. 
Rule \ltsrule{mIV} represents the case when the monitor receives a violating message $\pmsgk{l}{v}$ and consequently reaches the verdict $\no{}{P}$; 
the message is deemed violating since its label is not among those that the monitor expects to receive. 
Symmetrically, in rule \ltsrule{mEV} the monitor reaches $\no{}{E}$ when it receives a violating message from the external environment. 
The following example outlines the scenarios in which monitors reach a verdict.

\begin{example}\label{eg:monitor-violations}
  \Cref{fig:monitor-violations} depicts a monitor verifying the conformity of a \client with the session type $\sauth$ (from \Cref{eg:session-types-noa}). 
  In \Cref{fig:monitor-violations:nosp}, the \client sends the message $\texttt{Login}(\texttt{``Bob''})$.
  Since the type $\sauth$ states that the \client should send a message with label $\texttt{Auth}$, %
  the monitor reaches the verdict $\no{}{P}$ by rule \ltsrule{mIV}.
  In \Cref{fig:monitor-violations:nose}, the monitor receives $\texttt{Res}(\texttt{227})$ from the environment (which represents a buggy server). In this case the monitor reaches the verdict $\no{}{E}$ (by rule \ltsrule{mEV}) since the message does not conform with $\sauth$ which states that the \client should receive either $\texttt{Succ}$ or $\texttt{Fail}$. 
\exqed
\end{example}
 
\begin{remark}
  According to \Cref{fig:monitor-calculus},
  our monitors can reach a verdict explicitly in their syntax (by having $\no{}{P}$/$\no{}{E}$ in their body),
  or by just transitioning to a verdict via rules \ltsrule{mIV} or \ltsrule{mEV}.
  We will make use both methods for our synthesised monitors (see \Cref{s:monitor-synthesis}).\exqed
\end{remark}

\subsection{Composite Monitored System}
\label{s:composite-monitored-system}

\begin{figure}
  \small
  \centerline{
    \inference[\text{[\textsc{iSnd}]}]{P \xrightarrow{\triangleleft \texttt{l}(v)} P' & M \xrightarrow{\triangleright\texttt{l}(v)} M'}{\langle P ; M\rangle \xrightarrow{\tau} \langle P' ; M'\rangle}
    \hspace{1cm}
    \inference[\text{[\textsc{iRcv}]}]{P \xrightarrow{\triangleright\texttt{l}(v)} P' & M \xrightarrow{\triangleleft\texttt{l}(v)} M'}{\langle P ; M\rangle \xrightarrow{\tau} \langle P' ; M'\rangle}
  }\medskip
  \centerline{
    \inference[\text{[\textsc{iOut}]}]{M \xrightarrow{\envsndOP\texttt{l}(v)} M'}{\langle P ; M\rangle \xrightarrow{\envsndOP\texttt{l}(v)} \langle P ; M'\rangle}
    \hspace{1cm}
    \inference[\text{[\textsc{iIn}]}]{M \xrightarrow{\envrcvOP\texttt{l}(v)} M'}{\langle P ; M \rangle \xrightarrow{\envrcvOP\texttt{l}(v)} \langle P ; M'\rangle}
  }\medskip
  \centerline{
    \inference[\text{[\textsc{iProc}]}]{P \xrightarrow{\tau} P'}{\langle P ; M\rangle \xrightarrow{\tau} \langle P' ; M\rangle}
    \hspace{1cm}
    \inference[\text{[\textsc{iMon}]}]{M \xrightarrow{\tau} M'}{\langle P ; M\rangle \xrightarrow{\tau} \langle P ; M'\rangle}
  }
  \caption{Composite monitored system semantics.}\label{fig:composite-system}
\end{figure}
 
The rules in \Cref{fig:composite-system} formalise the behaviour of the monitor when composed with the process to monitor, while also interacting with an environment (\ie another process). 
This setup is depicted in \Cref{fig:chosen-design}. 
We refer to a process $P$ instrumented with a monitor $M$ as a \textbf{composite (monitored) system}, denoted as $\langle P; M\rangle$.
The rules \ltsrule{iRcv} and \ltsrule{iSnd} model
\begin{wrapfigure}[4]{r}{5.2cm}
  \captionsetup{justification=centering}%
  \centering
  \vspace{-0.5cm}%
  \centering
  \begin{tikzpicture}
    \node(client) [process, fill=gray!20] {$P$};
    \node(monitor) [monitor, right=1cm of client]{$M$};
    \node(server) [process, right=1cm of monitor] {$Q$};
    \draw [dotted] ($(client.north west)+(-0.12cm,0.12cm)$) rectangle ($(monitor.south east)+(0.12cm,-0.12cm)$);

    \draw [dotted] ($(server.north west)+(-0.12cm,0.12cm)$) rectangle ($(server.south east)+(0.12cm,-0.12cm)$);

    \draw[thin,double distance=2pt, fill=white] (client) -- node(client-mon-ch)[]{} (monitor);

    \draw[thin,double distance=2pt, fill=white] (monitor) -- (server);

    \node(composite-system) [above=0.22cm of client-mon-ch, font=\scriptsize] {\textit{composite system}};
    \node(composite-system) [above=0.14cm of server, font=\scriptsize] {\textit{environment}};

  \end{tikzpicture}
   \vspace{-0.2cm}%
  \caption[The composite system interacting with the environment.]{}\label{fig:chosen-design}
\end{wrapfigure}
the interaction \emph{within} the composite system, (\ie between the monitored process $P$ and the monitor $M$ in \Cref{fig:chosen-design}). 
Note that the interaction between the two is \emph{synchronous}: 
the monitor (\resp process) can only send a message when the process (\resp monitor) can receive the same message. 
If $P$ sends a message (by \ltsrule{iSnd}) that violates the monitor's inputs, $M$ is able to flag the violation by rule \ltsrule{mIV}. 
The rules \ltsrule{iOut} and \ltsrule{iIn} model the interaction between the composite system and the environment. 
As shown in \Cref{fig:chosen-design}, the monitor is the entity that interacts with the environment (represented as a process $Q$). 
Accordingly, the monitor can flag a message sent by the environment if the message violates the monitor's expected inputs, by rule \ltsrule{mEV}.
The rules \ltsrule{iProc} and \ltsrule{iMon} allow the monitored process and the monitor respectively to perform actions independent of each other (\eg to recurse or branch internally). 

Our partial identity monitors halt upon reaching a verdict, in contrast to instrumented sequence recognisers that operationally continue to process events without changing their (irrevocable) verdict~\cite{DBLP:conf/fossacs/Francalanza16,Monitors:21}.
As a result, our monitors also halt any interactions between the composite system and the environment. 
Because of this, monitor correctness is of paramount importance.
The following example outlines the impact of a poorly constructed monitor.

\begin{example}\label{eg:bad-monitor}
  Recall process\;
  \begin{math}
    P_\textit{auth} \,\equal\, \pmu{X}.\triangleleft \texttt{Auth}(\texttt{``Bob''}, \texttt{``pwd''}).P_\textit{res}
  \end{math} (\Cref{eg:process-calculus}),
  \;which adheres to the session type\;
  \begin{math}
    S_\textit{auth} \,\equal\, \strec{Y}.!\texttt{Auth}(\textit{uname}:\textsf{Str},\textit{pwd}:\textsf{Str}).S_\textit{res}
  \end{math}
  \;(\Cref{eg:type-system}).
  A monitor corresponding to $\sauth$ should \emph{receive} from $\pauth$, analyse the message, and forward it to the environment. 
  The following (erroneous) monitor might seem to monitor $\sauth$:
\begin{equation*}
  \label{mon-bad}M_\textit{bad}\ \equal\ \triangleright \texttt{Login}(\textit{uname}).\envsndOP\texttt{Login}(\textit{uname}).N_\textit{bad}
\end{equation*}
\noindent
If process  $\pauth$ is instrumented with monitor $M_\textit{bad}$, we observe the following behaviour:
    \begin{equation*}
    \instr{\pauth}{M_\textit{bad}} \;\xrightarrow{\tau}\; \instr{\triangleleft \texttt{Auth}(\texttt{``Bob''}, \texttt{``pwd''}).\pres[\nicefrac{P_\textit{auth}}{X}]}{M_\textit{bad}} 
    \;\xrightarrow{\tau}\; \instr{\pres[\nicefrac{P_\textit{auth}}{X}]}{\no{}{P}}
  \end{equation*}
  After $\pauth$ unfolds, %
  it sends the message $\texttt{Auth}(\texttt{``Bob''}, \texttt{``pwd''})$ to the monitor %
  as per $\sauth$. 
  However, $M_\textit{bad}$ can only receive messages with label $\texttt{Login}$, hence it transitions to $\no{}{P}$. \exqed 
\end{example} 
\subsection{Monitor Synthesis}
\label{s:monitor-synthesis}

 \Cref{def:synth-function} presents a synthesis procedure from session types (\Cref{fig:session-types-noa}) to monitors (\Cref{fig:monitor-calculus}). The monitors generated are meant to act as a proxy between the monitored process and the environment process, as outlined in \Cref{fig:potential-designs:b}. 
There are various practical advantages in having an automated synthesis function: 
it is less error prone, expedites development and improves the maintainability of the verification framework. 

\begin{definition}\label{def:synth-function}
  The \textit{monitor synthesis function} $\lsem \minus \rsem:S\mapsto M$ takes as input a session type $S$ and returns a monitor $M$. It is defined inductively, on the structure of the session type $S$: 
  
  \smallskip\centerline{
    $\lsem\stselNA{i}{I}\rsem\,\triangleq\,\triangleright\big\{\texttt{l}_i(x_i).\pif{\isValueB{i}{x_i}}{\envsndOP \texttt{l}_i(x_i).\lsem S_i\rsem}{\no{}{P}}\big\}_{i\in I}$
  }\medskip\centerline{
    $\lsem\stbraNA{i}{I}\rsem\,\triangleq\,\envrcvOP\big\{\texttt{l}_i(x_i).\pif{\isValueB{i}{x_i}}{\triangleleft \texttt{l}_i(x_i).\lsem S_i\rsem}{\no{}{E}}\big\}_{i\in I}$
  }\medskip\centerline{\hfill
    $\lsem\textsf{rec } X.S\rsem\triangleq \mu_X.\lsem S\rsem \qquad\qquad\quad\quad \lsem X\rsem \triangleq X \qquad\qquad\quad\quad \lsem\textsf{end}\rsem\triangleq \textbf{0}$\exqed} 
\end{definition}

The main cases of \Cref{def:synth-function} are those for the selection and branching types.
In the case of $S = \stselNA{i}{I}$, the synthesised monitor first waits to receive a message from the monitored process, with one of the labels specified in the type. 
Once the message is received, the monitor checks whether its payload is of the correct base type $\textsf{B}_i$, \ie $\isValueB{i}{x_i}$
(see \Cref{remark:predicates}), raising $\no{}{P}$ if it is not.
If $\isValueB{i}{x_i}$ is true, the monitor forwards the message towards the environment, and proceeds according to $\lsem S_i\rsem$. 
The synthesis for $S = \stbraNA{i}{I}$ is analogous, but the generated monitor receives a message from the environment, analyses it, and forwards it to the monitored process; any violations are attributed to the environment.

\begin{example}[Session Monitor Synthesis]\label{eg:monitor-synthesis}
  \begin{subequations}
  Recall the session type $\sauth$ in \Cref{eg:session-types-noa}:
   
  \smallskip\centerline{\(
    S_\textit{auth}\ \equal\ \strec{Y}.!\texttt{Auth}(\textsf{Str},\textsf{Str}).S_\textit{res}
    \quad\text{\small{where} }\;
    S_\textit{res}\ \equal\ \& \big\{?\texttt{Succ}(\textsf{Str}).S_\textit{succ},?\texttt{Fail}(\textsf{Int}).Y \big\}
  \)}\smallskip
  
  \noindent%
  The 
  synthesis for this type 
  first generates the recursion construct $\pmu{Y}$ 
  followed by the synthesis for the selection type: %

  \smallskip\centerline{\(
    \mauth \ \equal\ \lsem S_\textit{auth}\rsem\ \equal
    \begin{cases}
      \pmu{Y}.\triangleright \texttt{Auth}(\textit{uname},\textit{pwd}). \;\textsf{if } \bigl(\isValueB{\textsf{Str}}{\textit{uname}} \wedge \isValueB{\textsf{Str}}{\textit{pwd}}\bigr) 
    \\
    \qquad \qquad\textsf{ then }\envsndOP \texttt{Auth}(\textit{uname},\textit{pwd}).\lsem S_\textit{res}\rsem\textsf{ else }\no{}{P}  
    \end{cases}
  \)}\smallskip

  \noindent%
  Monitor $\mauth$ first waits to receive a message with label \texttt{Auth} from the monitored process (via $\triangleright$), checks the types of the payload  $\bigl(\isValueB{\textsf{Str}}{\textit{uname}} \wedge \isValueB{\textsf{Str}}{\textit{pwd}}\bigr)$, and proceeds to forward the message %
  to the  environment (via $\envsndOP$),  continuing as the monitor of $S_\textit{res}$:

  \smallskip\centerline{\(
    \lsem S_\textit{res}\rsem\ \equal\ \envrcvOP\left\{
      \begin{array}{ll} \texttt{Succ}(\textit{tok}).\textsf{if } \isValueB{\textsf{Str}}{\textit{tok}}\textsf{ then }\triangleleft \texttt{Succ}(\textit{tok}).\lsem S_\textit{succ} \rsem\textsf{ else } \no{}{E}, \\ 
      \texttt{Fail}(\textit{code}).\textsf{if } \isValueB{\textsf{Int}}{\textit{code}}\textsf{ then }\triangleleft \texttt{Fail}(\textit{code}).Y\textsf{ else } \no{}{E} 
    \end{array}\right\}
  \)}\smallskip
  
  \noindent%
  Observe that $\lsem S_\textit{res}\rsem$ inputs from the environment 
  and outputs to the monitored process.%
  \exqed
\end{subequations}
\end{example} 
If process $\pauth$ is instrumented with monitor $\mauth$ as the composite system $\instr{\pauth}{\mauth}$, we observe the behaviour outlined in \Cref{fig:potential-designs:b},
as we show in the following example.

\begin{example}
  \label{eg:good-monitor}
  \begin{subequations}
  Recall $\pauth$ defined in Example \ref{eg:process-calculus}:
  
  \smallskip\centerline{\(
    P_\textit{auth}\ \equal\ \pmu{X}.(\triangleleft \texttt{Auth}(\texttt{``Bob''}, \texttt{``pwd''})).\pres
    \quad\text{\small{where} }\;%
    \pres\ \equal\ \triangleright \big\{ \texttt{Succ}(\textit{tok}).P_\textit{succ},\texttt{Fail}(\textit{code}).P_\textit{fail}\big\}
  \)}\smallskip

  \noindent%
  When $\pauth$ is instrumented with the monitor $\mauth \equal \synth{\sauth}$ we observe the behaviour:
  \begin{align}
    \nonumber
    \instr{P_\textit{auth}}{\mauth} &\xrightarrow{\tau} \instr{P'_\textit{auth}}{\mauth} 
    \qquad
    \text{where }P'_\textit{auth} \equal \triangleleft \texttt{Auth}(\texttt{``Bob''}, \texttt{``pwd''}).\pres[\nicefrac{P_\textit{auth}}{X}]
    \\[2mm]
    \nonumber
    \instr{P'_\textit{auth}}{\mauth} &\xrightarrow{\tau} \instr{P'_\textit{auth}}{M'_\textit{auth}} \alignhfill{using \ltsrule{iMon}}\\
    \nonumber
    \text{where }M'_\textit{auth} &\equal \big(\triangleright \texttt{Auth}(\textit{uname},\textit{pwd}).
    \textsf{if }\bigl(\isValueB{\textsf{Str}}{\textit{uname}} \wedge \isValueB{\textsf{Str}}{\textit{pwd}}\bigr)\\[-1mm]
    \notag&\qquad\qquad\textsf{ then }\envsndOP \texttt{Auth}(\textit{uname},\textit{pwd}).\lsem S_\textit{res}\rsem\textsf{ else }\no{}{P}\big)[\nicefrac{\mauth}{Y}]
  \end{align}
  After unfolding, using the rules \ltsrule{iProc} and \ltsrule{iMon} respectively, the monitor can receive and the process can send, and  they can transition together to communicate:
  (see \pill{1} in \Cref{fig:potential-designs:b})
  \begin{align}
    \nonumber
    P'_\textit{auth} &\xrightarrow{\triangleleft \texttt{Auth}(\texttt{``Bob''}, \texttt{``pwd''})} P''_\textit{auth} 
    \;\text{where}\;P''_\textit{auth} \equal \triangleright \big\{\texttt{Succ}(\textit{tok}).P_\textit{succ},\texttt{Fail}(\textit{code}).P_\textit{fail}\big\}[\nicefrac{P_\textit{auth}}{X}]
    \\
    \nonumber
    M'_\textit{auth} &\xrightarrow{\triangleright \texttt{Auth}(\texttt{``Bob''}, \texttt{``pwd''})} M''_\textit{auth} 
    \text{ where }\\
    \nonumber
    M''_\textit{auth} & \equal \textsf{if } 
    \bigl(\isValueB{\textsf{Str}}{\texttt{``Bob''}} \wedge \isValueB{\textsf{Str}}{\texttt{``pwd''}}\bigr)
    \textsf{ then }\envsndOP \texttt{Auth}(\texttt{``Bob''}, \texttt{``pwd''}).\lsem S_\textit{res}\rsem[\nicefrac{\mauth}{Y}] \textsf{ else }\no{}{P}
    \\
    \nonumber
    & \instr{P'_\textit{auth}}{M'_\textit{auth}} \xrightarrow{\tau} \instr{P''_\textit{auth}}{M''_\textit{auth}} 
  \end{align}
  The monitor proceeds by checking the values of the payload values using the rule \ltsrule{iMon}.
  \begin{align}
    \nonumber
    M''_\textit{auth} &\xrightarrow{\tau} M'''_\textit{auth}
    \quad\text{where}\; M'''_\textit{auth}\ \equal\ \envsndOP \texttt{Auth}(\texttt{``Bob''}, \texttt{``pwd''}).\lsem S_\textit{res}\rsem[\nicefrac{\mauth}{Y}]
    \\
    \nonumber
    \instr{P''_\textit{auth}}{M''_\textit{auth}}&\xrightarrow{\tau} \instr{P''_\textit{auth}}{M'''_\textit{auth}} 
  \end{align}
  $M'''_\textit{auth}$ now forwards the message to the environment by rule \ltsrule{iOut}:
  (see \pill{2} in \Cref{fig:potential-designs:b})
  \begin{align}
    \nonumber
    \instr{P''_\textit{auth}}{M'''_\textit{auth}} \xrightarrow{\envsndOP \texttt{Auth}(\texttt{``Bob''}, \texttt{``pwd''})} \instr{P''_\textit{auth}}{\lsem S_\textit{res}\rsem[\nicefrac{\mauth}{Y}]
    }
  \end{align}
  The monitor is currently waiting to receive from the environment, since:%
  \begin{align}
    \notag
    \lsem S_\textit{res}\rsem[\nicefrac{\mauth}{Y}]
    \equal
    &
    {\envrcvOP}\left\{\mkern-10mu
    \begin{array}{ll} 
      \texttt{Succ}(\textit{tok}).\textsf{if }\isValueB{\textsf{Str}}{\textit{tok}}\textsf{ then}\triangleleft \texttt{Succ}(\textit{tok}).\lsem S_\textit{succ} \rsem[\nicefrac{\mauth}{Y}]
      \textsf{ else } \no{}{E} \\ 
      \texttt{Fail}(\textit{code}).\textsf{if } \isValueB{\textsf{Int}}{\textit{code}}\textsf{ then}\triangleleft \texttt{Fail}(\textit{code}).\mauth\textsf{ else } \no{}{E} 
    \end{array}
    \mkern-10mu\right\}
  \end{align}
  If the monitor receives the message $\texttt{Succ}(\texttt{321})$, it  forwards the message to the monitored process and proceeds according to $\lsem S_\textit{succ} \rsem[\nicefrac{\mauth}{Y}]
  $.
  If the monitor receives the message $\texttt{Fail}(\texttt{1})$
  (see \pill{3} in \Cref{fig:potential-designs:b})
  it forwards the message to the process $P''_\textit{auth}$
  (see \pill{4} in \Cref{fig:potential-designs:b}):
  \begin{gather}
    \nonumber
    \triangleleft \texttt{Fail}(\texttt{1}).\mauth \xrightarrow{\triangleleft \texttt{Fail}(\texttt{1})} \mauth 
    \qquad\qquad\qquad
    P''_\textit{auth} \xrightarrow{\triangleright \texttt{Fail}(\texttt{1})} P_\textit{fail}[\nicefrac{P_\textit{auth}}{X}][\nicefrac{\texttt{1}}{code}]
    \\
    \nonumber
    \instr{P''_\textit{auth}}{\triangleleft \texttt{Fail}(\texttt{1}).\mauth} \xrightarrow{\tau} \instr{P_\textit{fail}[\nicefrac{P_\textit{auth}}{X}][\nicefrac{\texttt{1}}{code}]}{\mauth}
  \end{gather}
  The composite system can now proceed %
  with the monitor restarting as $\mauth$.
  \exqed
\end{subequations}
\end{example} 
Should the monitored process send a message that violates the session type, the monitor can flag the violation upon receiving a message, as the following example shows. 

\begin{example}\label{eg:bad-process}
  Consider the scenario in \Cref{fig:monitor-violations:nosp}, where the \client is the process $P_\textit{bad}$: 
  
  \smallskip\(%
    \nonumber
    P_\textit{bad}\ \equal\ \triangleleft \texttt{Login}(\texttt{``Bob''}).\triangleright \texttt{Res}(\textit{tok}:\textsf{Str}).P_\textit{res}
  \)\smallskip%
  
  \noindent%
  and recall the monitor $\mauth$ (from \Cref{eg:monitor-synthesis,eg:good-monitor})
  obtained from the session type $\sauth$.
  When $P_\textit{bad}$ is instrumented with $\mauth$, we observe the following behaviour:
  \begin{align*}
    \instr{P_\textit{bad}}{\mauth} &\xrightarrow{\tau} \instr{P_\textit{bad}}{M'_\textit{auth}} \xrightarrow{\tau} \instr{\triangleright \texttt{Res}(\textit{tok}:\textsf{Str}).P_\textit{res}}{\no{}{P}} \alignhfill{using \ltsrule{iMon},\ltsrule{iSnd}}
    \\
    & \text{where } M'_\textit{auth} \equal 
    \begin{cases}
      \big(\triangleright \texttt{Auth}(\textit{uname},\textit{pwd}).\textsf{if }\bigl(\isValueB{\textsf{Str}}{\textit{uname}} \wedge \isValueB{\textsf{Str}}{\textit{pwd}}\bigr)\\
    \qquad\qquad\textsf{ then }\envsndOP \texttt{Auth}(\textit{uname},\textit{pwd}).\lsem S_\textit{res}\rsem
    \textsf{else }\no{}{P}\big)[\nicefrac{\mauth}{Y}]   
    \end{cases}
  \end{align*}
  \ie $\mauth$ unfolds, %
  receives $\texttt{Login}(\texttt{``Bob''})$ from $P_\textit{bad}$,
  and flag the rejection verdict $\no{}{P}$.%
  \exqed 
\end{example}  
\section{Formal Analysis and Results}
\label{s:formal-results}

In \Cref{s:designing-hybrid-methodology} we argued for the importance of monitor correctness. 
This has also been recognised by other works that study monitoring techniques for session types~\cite{DBLP:journals/tcs/BocchiCDHY17,DBLP:conf/popl/JiaGP16,DBLP:conf/esop/GommerstadtJP18}.  
However, these attempts all propose their own bespoke notion of monitor correctness that is often hard to relate to the others.
\begin{figure}
  \centering
  \scalebox{0.9}{
  \begin{minipage}{0.9\linewidth}
    \centering
    \begin{tikzpicture}[decoration=penciline]
      \node(client-left) [static-process, fill=gray!20] {\textit{P}};
      \node(violates-s) [decorate, draw,  above right=0.5cm of client-left, font=\footnotesize] {\textit{violates $\varphi$}};
      \node(satisfies-s) [decorate, draw, below right=0.5cm of client-left, font=\footnotesize] {\textit{satisfies $\varphi$}};
      \draw[dotted, semithick, bend left=20] ($(client-left.north east)$) edge[-stealth] (violates-s);
      \draw[dotted, semithick, bend right=20] ($(client-left.south east)$) edge[-stealth] (satisfies-s);

      \draw [dotted] ($(client-left.north west)+(-0.3cm,1.1cm)$) rectangle ($(client-left.south east)+(2.5cm,-1.1cm)$);
    
      \node(soundness)[single arrow,shape border rotate=180, semithick, draw=black, rounded corners=1pt, fill=white, font=\scriptsize,  text centered, anchor=center, minimum width=7mm, minimum height=25mm,inner sep=0mm, single arrow head extend=1mm, inner sep=0mm, align=center, right=3.7cm of client-left, yshift=0.55cm, drop shadow=ashadow] {\textit{\phantom{p}soundness\phantom{p}}};
      
      \node[single arrow, fill=white, semithick, draw=black, below=0.6cm of soundness, rounded corners=1pt, font=\scriptsize, text centered, anchor=center, minimum width=7mm, minimum height=25mm, inner sep=0mm, single arrow head extend=1mm, inner sep=0, align=center, drop shadow=oppashadow] {\textit{completeness}};
    
      \node(client-right) [process, right=11cm of client-left, fill=gray!20] {\textit{P}};
      \node(monitor) [monitor, left=0.45cm of client-right]{\textit{monitor}};
      \draw (client-right) edge[{Arc Barb[reversed]}->] (monitor);
      \node(spec-s)[decorate, draw, font=\footnotesize, above=0.4cm of monitor] {$\varphi$};
      \draw[-angle 90,dashed] (spec-s) edge (monitor);

      \node(satisfaction-mon) [circle, draw, below left=0.5cm and 0.6cm of monitor, text=green!70!black, font=\scriptsize, inner sep=1pt, minimum size=4mm] {\cmark};
      \node(violation-mon) [circle, draw, above left=0.5cm and 0.6cm of monitor, text=red!70!black, font=\scriptsize, inner sep=1pt, minimum size=4mm] {\xmark};

      \draw[dotted, semithick, bend left=20] ($(monitor.south west)$) edge[-stealth] (satisfaction-mon);
      \draw[dotted, semithick, bend right=20] ($(monitor.north west)$) edge[-stealth] (violation-mon);

      \draw [dotted] ($(client-right.north west)+(-3.3cm,1.1cm)$) rectangle ($(client-right.south east)+(0.3cm,-1.1cm)$);
    \end{tikzpicture}
  \end{minipage}
  }%
  \caption{Monitoring soundness and completeness, from a logic-based viewpoint \cite{DBLP:journals/fmsd/FrancalanzaAI17, DBLP:journals/pacmpl/AcetoAFIL19, DBLP:conf/sefm/AcetoAFIL19}.}
  \label{fig:sound-vs-complete}
\end{figure}
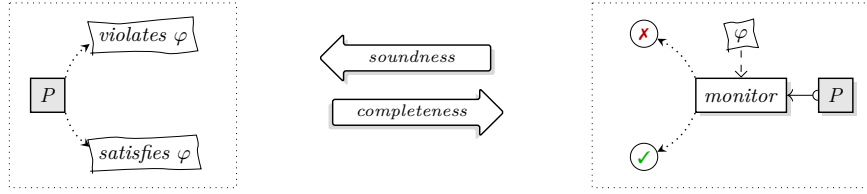 %
Instead, we strive towards a more systematic approach for monitor correctness and study monitor correctness
in relation to an independent characterisation of process correctness.
More concretely, we assess the correctness of \emph{session monitors} in relation to \emph{session typing}. 
We draw inspirations from a recent body of work that captures this relationship in terms of \emph{soundness} and \emph{completeness}~\cite{DBLP:journals/fmsd/FrancalanzaAI17, DBLP:journals/pacmpl/AcetoAFIL19}, as depicted in \Cref{fig:sound-vs-complete}.
In such body of work, \emph{monitor soundness} states that if a monitor $M$ is monitoring a process $P$ for a property $\varphi$, and $M$ reaches a rejection (\resp acceptance) verdict, then such a verdict must correspond to $P$'s violations (\resp satisfactions) of property $\varphi$.
\emph{Monitor completeness} is the dual property: if a process $P$ violates (\resp satisfies) a property $\varphi$, then the monitor that runtime-checks $P$ for $\varphi$ must reach a rejection (\resp acceptance) verdict.
This formulation is appealing to our study for a number of reasons:
\begin{itemize}
  \item The touchstone logic used to specify process correctness is the Hennessy-Milner Logic with minimal and maximal fixpoints (recHML)~\cite{Aceto07}; like session types, it has a tight relation to ($\omega$-)regular properties, and a long tradition of automata-based interpetations.
  \item Recent work~\cite{DBLP:conf/sefm/AcetoAFIL19,AcetoAFIL21:SOSYM} has extended this framework to a spectrum of correctness criteria.  This gives us the flexibility of identifying the criteria that best fit our concerns.  
\end{itemize} 
To study session types monitorability, we adapt this theoretical framework to our setting: 

\begin{enumerate}[{M}1]%
\item\label{item:mon:spec} instead of logic formulas as specifications, we adopt session types as specifications; and
\item\label{item:mon:charact} to characterise processes satisfying a specification, we use the session typing system.
\end{enumerate}

\noindent
This leads to important differences between our approach and \cite{DBLP:journals/fmsd/FrancalanzaAI17, DBLP:journals/pacmpl/AcetoAFIL19}:

\begin{enumerate}[{D}1]%
  \item\label{item:difference:syntax-semantics} by item~M\ref{item:mon:charact}, our processes characterisation is syntactic (rather than semantic), which is further removed from the runtime behaviour observed by the monitor;
  \item session types describe interactions between two parties, and our monitors can attribute a violation to a party.  By contrast, monitors for recHML formulas flag generic rejections;
  \item we here limit our analysis to rejection monitors and do not consider acceptance verdicts. 
\end{enumerate}

\noindent
Consequently, we formalise our notions of monitoring soundness and completeness as follows.
Here, $t$ represents a \emph{trace}, \ie finite a sequence of environment send/receive actions
$\envsndOP\texttt{l}(v)$ and $\envrcvOP\texttt{l}(v)$ (from \Cref{fig:composite-system});
moreover, $\xRightarrow{t}$ is a sequence of transitions where the actions in $t$
are interleaved with finite sequences of $\tau$-transitions.

\begin{definition}[\textit{Session Monitor Soundness}]
  \label{def:internal-static-soundness}
  A monitor $M$ \emph{soundly monitors} for a session type $S$ iff, for all $P$, %
  if there is a trace $t$ such that\; $\instr{P}{M} \xRightarrow{t} \instr{P'}{\no{}{P}}$, \;then\; $\emptyset\cdot\emptyset \rightvdash P:S$ \;does \emph{not} hold. \exqed
\end{definition}

\begin{definition}[\textit{Session Monitor Completeness}]
  \label{def:internal-static-completeness}
  A monitor $M$ monitors for a session type $S$ in a \emph{complete} manner, iff for all processes $P$, %
  whenever\; $\emptyset\cdot\emptyset \rightvdash P:S$ \;does \emph{not} hold, then there exists a trace $t$ such that  $\instr{P}{M}\xRightarrow{t}\instr{P'}{\no{}{P}}$. \exqed
\end{definition}

\subsection{Soundness of Session Type Monitoring}
\label{sec:monitor-soundness-maintext}

A tenet of \cite{DBLP:journals/fmsd/FrancalanzaAI17, DBLP:journals/pacmpl/AcetoAFIL19, DBLP:conf/sefm/AcetoAFIL19} is that, in order to have monitor correctness, soundness (\Cref{def:internal-static-soundness}) is not negotiable. 
We here show that our monitor synthesis procedure is sound, \ie we show that for any session type $S$, monitor $\synth{S}$ observes \Cref{def:internal-static-soundness} \wrt specification $S$.

\begin{theorem}[\textit{Synthesis Soundness}]\label{thrm:sound}
  For all session types $S$ and processes $P$, %
  if there exists a trace $t$ such that $\instr{P}{\synth{S}}\xRightarrow{t} \instr{P'}{\textsf{no}_P}$, \;then\; $\emptyset\cdot\emptyset \rightvdash P:S$ \;does \emph{not} hold. 
\end{theorem}
\begin{proof}
  Instead of proving the statement directly, we prove its contrapositive:
  \begin{quote}
    For all session types $S$ and processes $P$ %
    such that\; $\emptyset\cdot\emptyset\rightvdash P:S$,
    \;if\; $\instr{P}{\synth{S}} \xRightarrow{t} \instr{P'}{M'}$
    then\; $M'\not\equal \textsf{no}_P$.
  \end{quote}
  To this end, we first establish a  \emph{subject reduction}
  result, relying on standard properties of our type system:
  this determines how process $P$ evolves \wrt its session type $S$.
  Then, we prove the contrapositive statement above %
  by \emph{lexicographical induction} on the derivation of\; $\emptyset\cdot\emptyset\rightvdash P:S$ \;and the number of transitions in the trace $\instr{P}{\synth{S}} \xRightarrow{t} \instr{P'}{M'}$. This requires some sophistication, because as the instrumented system $\instr{P}{\synth{S}}$ evolves,
  for each step of $P$ the monitor $\synth{S}$ %
  (as generated by \Cref{def:synth-function}) may take multiple steps to evaluate synthesised conditions before it can forward messages.
  Hence, we prove additional results to handle such cases, and formulate a suitable
  induction hypothesis allowing us to complete the proof of the contrapositive statement.
  \Cref{thrm:sound} follows as a corollary.
\end{proof}

As a 
by-product of \Cref{thrm:sound} we also deduce that if a process $P$ has type $S$,
then the instrumented process $\instr{P}{\synth{S}}$ can only get stuck due to an \emph{external} violation, \ie $\no{}{E}$; this arises when the environment sends a message with a wrong label or payload type.
This result is formalised in \Cref{thm:blame-correctnes} below, and is reminiscent of the notion of \emph{blaming} in gradual types
(\ie 
untyped components can always be blamed in case of errors~\cite{DBLP:conf/popl/AhmedFSW11,DBLP:journals/pacmpl/IgarashiTVW17}).

\begin{corollary}[Monitor Blaming]
  \label{thm:blame-correctnes}
  For any process $P$ and session types $S$ where\: $\emptyset\cdot\emptyset\rightvdash P:S$, for any trace $t$ such that\; $\instr{P}{\synth{S}} \xRightarrow{t} \instr{P'}{M'} \mathrel{\not\rightarrow}{}$ where  $P \neq \pnil$, we have $M' \equal \no{}{E}$. \exqed
\end{corollary}

\subsection{On the Completeness of Session Type Monitoring}
\label{sec:monitor-synthesis-completeness}

Monitor soundness, by itself, is a weak result. For instance, the monitor that merely acts as a forwarder between the monitored process and the environment, \emph{never raising any detections}, is trivially sound but, arguably, not very useful. 
One way to force the monitor to produce useful detections is via \emph{completeness}, as per \Cref{def:internal-static-completeness} above.
We investigate completeness for our synthesised monitors by 
first establishing a ``weak'' completeness result (\Cref{sec:weak-completeness})
showing how ill-typed processes can misbehave when instrumented. 
Then, we prove that a ``full'' completeness result is impossible in our black-box monitoring model (\Cref{sec:completeness-impossible}).

\subsubsection{Weak Monitor Synthesis Completeness}
\label{sec:weak-completeness}

To achieve our completeness result, in this section we need a \emph{precise typing} assumption on predicates $A$:
ill-typed predicates do not evaluate to a boolean --- \ie
if\; $\Gamma \rightvdash A : \textsf{Bool}$ \;does \emph{not} hold,
then\; $A \not\Downarrow \true$ \;and\; $A \not\Downarrow \false$.
Furthermore, we need to limit our analysis to processes without \emph{dead code} (\Cref{def:no-dead-code} below).  
For the process language of \Cref{fig:process-calc-syntax-semantics}, this means: for every ``if'' statement occurring in a process $P$, there are executions of $P$ where the left branch is taken, and executions where the right branch is taken.
These executions depend on $P$'s inputs, which may cause different instantiations to $P$'s variables. 
\Cref{ex:why-no-dead-code} illustrates why we need this assumption; note that these assumptions are \emph{not} needed for monitor soundness.

\begin{example}
  \label{ex:why-no-dead-code}%
  The process\; $P = \pif{\textsf{tt}}{\psnd{1}.\pnil}{\psnd{2}.\pnil}$
  \;is \emph{not} typable with\;
  $S = \stselNA{i}{\{1\}}$ (for any $S_i$):
  it is only typable with internal choices of the form $\stselNA{i}{1..n}$,
  with $n \ge 2$.
  Yet, $P$ would operate correctly if instrumented with monitor
  $\lsem{S}\rsem$, because its ``else'' branch is dead code.
  If we remove the dead code from $P$, the remaining process\;
  $\psnd{1}.\pnil$ \;is typable with $S$, and behaves like $P$. \exqed
\end{example}

\begin{definition}
  \label{def:no-dead-code}
  \emph{A process $P$ has no dead code} iff
  for all its subterms of the form\; $P' = \pif{A}{Q}{Q'}$,
  \;there exist traces $t$ and $t'$
  and substitutions $\sigma$ and $\sigma'$
  such that\; $P \xRightarrow{t} P'\sigma \xrightarrow{\tau} Q\sigma$
  (hence, $A\sigma \Downarrow \true$)
  \;and\; $P \xRightarrow{t'} P'\sigma' \xrightarrow{\tau} Q'\sigma'$
  (hence, $A\sigma \Downarrow \false$).%
  \exqed
\end{definition}

With the ``no dead code'' assumption, we can formulate our weak completeness result.  
It states that when a process $P$ is ill-typed for a session type $S$, then the monitored system  $\instr{P}{\lsem S\rsem}$ exhibits at least one execution that gets stuck due to $P$'s behaviour, without any violation by the environment.   

\begin{restatable}[Weak Monitor Synthesis Completeness]{theorem}{lemMonitorCompleteness}
  \label{lem:partial-monitor-completeness}
  Take any closed process $P$ without dead code
  such that\; $\emptyset\cdot\emptyset \rightvdash P:S$ \;does \emph{not} hold.
  Then, there exists a trace $t$ such that\; %
  $\instr{P}{\lsem S\rsem} \xRightarrow{t} \instr{P'}{M'} \mathrel{\not\rightarrow}{}$,
  \;with
  $P' \not\equal \pnil$ or $M' \not\equal \pnil$;
  moreover, $M' \not\equal \no{}{E}$.
\end{restatable}
\begin{proof}
  The proof is based on \emph{failing derivations}, inspired by \cite{BHLN12,BHLN17}.
  It consists of 6 steps.
\begin{enumerate}
\item We define the \emph{rule function $\Phi$} %
  that, following the typing rules in \Cref{fig:session-typing-rules-noa},
  maps a judgement
  of the form\; $J = \Theta\cdot\Gamma \rightvdash P:S$ \;
  to either the set of
  all judgements in $J$'s premises (for inductive rules), or $\{\true\}$ (for axioms),
  or $\emptyset$ (if $J$ does not match any rule);
\item we formalise a \emph{failing derivation} of a session typing judgement\;
  $\Theta\cdot\Gamma \rightvdash P:S$
  \;as a finite sequence of judgements\; $\mathcal{D} = (J_0, J_1, \ldots, J_n)$ \;such that:
  \begin{enumerate}[(i)]%
    \item for all $i \in 0..n$, $J_i$ is a judgement of the form\;
      $\Theta_i\cdot\Gamma_i \rightvdash P_i:S_i$;
    \item $J_0 = \Theta\cdot\Gamma \rightvdash P:S$
      (\ie the failing derivation $\mathcal{D}$ begins with the judgement of interest);
    \item  $\forall i \in 1..n$, $J_i \in \Phi(J_{i-1})$
      (\ie each judgement in $\mathcal{D}$ is followed by one of its premises); 
    \item $\Phi(J_n) = \emptyset$
      (\ie the last judgement in $\mathcal{D}$ does not match any rule in \Cref{fig:session-typing-rules-noa})
  \end{enumerate}
\item\label{item:mon-comp-proof:fail-type} we prove there is a failing derivation of\;
  $J = \Theta\cdot\Gamma \rightvdash P:S$ \;if and only if $J$ is \emph{not} derivable;%
\item\label{item:mon-comp-proof:neg-typing} we formalise a \emph{negated} typing judgement\;
  $\Theta\cdot\Gamma \nrightvdash P:S$
  and prove that it holds if and only if there is a corresponding 
  failing derivation of\; $\Theta\cdot\Gamma \rightvdash P:S$;
\item\label{item:mon-comp-proof:neg-typing-iff}
  thus, from items~\ref{item:mon-comp-proof:fail-type} and \ref{item:mon-comp-proof:neg-typing} above,
  we know that $\Theta\cdot\Gamma \rightvdash P:S$ \;is \emph{not} derivable
  if and only if\; $\Theta\cdot\Gamma \nrightvdash P:S$ \;is derivable.
  Consequently, the judgement\; $\Theta\cdot\Gamma \nrightvdash P:S$
  \;tells us exactly what are the possible shapes of $P$ and $S$ covered by the theorem's statement;
\item finally, we use all ingredients above to prove the thesis.
  From a failing derivation of \;$\Theta\cdot\Gamma \rightvdash P:S$
  \;(item~\ref{item:mon-comp-proof:fail-type}), we construct
  a trace $t$ leading from\;
  $\instr{P}{M}$ \;to some\; $\instr{P'}{M'}$;
  further, using the corresponding derivation of\;
  \;$\Theta\cdot\Gamma \nrightvdash P:S$
  \;(items~\ref{item:mon-comp-proof:neg-typing}, \ref{item:mon-comp-proof:neg-typing-iff}),
  we prove that $t$ is a valid trace,
  and\; $\instr{P'}{M'}\mathrel{\not\rightarrow}{}$
  with
  $P' \not\equal \pnil$ or $M' \not\equal \pnil$,
  and $M' \not\equal \no{}{E}$. \qedhere
\end{enumerate}
\end{proof}

Although \Cref{lem:partial-monitor-completeness} is weaker than the ideal requirement set out in \Cref{def:internal-static-completeness},
its proof sheds light on all the possible reasons why an ill-typed monitored process gets stuck:

\begin{enumerate}%
  \item\label{item:mon:out-wrong-out} the monitor reaches a process rejection verdict, $M' = \no{}{P}$, because the process sends a message
  with a wrong label, or payload value of a wrong base type.
\item\label{item:mon:out} the monitor blocks waiting for the process to send a message, but:
  \begin{enumerate}[a.]%
  \item\label{item:mon:out-in} $P'$ is attempting to receive a
    message itself or
  \item\label{item:mon:out-end} $P' = \pnil$ (i.e., $P'$ has terminated its execution);
  \end{enumerate}
\item\label{item:mon:in} the monitor blocks waiting for the process to receive a message, but:
  \begin{enumerate}[a.]%
  \item\label{item:mon:in-wrong-in} the process is also waiting to receive a message but does not support the required message label being sent or
  \item\label{item:mon:in-out} $P'$ is attempting to send a message itself or
  \item\label{item:mon:in-end} $P' = \pnil$; 
  \end{enumerate}
\item\label{item:mon:end} the monitor expects the process to end,
  but $P'$ is trying to send/receive more messages;
\item\label{item:mon:wrong-exp} $P'$ is stuck on an ill-typed expression.
\end{enumerate}

\begin{remark}
Process violations are only flagged $\no{}{P}$ (as required in \Cref{def:internal-static-completeness}) is case~\ref{item:mon:out-wrong-out}. 
We now discuss how a practical monitor implementation could, in principle,
  detect violations in other cases, and highlight when this additional detection power
  would require additional assumptions that go beyond our black-box monitoring design.
\begin{itemize}
\item In cases~%
  \ref{item:mon:in-wrong-in} and \ref{item:mon:wrong-exp},
  the trace $t$ may lead to a run-time error; this could be flagged %
  by assuming that the monitor can detect whether the monitored process has crashed;
\item In case~\ref{item:mon:end}, the monitor expects %
  the session to be ended. This could be handled by assuming and end-of-session signal:
  the monitor can wait for such a signal, and flag any other message sent by the process.
  However, if the process is attempting to receive (instead of ending the session), the detection is more subtle,
  as in case %
  \ref{item:mon:out-in} below;
\item Cases~%
  \ref{item:mon:out-end}
  and %
  \ref{item:mon:in-end} could be similarly handled by assuming an end-of-session signal;
\item Case %
  \ref{item:mon:out-in} is more subtle:
  both the process and monitor are waiting for a message. 
  Reception timeouts from the monitor side are inadequate because they 
  lead to unsound detections.  
  To accurately handle this case, we would need to instrument the
  process executable, which breaks our black-box assumptions from \Cref{s:monitor-design}.
  Similarly, flagging a violation in case %
  \ref{item:mon:in-out} 
  also requires access to the process code, 
  again breaking our black-box design. 
  \exqed
\end{itemize}
\end{remark}

\subsubsection{Impossibility of Sound and Complete Session Monitoring}
\label{sec:completeness-impossible}

The weakness of our completeness result in \Cref{lem:partial-monitor-completeness} is not specific to our monitor synthesis function.
Rather, we show that this is an inherent limit of the operational model (\Cref{fig:monitor-calculus,fig:composite-system}) that captures the black-box monitor design decisions of \Cref{s:monitor-design}.
Similar impossibility results often arise for 
reasonably expressive specification languages (such as the logics in 
\cite{DBLP:journals/fmsd/FrancalanzaAI17,DBLP:conf/fsttcs/AcetoAFI17,DBLP:journals/pacmpl/AcetoAFIL19,DBLP:conf/sefm/AcetoAFIL19}), where it is usually the case that only a subset of specifications can be monitored in a sound and complete way.

\begin{theorem}[Impossibility of Sound and Complete Session Monitoring] \label{thm:impossibility}
  A (closed) session type $S \neq \textsf{end}$ cannot have a sound and complete monitor under the semantics of \Cref{fig:monitor-calculus}.%
\end{theorem}

\begin{proof}
  We proceed by case analysis on the structure of $S$:
  \begin{description}
  \item[Case $S=\stbraNA{i}{I}$:] 
  We assume that a complete  monitor $M$ for $S$ exists and proceed to show that such a monitor is necessarily unsound for $S$.
    Fix a complete monitor $M$ for $S$.
    Consider the process\; $P_2 = \triangleright \big\{\texttt{l}_{i}(x_{i}).Q_{i}\big\}_{i \in I}$ \;that is well-typed \wrt the session type $S$.
    Then, consider the process $P_1$ obtained by pruning some of the top-level external choices of $P_2$,
     \ie $P_1 = \triangleright \big\{\texttt{l}_{j}(x_{j}).Q_{j}\big\}_{j \in J}$ where $J \subset I$ (a strict inclusion). 
     Observe that $P_1$ is ill-typed for $S$, and thus, by completeness  (\Cref{def:internal-static-completeness}), $M$ should reject $P_1$, (\ie there must exists a trace $t$ such that $\langle P_1;M \rangle \xRightarrow{t} \langle P'_1;\no{}{P} \rangle$). 
    There are two ways for $M$ to reach such a verdict: 
    \begin{itemize}
      \item $M \xRightarrow{t} \no{}{P}$ without interacting with $P_1$.
        In this case, the same rejection verdict is reached by the composite system $\instr{P_2}{M}$. 
        Since $P_2$ is well-typed for $S$, this means that $M$ is unsound for $S$ by \Cref{def:internal-static-soundness};
      \item $M$ reaches the rejection verdict after interacting (at least  once) with $P_1$.
      In this case, we have\; $P_1 \xrightarrow{\triangleright \texttt{l}_j(v)} Q_j$ \;(for some $j \in J$),
        and there are $t_1,t_2,P'_1$ such that $t = t_1.t_2$ \;and\; $M \xRightarrow{{t_1.\triangleleft \texttt{l}_j(v)}} M'$ %
        \;and\; $\instr{Q_j}{M'} \xRightarrow{t_2} \instr{P'_1}{\no{}{P}}$.
        \;But then, since $j \in J \subseteq I$,
        we also have\; $\instr{P_2}{M} \xRightarrow{t} \instr{P'_1}{\no{}{P}}$.  Since $M$ rejects the well-typed process $P_2$, this again makes $M$ unsound for $S$ by \Cref{def:internal-static-soundness}.
    \end{itemize}
    We have thus shown that a complete monitor $M$ for $S$ is necessarily unsound.
  \item[Case $S=\stselNA{i}{I}$:] Assume that a complete monitor $M$ for $S$ exists.
  The process $P_1 = \pnil$ is ill-typed for $S$ (since it does not produce any of the expected outputs). 
  By \Cref{def:internal-static-completeness} (Completeness), there must exist a trace $t$ such that $\langle P_1;M \rangle \xRightarrow{t} \langle P'_1;\no{}{P} \rangle$. 
  From the structure of $P_1$ it is clear that $M$ reaches its rejection verdict without interacting with $P_1$, \ie $M \xRightarrow{t} \no{}{P}$. 
  This also means that $M$ would also reach a rejection verdict when instrumented with $P_2 = \psnd{k}.Q'_2$ with $k \in I$ and is well-typed \wrt $S$.  
  This makes $M$ unsound by \Cref{def:internal-static-soundness}. 
\end{description}
Recall that all session types are assumed to be guarded. 
Since the above two cases rule out all the guarding constructs, $\stselNA{i}{I}$ and $\stbraNA{i}{I}$, we conclude that there is no closed (guarded) session type that can be monitored for soundly and completely, except for all the trivial session types that equate to \textsf{end}.%
\end{proof}

\section{Realisability and Implementation}
\label{s:realisabilty}

Up to this point we have considered a level of abstraction that allows us to model session monitors and monitored components, reason about their behaviour, and prove their properties.
We now illustrate how our theoretical developments can be translated into an actual implementation of session monitoring, targeting the Scala programming language.
The key idea is to turn our monitor synthesis procedure (\Cref{def:synth-function})
into a code generation tool that, given a protocol specification (as a session type), %
produces the Scala code of a corresponding executable monitor.
The tool is called \stmonitor, and is provided as companion artifact to this paper.
It is also available at:
\begin{center}
  \url{https://github.com/chrisbartoloburlo/stmonitor} \quad(release tag \texttt{v0.0.1})
\end{center}

\noindent%
We describe \stmonitor in \Cref{sec:implementation} %
--- but first, we augment session types with \emph{assertions} (\Cref{sec:session-type-assertions}).

\subsection{Introducing Assertions in Session Types Specifications}
\label{sec:session-type-assertions}%

\newcommand{\nopdA}{\textsf{no}_P^A}
\newcommand{\noedA}{\textsf{no}_E^A}

Since we use session types as specifications 
for a tool that generates executable monitors,
it is convenient to enrich them with \emph{assertions} on the values being sent or received.
We augment the session types syntax (\Cref{fig:session-types-noa}) by extending selection and branching as follows:

\newcommand{\stbraHL}[2]{\& \big\{?\texttt{l}_{#1}(\hlightlightyellow{$x_{#1}$}:\textsf{B}_{#1})[\hlightlightyellow{$A_{#1}$}].S_{#1}\big\}_{#1 \in #2}}
\newcommand{\stselHL}[2]{\oplus \big\{!\texttt{l}_{#1}(\hlightlightyellow{$x_{#1}$}:\textsf{B}_{#1})[\hlightlightyellow{$A_{#1}$}].S_{#1}\big\}_{#1 \in #2}}

\smallskip\centerline{\(
  S \;\;\Coloneqq\;\; \stselHL{i}{I} \quad\vert\quad \stbraHL{i}{I} \quad\vert\quad \ldots
\)}\smallskip

The assertions $A_i$ are predicates of the process calculus (\Cref{fig:process-calc-syntax-semantics}, \Cref{remark:predicates}),
and they can refer to the named payload variables $x_i$.
Such assertions do not influence type-checking: %
they are copied in the synthesised monitors, where they are used to flag the new violations
$\nopdA$ (assertion violation by the process) and $\noedA$ (external assertion violation).
To achieve this, we update our monitor synthesis function (\Cref{def:synth-function}) as follows:

{\small
\medskip\centerline{
  $\lsem\stsel{i}{I}\rsem\,\triangleq\,\triangleright\big\{\texttt{l}_i(x_i).\pif{\isValueB{i}{x_i}}{\hlightlightyellow{$\left(\pif{A_i}{\envsndOP \texttt{l}_i(x_i).\lsem S_i\rsem}{\nopdA}\right)$}}{\no{}{P}}\big\}_{i\in I}$
}\medskip\centerline{
  $\lsem\stbra{i}{I}\rsem\,\triangleq\,\envrcvOP\big\{\texttt{l}_i(x_i).\pif{\isValueB{i}{x_i}}{\hlightlightyellow{$\left(\pif{A_i}{\triangleleft \texttt{l}_i(x_i).\lsem S_i\rsem}{\noedA}\right)$}}{\no{}{E}}\big\}_{i\in I}$
}\medskip
}%

The only changes are \hlightlightyellow{highlighted}: if the monitored process sends a message that violates the assertion, it is flagged with $\nopdA$;
symmetrically, if a message that violates the assertion is received from the environment, then the message is flagged with $\noedA$.

\begin{example}\label{eg:session-types}\begin{subequations}
  Recall  $S_\textit{auth}$ from \Cref{eg:session-types-noa}. We can refine it with assertions to check the validity of the data being transmitted and received:
  \begin{align}
    \nonumber
    S^{A}_\textit{auth}\ \equal\ &\strec{Y}.!\texttt{Auth}(\hlightlightyellow{$\textit{uname}$}:\textsf{Str},\hlightlightyellow{$\textit{pwd}$}:\textsf{Str})[\hlightlightyellow{$\textit{validUname}(\textit{uname})$}].S^{A}_\textit{res}
    \\
    \nonumber
    S^{A}_\textit{res}\ \equal\ &\& \big\{?\texttt{Succ}(\hlightlightyellow{$\textit{tok}$}:\textsf{Str})[\hlightlightyellow{$\textit{validTok}(\textit{tok}, \textit{uname})$}].S^{A}_\textit{succ},?\texttt{Fail}(\hlightlightyellow{$\textit{code}$}:\textsf{Int})[\hlightlightyellow{$\true$}].Y \big\}
  \end{align}
  In $S^{A}_\textit{auth}$, when the \client sends $\texttt{Auth}(\textit{uname},\textit{pwd})$, 
  the value of \textit{uname} is passed to the predicate $\textit{validUname}$ which ensures that the supplied \textit{uname} is given in the correct format. 
  If the \server replies with $\labk{Succ}(\textit{tok})$, the token \textit{tok} and username \textit{uname} are validated by the cryptographic predicate $\textit{validTok}$,
  which tests whether the token is correct for the given username. 
  If so, the \client continues along session type $S^{A}_\textit{succ}$.
  Otherwise, if the \server chooses to send \labk{Fail} with the error \textit{code}, the trivial assertion check $\true$ is performed.%
  \exqed
\end{subequations}
\end{example} 
Notice that, when all assertions are trivially true, the augmented monitor synthesis is equivalent to the original \Cref{def:synth-function}.
Otherwise, the synthesised monitors with assertions are more restrictive: executions where no violations $\no{}{P}$ nor $\no{}{E}$ were detected
might now violate an assertion and result in $\nopdA$ or $\noedA$.
The introduction of such assertions in our theory changes our monitorability results as follows:
\begin{itemize}
  \item soundness (\Cref{thrm:sound}) is preserved --- which is crucial for practical usability;%
  \item blaming (\Cref{thm:blame-correctnes}) is weakened: an instrumented well-typed process may violate an assertion, and be flagged with $\nopdA$;
  \item weak detection completeness (\Cref{lem:partial-monitor-completeness}) is \emph{not} preserved: %
     assertions can in principle be unsatisfiable,
    hence some ill-typed processes may not be flagged because all their traces end with an environment assertion violation $\noedA$. 
\end{itemize}

\subsection{Implementation}
\label{sec:implementation}

We now illustrate the implementation of our session monitor synthesis tool. It generates runnable Scala code from session types, possibly including the assertions discussed in \Cref{sec:session-type-assertions}.

\subparagraph{Implementation framework.}
Our synthesised monitors uses the session programming library \lchannels \cite{DBLP:conf/ecoop/ScalasY16}.
It allows for implementing a session type $S$ in Scala, by
\begin{enumerate}%
  \item defining a set of \emph{Continuation-Passing-Style Protocol classes} (\textsf{CPSPc}) corresponding to $S$, and
  \item using a communication API that, by leveraging such \textsf{CPSPc}, lets the Scala compiler spot protocol violations.
\end{enumerate}
By using \lchannels, we are more confident that if a syntesised monitor for session type $S$ compiles, then it correctly sends/receives messages according to $S$.
Moreover, \lchannels abstracts communication from the underlying message transport, hence it allows our monitors to interact with clients or servers written in any programming language.

\subparagraph{Implementation of the session monitor synthesis.}
Overall, our Scala monitor generation requires 3 user-supplied inputs:
\begin{enumerate}[{i}1]%
  \item\label{item:synth:type} a session type $S$ (with or without assertions) describing the desired protocol;
  \item\label{item:synth:func} for each assertion in $S$ (if any), a corresponding Scala function returning true/false; and
  \item\label{item:synth:cm} a \emph{Connection Manager} class (discussed below) to interact with the monitored process.
\end{enumerate}
Given a session type (input \ref{item:synth:type}) our monitor synthesiser tool generates: 
\begin{enumerate}
  \item the protocol classes (\textsf{CPSPc}) for representing the session type in Scala + \lchannels, and
  \item the Scala source code of a runtime monitor (requiring inputs \ref{item:synth:func} and \ref{item:synth:cm} to compile).
\end{enumerate}

\newcommand{\sauthA}{S^{A}_\textit{auth}}
\newcommand{\sresA}{S^{A}_\textit{res}}

The generated monitor acts as a mediator between client and server: one is on the \emph{internal} side of the monitor (\ie the instrumented process), while the other is on the external side. 
The internal side is untrusted: its messages are run-time checked, to ensure they follow the desired protocol (\eg session type $\sauthA$ in \Cref{eg:session-types}). 
Instead, the \emph{external} side is trusted: it is (mostly) expected to follow the dual protocol (\eg the dual session type $\overline{\sauthA}$). 
This design choice allows us to simplify the monitor implementation, as its communication with the external side are handled by \lchannels.
However, our design does not limit the flexibility of the approach, since an untrusted peer can be made trusted by instrumenting it with a monitor (see discussion below).

\subparagraph{Monitor synthesis in practice.}
We now illustrate the scenario depicted in \Cref{fig:implementation-design} where: 

\begin{enumerate}
\item we want a client/server system to implement the session type (with assertions) $\sauthA$ (\Cref{eg:session-types});
\item we trust the $\server$ (\eg because it is type-checked), and %
\item we want to instrument a $\client$ whose source code is inaccessible or cannot be verified.
\end{enumerate}
Other variations of this scenario are possible. 
For instance, we could similarly instrument an untrusted server, by running our monitor synthesiser on the dual session type $\overline{\sauth}$.
The resulting combination of monitor-and-server is then trusted
and can interact via \lchannels.
As a result, it could then be used as the trusted \server in 
\Cref{fig:implementation-design}.
 
\begin{wrapfigure}[6]{r}{6cm}
  \captionsetup{justification=centering,font=scriptsize}
  \centering
  \vspace{-0.4cm}%
\centering
\begin{tikzpicture}[decoration=penciline]
  \node(client) [process, fill=gray!20] {$\client$};
  \node(cm) [draw, right=0.5cm of client, rounded corners, font=\scriptsize] {\textsf{CM}};
  \node(monitor) [process, right=0.2cm of cm]{$\textit{mon}$};
  \node(server) [process, right=1.5cm of monitor] {$\server$};
  
  \draw [thin, dotted, rounded corners] ($(client.north west)+(-0.1cm,0.17cm)$) rectangle ($(monitor.south east)+(0.1cm,-0.15cm)$);
  
  \draw[thin,double distance=2pt, fill=white] (cm) --node[above, font=\scriptsize, yshift=-1pt]{$\ast$} (client);
  \draw[thin,double distance=2pt, fill=white] (monitor) -- node(lchannels)[above, font=\scriptsize, yshift=-1pt]{\lchannels} (server);
  \node(cpspc) [draw, above=0.5cm of lchannels, rounded corners, font=\footnotesize] {$\textsf{CPSPc}$};
  \node(synth) [draw, left=0.6cm of cpspc, rounded corners, font=\footnotesize] {$\textsf{synth}$};
  \node(sauth) [decorate, draw, left=0.7cm of synth] {$S_\textit{auth}$};

  \node(cpspc) [draw, above=0.5cm of lchannels, rounded corners, font=\footnotesize] {$\textsf{CPSPc}$};

  \draw[-] (cm) -- (monitor);

  \draw[-angle 90,dashed] (sauth) edge (synth);
  \draw[-angle 90,dashed,bend left=20] (synth) edge (monitor);
  \draw[-angle 90,dashed] (synth) edge (cpspc);

\end{tikzpicture}
   \vspace{-0.5cm}%
  \caption[The composite system interacting with the environment.]{}\label{fig:implementation-design}
\end{wrapfigure}

The generated monitor (\textit{mon}) intercepts all messages between \client and \server. 
The communication between \textit{mon} and $\server$ occurs via \lchannels; instead, the communication between the monitor and the client is handled by a \emph{Connection Manager} (\CM):
a user-supplied Scala class, input~\ref{item:synth:cm}, which acts as a \emph{translator} and \emph{gatekeeper}, by transforming each messages from the monitor-client transport protocol into a corresponding CPSP class, and \emph{vice versa}.
With this design, the code generated for the monitor is abstracted from the low-level details of the protocols used by both the $\client$ and $\server$.

There is a tight correspondence between the monitors generated by our tool,
and our formal monitor synthesis. %
This increases our confidence that the results in \Cref{s:formal-results} carry over to our implementation and that our tool is indeed correct. 
In the sequel, we illustrate the generated monitoring code for \Cref{eg:session-types} above, showing the monitoring of a selection type (\Cref{lst:srv-receive}) and branching type (\Cref{lst:srv-send}).

\begin{figure}[h]
  \begin{subfigure}[l]{0.4\textwidth}\small
    \begin{flalign*}
      &\lsem \sauthA\rsem\ \equal\ &\\
      &\quad\pmu{Y}.\triangleright \big\{ \texttt{Auth}(\textit{uname}:\textsf{Str},\textit{pwd}:\textsf{Str}).&\\
      &\quad\textsf{if }\isValueB{\textsf{Str}}{uname} \wedge \isValueB{\textsf{Str}}{pwd}&\\
      &\;\quad\textsf{then if }\textit{validUname}(\textit{uname})&\\
      &\quad\quad\textsf{then }\envsndOP \texttt{Auth}(\textit{uname},\textit{pwd}).\lsem S_\textit{res}\rsem&\\
      &\quad\quad\textsf{else }\nopdA&\\
      &\;\quad\textsf{else }\no{}{P}\big\}&
    \end{flalign*}
    \vspace{0.4cm}
  \end{subfigure}
  \hfill
  \begin{subfigure}{0.58\textwidth}
\begin{lstlisting}[label={lst:srv-receive},captionpos=b,aboveskip=3mm,numbers=right,frame=tbl,rulecolor=\color{black},firstnumber=1, xrightmargin=1.5ex]
def receiveAuth(srv:Out[Auth],client:CM): Unit ={
  client.receive() match {
    case msg @ Auth(_, _) =>
      if (validUname(msg.uname)) {
        val cont = srv !! Auth(msg.uname,
                                  msg.pwd)_
        payloads.Auth.uname = msg.uname
        sendChoice1(msg.cont, client)
      } else { 
        /* INTERNAL VIOLATION (assertion) */ 
      }
    case _ =>  
      /* INTERNAL VIOLATION: invalid message */
} }
\end{lstlisting}
  \end{subfigure}
  \vspace{-3mm}
  \caption{\centering Comparison between the formal and implementation synthesis of the internal choice.}
\end{figure}

The internal receive operator of the monitor calculus ($\triangleright$) corresponds to line 2 in \Cref{lst:srv-receive}, where the monitor invokes the \texttt{receive} method of the \CM. %
Depending on the type of message received, the monitor performs a series of checks. By default, a catch-all case (line 12) handles any messages violating the protocol:
this is similar to rule \ltsrule{mIV} of the formal monitor (\Cref{fig:monitor-calculus}),
which flags the violation $\no{}{P}$.
If \texttt{Auth} is received, the monitor initially invokes the function \texttt{validUname()} with argument \texttt{uname}; 
such a function is user-supplied (see input \ref{item:synth:func} above). 
If the function returns \texttt{false}, the monitor flags the violation and halts (line 10):
this corresponds to the external assertion violation $\nopdA$ in $\lsem \sauthA\rsem$.
Otherwise, if \texttt{validUname()} returns \texttt{true}, the message is forwarded to the \server (line 5). 
The function used to forward the message ({\small\texttt{!!}}), which is part of \lchannels, corresponds to the external output operator $\envsndOP$ of $\lsem \sauthA\rsem$;
it returns a continuation channel that is stored in \texttt{cont}. 
To associate the payload identifiers of $\sauthA$ to their current values, the monitors maintain a mapping, called \texttt{payloads}. 
In this case, the value of \texttt{uname} is stored (line 7) since it is used later on in $\sauth$. 
Finally, the monitor moves to the next state %
\texttt{sendChoice1} (\Cref{lst:srv-send}), passing the channel stored in \texttt{cont} to continue the protocol (line 8).
\begin{figure}[h]
  \begin{subfigure}[l]{0.37\textwidth}\small
    \vspace{-0.3cm}
      \begin{flalign*}
      &\lsem \sresA \rsem \equal&\\
      &\quad\envrcvOP \big\{ \texttt{Succ}(\textit{tok}:\textsf{Str}).\textsf{if }\isValueB{\textsf{Str}}{tok}&\\
      &\quad\qquad\textsf{then if }\textit{validTok}(\textit{tok},\textit{uname})&\\
      &\;\quad\qquad\textsf{then }\triangleleft \texttt{Succ}(\textit{tok}).\lsem S_\textit{succ} \rsem&\\
      &\;\quad\qquad\textsf{else } \noedA \textsf{ else }\no{}{E},&\\
      &\qquad\texttt{Fail}(\textit{code}:\textsf{Int}).\textsf{if }\isValueB{\textsf{Int}}{\textit{code}}&\\
      &\;\qquad\textsf{then if tt then }\triangleleft \texttt{Fail}(\textit{code}).Y&\\
      &\;\qquad\textsf{else } \noedA \textsf{ else }\no{}{E}\big\}&
    \end{flalign*}
  \end{subfigure}
  \hfill
  \begin{subfigure}{0.6\textwidth}
\begin{lstlisting}[label={lst:srv-send},captionpos=b,aboveskip=3mm,numbers=right,frame=tbl,rulecolor=\color{black},firstnumber=1, xrightmargin=1.5ex]
def sendChoice1(srv:In[Choice1],Client:CM):Any = {
  srv ? {
    case msg @ Succ(_) =>
      if (validTok(msg.tok, payloads.Auth.uname)) {
        Client.send(msg)
        /* Continue according to S_succ */
      } else {
        /* EXTERNAL VIOLATION (assertion) */
      }
    case msg @ Fail(_) =>
      Client.send(msg)
      receiveAuth(msg.cont, External)
} }
\end{lstlisting}
  \end{subfigure}
  \vspace{-3mm}
  \caption{\centering Comparison between the formal and implementation synthesis of the external choice.}
\end{figure}

According to $\sresA$, the server can choose to send either \texttt{Succ} or \texttt{Fail}. 
The monitor waits to receive either of the options from the \server, using the method {\small\texttt{?}} from \lchannels (line 2).
This corresponds to the external input operator of the monitor calculus ($\envrcvOP$) used in $\lsem \sresA \rsem$, which can also receive both options from the \server. 
\begin{itemize}
  \item If the \server sends $\texttt{Succ}(\mathit{toc})$, %
  the first case is selected (line 3). 
  The monitor evaluates the assertion \texttt{validTok} on $\textit{tok}$ and $\textit{uname}$ (stored in \Cref{lst:srv-receive}, and now retrieved from the \texttt{payloads} mapping). 
  If it is satisfied, the message is forwarded to the client (line 5) via \CM's \texttt{send} method, %
  which corresponds to the internal send operator ($\triangleleft$) in the monitor calculus. 
  The monitor then proceeds according to the monitor $\lsem S_\textit{succ} \rsem$. 
  Otherwise, the monitor logs a violation and halts (line 8);
  similarly, $\lsem\sresA\rsem$ flags the violation $\noedA$ indicating an external assertion violation. 

  \item Instead, if the \server sends \texttt{Fail} (line 10), the monitor forwards it to the client; %
  there are no assertion checks here, as the assertion after $\texttt{Fail}$ in $\lsem\sresA\rsem$ is \textsf{tt}. 
  Then, following the recursion in $\lsem\sresA\rsem$,
  the monitor (on line 12) loops to \texttt{receiveAuth} (\Cref{lst:srv-receive}). 
\end{itemize}
Unlike the synthesised code of \texttt{receiveAuth} (that handles the previous external choice, in \Cref{lst:srv-receive}), there is no catch-all case for unexpected messages from the \server. This is by design. As explained above we use \lchannels to interact with the ``trusted'' external side, hence the interaction with the server is typed, and a catch-all case would be unreachable code.
Still, \lchannels throws an exception (crashing the monitor) if it receives an invalid message
--- which corresponds to the monitor $\lsem\sresA\rsem$ flagging an external violation via rule \ltsrule{mEV}.
\section{Empirical Evaluation}
\label{s:implementation}

We evaluate the feasibility of our implementation by measuring the overheads induced by the run-time checks of our synthesised monitors (\Cref{s:realisabilty}). %
We consider 3 application protocols, modelled as session types, as our benchmarks:
\begin{enumerate}
  \item A ping-pong protocol, based on a request-response over HTTP (a style of protocol that is typical, \eg in applications based on web services).
  Although it is a fairly simple protocol, our implementation uses HTTP to carry ping/pong messages, highlighting the fact that our generated monitors are independent from the message transport in use;
  \item A fragment of the Simple Mail Transfer Protocol (SMTP) \cite{SMTP}. This benchmark represents  a more complex protocol featuring nested internal/external choices;
  \item A fragment of the HTTP protocol, also featuring sequences of nested internal/external choices.
\end{enumerate}

\subparagraph{Ping-pong over HTTP.} 
In this protocol, a client is expected to recursively send messages with label $\labk{Ping}$ to the server which, in turn, replies with $\labk{Pong}$.
The protocol proceeds until the client sends $\labk{Quit}$. The client-side protocol is shown below (the server-side is dual).

\smallskip\centerline{\(
  S_\textit{pong}\ \equal\ \strec{X}.(\oplus \big\{ !\labk{Ping}().?\labk{Pong}().X,\ !\labk{Quit}()\big\})\nonumber
\)}\smallskip

Notice that the protocol has no explicit reference to HTTP. In fact, we use HTTP as a mere message transport, by providing a suitable Connection Manager to the synthesised monitor (which is transport-agnostic). Concretely, the ping-pong is implemented with the server handling requests on an URL like\; {\small\texttt{http://127.0.0.1/ping}}, \;and the client performing a {\small\texttt{GET}} request on that URL, and reading the response.
For this benchmark, the setup is:
\begin{itemize}%
  \item the client is on the internal side of the generated monitor, hence subject to scrutiny;
  \item the server is on the external side of the generated monitor.
\end{itemize}
As untrusted client we use a standard, unmodified load testing tool: Apache JMeter ({\small\url{https://jmeter.apache.org/}})
configured to send HTTP requests at an increasing rate.%

\subparagraph{SMTP.}
We model a fragment of the SMTP protocol (server-side) as the session type $S_\textit{smtp}$:
{\small
\begin{align}
  S_\textit{smtp}\ \equal\ & !\labk{M220}(\textit{msg}: \textsf{Str}).\&\big\{?\labk{Helo}(\textit{host}: \textsf{Str}).!\labk{M250}(\textit{msg}: \textsf{Str}).S_\textit{mail}, ?\labk{Quit}().!\labk{M221}(\textit{msg}: \textsf{Str}) \big\}\nonumber\\
  \label{s-mail-1}S_\textit{mail}\ \equal\ & \strec{X}.(\&\big\{?\labk{MailFrom}(\textit{addr}: \textsf{Str}).!\labk{M250}(\textit{msg}: \textsf{Str}). \strec{Y}.(\&\big\{\\
  \label{s-mail-2}&\qquad?\labk{RcptTo}(\textit{addr}: \textsf{Str}).!\labk{M250}(\textit{msg}: \textsf{Str}).Y,\\
  \label{s-mail-3}&\qquad?\labk{Data}().!\labk{M354}(\textit{msg}: \textsf{Str}).?\labk{Content}(\textit{txt}: \textsf{Str}).!\labk{M250}(\textit{msg}: \textsf{Str}).X,\\
  &\qquad?\labk{Quit}().!\labk{M221}(\textit{msg}: \textsf{Str})\big\}), ?\labk{Quit}().!\labk{M221}(\textit{msg}: \textsf{Str})\big\})\nonumber
\end{align}
}%
When a client establishes a connection, the server sends a welcome message (\texttt{M220}), and waits for the client to identify itself (\texttt{Helo}). 
Then, the client can recursively send emails by specifying the sender and recipient address(es), followed by the mail contents. 
The client can send multiple emails by repeating the loop on ``$X$'' between lines \eqref{s-mail-1} and \eqref{s-mail-3}.

The SMTP protocol runs over TCP/IP. The specification above (and the synthesised monitors) are again transport-agnostic: we handle TCP/IP sockets by providing a suitable Connection Manager to the synthesised monitor.

For this benchmark, the setup used is ``dual'' to that of the HTTP ping-pong benchmark above:
\begin{itemize}%
  \item the server is on the internal side of the generated monitor, hence subject to scrutiny;
  \item the client is on the external side of the generated monitor.
\end{itemize}
For this experiment, we implement an SMTP client that sends emails to the server, and measures the response time.
We take such measurements against two (untrusted and monitored) servers, both configured
to accept incoming emails and discard them:
\begin{enumerate}
  \item a default instance of \texttt{smtpd} from the Python standard library;\footnote{\url{https://docs.python.org/3/library/smtpd.html}}
  \item a default instance of Postfix,\footnote{\url{http://www.postfix.org/}} one of the most used SMTP servers \cite{SMTPSurvey}.
\end{enumerate}

\subparagraph{HTTP.}
In this benchmark,  we do \emph{not} use HTTP as a mere message transport (unlike the ping-pong benchmark above). 
Rather, we model HTTP headers, requests, and responses with a session type, which we use to synthesise a monitor that checks the interactions between a trusted server and an untrusted client. 
We focus on a fragment of HTTP that is sufficient for supporting typical client-server interactions (\eg when the client is the Mozilla Firefox browser). 
The HTTP session type (here omitted due space reasons) and its (trusted) server implementation are adapted from the \lchannels examples \cite{DBLP:journals/darts/ScalasDHY17}.
For benchmarking, we use Apache JMeter ({\small\url{https://jmeter.apache.org/}}) as untrusted client.

\subparagraph{Benchmarking setups and measurements.}
In all of our benchmarks, we study the overhead of our synthesised monitors by comparing:
\begin{itemize}%
  \item an \emph{unsafe} setup: the client and server interact directly;
  \item a \emph{monitored} setup: communication between the trusted and untrusted components is mediated by our synthesised monitors, which halts when it detects a violation --- as described earlier in \Cref{fig:implementation-design}.
\end{itemize}

We follow a multi-faceted approach, as advocated by \cite{DBLP:conf/fase/AcetoAFI21}, and base our study on three measurements: average response time, average CPU utilisation, and maximum memory consumption.
The response time is arguably the most important measurement, since slower response times can be immediately perceived when interacting with a monitored system.
We measure them by running experiments of increasing length: 
for ping-pong and HTTP, we perform an increasing number of request-response loops, whereas
for SMTP, we send an increasing number of emails.
The general expectation is: for longer experiments, the average response time and CPU usage should decrease, while the maximum memory consumption should increase.
We repeat each experiment 30 times, and we plot the average of all results.

In our benchmarks, overheads can have two forms:

\begin{enumerate}[{Overhead} 1:]%
  \item\label{item:overhead:msg} the translation and duplication of messages being forwarded between client and server;
  \item\label{item:overhead:checks} the run-time checks needed to ensure that the desired session type is being respected.
\end{enumerate}

Overhead \ref{item:overhead:msg} is unavoidable for the most part. 
By their own nature, partial identity monitors (like ours) must receive and forward all messages. 
This overhead can only be minimised by using more efficient message transports.
By contrast, overhead \ref{item:overhead:checks} is specifically caused by our monitor synthesis.
Our benchmarks were specifically designed to accurately capture this latter form of overhead.
In order to better distinguish overhead \ref{item:overhead:msg} from overhead \ref{item:overhead:checks}, our benchmarks run the trusted side (client or server) and the synthesised monitors on a same JVM instance, where they interact in the most efficient way (\ie through the \texttt{LocalChannel} transport provided by the \lchannels library). 
This minimises overhead \ref{item:overhead:msg}, and allows us to better observe the impact of overhead \ref{item:overhead:checks}. Clearly, the untrusted side of each benchmark (\ie the black-box client or server being monitored) always runs as a separate process.

Despite this, our synthesised monitors can still be deployed independently of the trusted side (\ie on their own JVM, possibly across a network) because they are agnostic to the message transports in use; 
this is made possible by the use of connection managers and \lchannels. We demonstrate this capability by also taking measurements for
  a \emph{detached} setup, where the trusted component and monitor run on separate JVMs (on a same host), and interact via TCP/IP (through a suitable message transport for \lchannels).
  This setup is more flexible, but the slower message transport increases overhead \ref{item:overhead:msg}. We implemented this setup for ping-pong and SMTP, measuring their response times.

\subparagraph{Results and analysis.}

The benchmark results are reported \Cref{fig:benchmarks}.
\begin{figure}
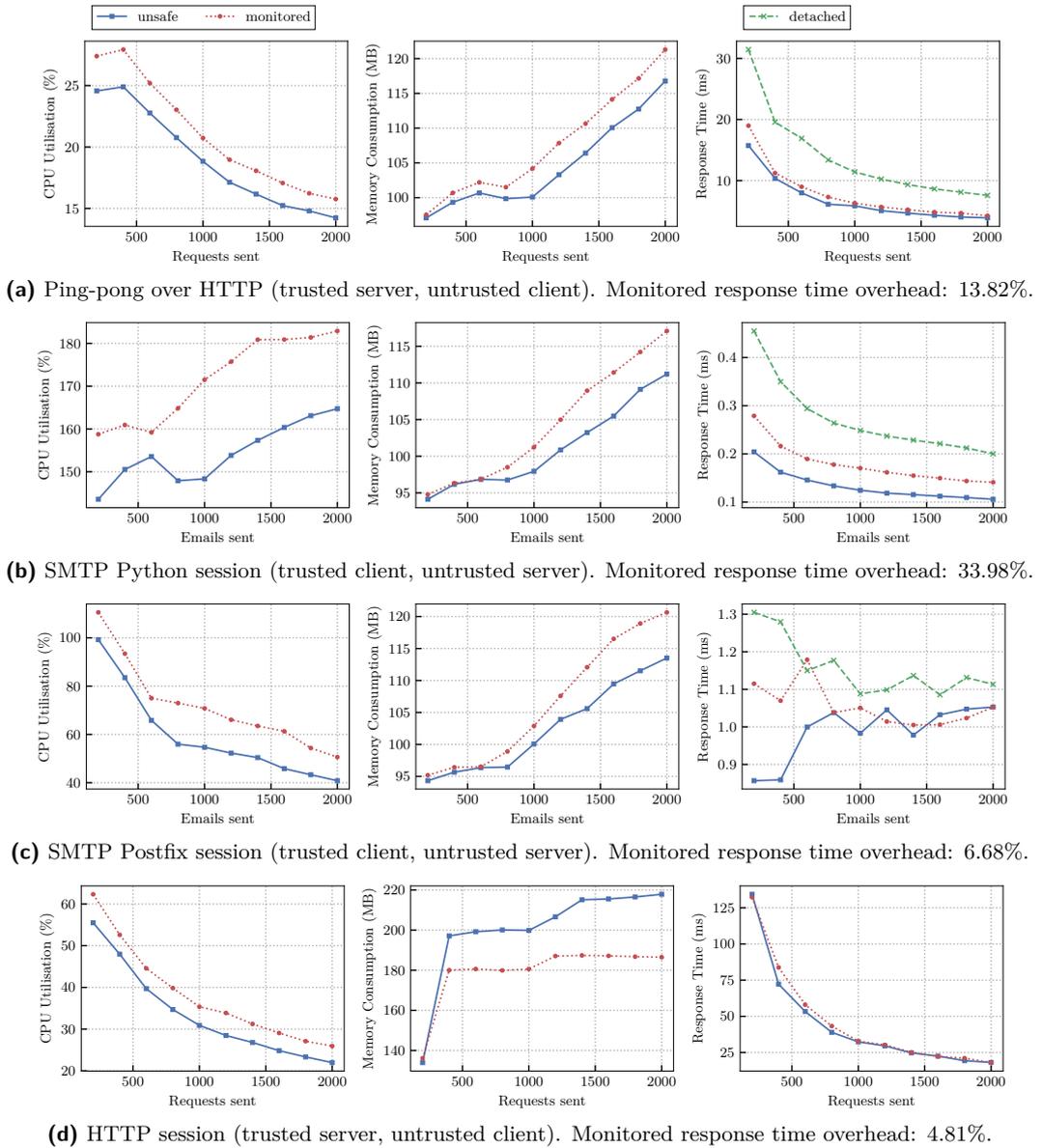

  \newcommand{\plotScale}{0.65}
  \centering
  \begin{subfigure}{\linewidth}
    \centering
    \scalebox{\plotScale}{
\begingroup%
\makeatletter%
%
\makeatother%
\endgroup%
     }
    \caption{\centering Ping-pong over HTTP (trusted server, untrusted client). Monitored response time overhead: $13.82\%$.}
    \label{fig:pingpong-benchmarks}
  \end{subfigure}
  
  \medskip

  \begin{subfigure}{\linewidth}
    \centering
    \scalebox{\plotScale}{
\begingroup%
\makeatletter%
%
%
\makeatother%
\endgroup%
     }
    \caption{\centering SMTP Python session (trusted client, untrusted server). Monitored response time overhead: $33.98\%$.}
    \label{fig:smtp-benchmarks}
  \end{subfigure}

  \medskip

  \begin{subfigure}{\linewidth}
    \centering
    \scalebox{\plotScale}{
\begingroup%
\makeatletter%
%
%
\makeatother%
\endgroup%
     }
    \caption{\centering SMTP Postfix session (trusted client, untrusted server). Monitored response time overhead: $6.68\%$.}
    \label{fig:smtp-postfix-benchmarks}
  \end{subfigure}

  \medskip

  \begin{subfigure}{\linewidth}
    \centering
    \scalebox{\plotScale}{
\begingroup%
\makeatletter%
%
%
\makeatother%
\endgroup%
     }
    \caption{\centering HTTP session (trusted server, untrusted client). Monitored response time overhead: $4.81\%$.}
    \label{fig:http-benchmarks}
  \end{subfigure}

  \caption{Benchmark results: average CPU usage, maximum memory consumption, and average response time.
    {\footnotesize(30 runs, 2 CPUs (Intel Pentium Gold G5400 @ 3.70GHz), 8 GB RAM, Ubuntu 20.04)}}
  \label{fig:benchmarks}
\end{figure}
For the ping-pong benchmark (\Cref{fig:pingpong-benchmarks}), the impact of monitors is noticeable but limited: for the ``monitored'' setup (which highlights overhead \ref{item:overhead:checks}), the response times are less than 14\% slower;
the ``detached'' monitor setup is unsurprisingly slower, due to its slower message transport (which increases overhead \ref{item:overhead:msg}).
For the SMTP benchmark (\Cref{fig:smtp-benchmarks,fig:smtp-postfix-benchmarks}), we can observe different behaviours:
\begin{itemize}
\item the Python \texttt{smtpd} server (\Cref{fig:smtp-benchmarks}) has extremely fast response times: it is essentially
  a dummy server that receives emails and does nothing with them. This is also evident from the CPU usage: it constantly increases, because the SMTP client receives immediate responses, no matter how many emails it sends, with or without a monitor.
  Consequently, our monitors cause a relatively high impact on such fast response times (almost 34\%);
\item the Postfix SMTP server (\Cref{fig:smtp-postfix-benchmarks}) is more realistic: unlike Python \texttt{smtpd}, it takes some time (with fluctuations) to process each email and respond to the client. Consequently, our monitors have a relatively small impact on the response times (less than 7\%).
\end{itemize}
As in the case of ping-pong, the ``detached'' monitor setup for both SMTP benchmarks is slower, as it uses a slower message transport (which increases overhead \ref{item:overhead:msg}).
Finally, the HTTP benchmark (\Cref{fig:http-benchmarks}) shows a response time overhead that is below 5\%.
By and large, these overhead levels are tolerable for many applications that are not mission critical, and are comparable to the overhead experienced when running state-of-the-art RV tools~\cite{DBLP:series/lncs/BartocciFFR18}.

\section{Conclusion}
\label{s:conclusion}

We presented a formal analysis for the \emph{monitorability limits} of (binary) session types \wrt a partial-identity monitor model; to wit, this is the first monitorability assessment of session types.
We couple this study %
with an implementation of session monitor synthesis.%

More in detail, our contributions are the following.
On the the theoretical side, we provide the first treatment of the \emph{monitorability} of session types,
and \emph{detection-soundness} and \emph{detection-completeness} properties of session monitors,
and we prove that our autogenerated session monitors enjoy both the former and (to a lesser extent) the latter.
We also present an impossibility result of completeness for our black-box monitoring setup --- which is a novel result to the area of session type monitoring. 
On the practical side, we evaluate the viability of our implementation (called \stmonitor) via benchmarks. 
The results show that our monitor synthesis procedure only introduces limited overheads.%

\subsection{Related Work}

Several papers address the monitoring of session-types-based protocols
--- but no previous work studies the formal problem of session monitorability;
furthermore, their approaches differ from ours in various ways, as we now discuss.

The work \cite{DBLP:journals/tcs/BocchiCDHY17} %
formalises a theory of process networks including monitors generated from (multiparty) session types. 
The main differences with our work are:
\begin{enumerate}%
\item%
  the design of \cite{DBLP:journals/tcs/BocchiCDHY17} is based on a global, centralised router %
  providing a \emph{safe transport network} %
  that dispatches messages between participant processes; %
  correspondingly, its implementation \cite{DBLP:journals/fmsd/DemangeonHHNY15,DBLP:conf/rv/NeykovaYH13,DBLP:conf/rv/HuNYDH13}
  includes a Python library for monitored processes to access the safe transport network.
  By contrast, we do not assume a specific message routing system,
  and our theory and implementation address the monitoring of black-box components;
\item%
  the results in \cite{DBLP:journals/tcs/BocchiCDHY17}
  do not consider limits related to session monitorability.
  Their results (\eg \emph{transparency}) are analogous to our \emph{detection soundness}
  (\Cref{thrm:sound}), \ie synthesised monitors do not disrupt communications of well-typed processes;
  they do not address \emph{completeness}
  (\Cref{lem:partial-monitor-completeness} and \Cref{thm:impossibility}), \ie to what extent can a monitor detect ill-typed processes.
\end{enumerate}
Furthermore, our work and \cite{DBLP:journals/tcs/BocchiCDHY17} differ in a fundamental design choice:
when our monitors detect an invalid message, %
they flag a violation and halt --- %
whereas monitors in \cite{DBLP:journals/tcs/BocchiCDHY17} drop invalid messages, %
and keep forwarding the rest.
The latter is akin to \emph{runtime enforcement via suppressions}~\cite{DBLP:conf/concur/AcetoCFI18};
studying this design with our theory is interesting future work.

Our protocol assertions (\Cref{sec:session-type-assertions}) are reminiscent of %
\emph{interaction refinements} in \cite{Neykova18CC}, %
that are also statically generated (by an F\# type provider), %
and dynamically enforced when messages are sent/received. 
However, our approach and design are different from \cite{Neykova18CC}: 
we synthesise session monitoring processes %
that can be deployed over a network, to instrument black-box processes %
--- whereas \cite{Neykova18CC} %
expects the runtime-verified code to be written with a specific language and framework,
and injects dynamic checks in the program executable. 
Furthermore, the work \cite{Neykova18CC} %
does not address session monitorability limits.

The work \cite{DBLP:conf/cc/NeykovaY17} proposes a methodology to supervise (multiparty) session protocols, and recover them in case of failure of some component; it also includes an implementation in Erlang.
Similarly to this work, in \cite{DBLP:conf/cc/NeykovaY17} each component is observed by a session monitor; unlike this work, \cite{DBLP:conf/cc/NeykovaY17} does not address any aspect of session monitorability, and focuses on proving that its recovery strategy does not deadlock.

The work by Gommerstadt \etal \cite{DBLP:conf/esop/GommerstadtJP18} considers a partial identity monitor model for session types that is close to the one discussed in \Cref{s:designing-hybrid-methodology}.
They however do not provide any synthesis function and assume that monitors are constructed by hand.
To complement this, they define a dedicated type system to prove that the monitor code behaves as a partial identity, \eg it forwards messages in the correct order, without dropping them. 
They do not study session monitorability. 
To our knowledge, their approach has not been implemented as a tool nor has it been assessed empirically either.

Melgratti and Padovani~\cite{DBLP:journals/pacmpl/MelgrattiP17} propose monitors that act as wrappers around a session library. 
This technique effectively \emph{inlines} the monitors  in the monitored process code. 
In fact, their implementation assumes that the processes under scrutiny are written in OCaml using the FuSE library. 
In contrast, we synthesise \emph{outline} monitors as independent processes that observe black-box implementations written in \emph{any} language/library.
The work proves a series of results that are akin to our notion of monitoring soundness, without addressing completeness.

In separate work, Waye \etal\cite{DBLP:journals/pacmpl/WayeCD17} monitor black-box services, focusing exclusively on \emph{request-response} protocols. 
Unlike our session-type monitors, they do not support protocols with prescribed sequences of internal/external choices and recursion. 
In fact, their contracts are analogous to enhanced assertions on transmitted/received values (reminiscent of the assertion introduced in \Cref{sec:session-type-assertions}).
Although they provide soundness results for their monitoring framework, they do not consider any further monitorability issues.

The recent work \cite{DBLP:conf/tacas/HamersJ20} presents a runtime verification framework for communication protocols (based on multiparty session types) in Clojure. Unlike this work, \cite{DBLP:conf/tacas/HamersJ20} expects monitored applications to be written in a specific language and framework --- whereas we address the monitoring of black-box processes.
Again, \cite{DBLP:conf/tacas/HamersJ20} does not study session monitorability.%

\subsection{Future Work}

This work is our first step along a new line of research on the relative power of static versus run-time verification methods.
In general terms, given a calculus $C$ with a type system $T$ and run-time monitoring system $M$,
\emph{monitoring soundness} tells us whether $M$ is flagging ``real'' errors according to $T$.
Dually, \emph{monitoring completeness} tells us whether $T$ is too restrictive \wrt $M$
(\ie whether $T$ is rejecting too many processes that $M$ deems well-behaved).
In this work, we demonstrate a rather tight connection
between the chosen process calculus ($C$) and session type system ($T$), and our session monitors ($M$):
our synthesised monitors are sound (\Cref{thrm:sound}), and most processes rejected by the type system behave incorrectly (\Cref{lem:partial-monitor-completeness}).
Our plan is to study more instances of $C$, $T$ and $M$ --- both in theory, and in practice.

One avenue worth exploring is that of increasing the observational powers of the monitoring setup considered, in order to extend session monitorability.
The work by Aceto \etal~\cite{DBLP:conf/fossacs/AcetoAFI18} is a systematic study that considers a variety of extensions to the traditional monitoring setup (consisting of one monitor observing events describing the computation effected by the process under scrutiny). 
The extensions considered include traces that report process termination and events that could not have been produced at different stages of the computation (\ie refusals~\cite{DBLP:journals/tcs/Phillips87}).
They also consider monitoring setups where a process is monitored over multiple runs.
In each case, they show the maximal properties that can be monitored for in a sound and complete manner, characterised a syntactic fragments of the modal $\mu$-calculus.
We intend to consider how any of the proposed extensions would affect our monitorability results and the extent to which they are implementable in practice.
Other bodies of work take a slightly different approach to monitorability, by weakening the completeness requirement from their notion of adequate monitoring~\cite{DBLP:conf/csl/AcetoAFIL21,AcetoAFIL21:SOSYM}.
It would be worthwhile exploring the effect of having such weakened completeness requirements on the monitorability of session types.  

Although we have limited ourselves to binary session, we plan to extend the framework above
to the static and run-time verification of multiparty and asynchronous sessions~\cite{HYC08,HYC16}. 
This will most likely require us to consider communicating monitors, that cooperate to aggregate observations made from analysing communications on distinct channels. 
For multiparty sessions, we can benefit from previous work~\cite{DBLP:conf/ecoop/ScalasDHY17,DBLP:journals/darts/ScalasDHY17} %
where \lchannels is used to implement multiparty protocols written in Scribble %
\cite{scribble,YHNN2013}.
Our implementations should also benefit from insights gained from numerous work on decentralised runtime verification~\cite{DBLP:journals/fmsd/0002F16,DBLP:conf/sefm/AttardF17,DBLP:journals/fmsd/FinkbeinerHST19}.
For both multiparty and asynchronous sessions, we can benefit from the research on precise session subtyping~\cite{GhilezanPPSY21,GhilezanJPSY19,DBLP:journals/lmcs/ChenDSY17}.

In this work, our session monitors adhere to the \emph{``fail-fast''} design methodology: %
if a protocol violation occurs, the monitor flags the violation and halts.  
In the practice of distributed systems, %
``fail-fast'' is advocated as an alternative to defensive programming~\cite{erlangCesarini2009}; %
it is also in line with existing literature on runtime verification~\cite{DBLP:series/lncs/BartocciFFR18}. %
As mentioned above, an interesting research direction is to adapt our session monitorability framework
to \emph{suppressions}~\cite{DBLP:conf/concur/AcetoCFI18}, \ie by dropping invalid messages without halting the monitor,
as in \cite{DBLP:journals/tcs/BocchiCDHY17}.
Finally, we plan to investigate how to handle violations by adding \emph{compensations} to our session types
--- \ie by formalising how the protocol should proceed if a violation is detected at a certain stage.
In this setting, the monitors would play a more active role in handling violations, and their synthesis would need to be more sophisticated;
this new research could be related to the work on session recovery \cite{DBLP:conf/cc/NeykovaY17}.
 
\bibliography{refs-condensed}

\begin{thebibliography}{10}

\bibitem{DBLP:journals/pacmpl/AcetoAFIL19}
Luca Aceto, Antonis Achilleos, Adrian Francalanza, Anna
  Ing{\'{o}}lfsd{\'{o}}ttir, and Karoliina Lehtinen.
\newblock Adventures in monitorability: from branching to linear time and back
  again.
\newblock {\em Proc. {ACM} Program. Lang.}, 3({POPL}):52:1--52:29, 2019.
\newblock \href {https://doi.org/10.1145/3290365} {\path{doi:10.1145/3290365}}.

\bibitem{DBLP:conf/sefm/AcetoAFIL19}
Luca Aceto, Antonis Achilleos, Adrian Francalanza, Anna
  Ing{\'{o}}lfsd{\'{o}}ttir, and Karoliina Lehtinen.
\newblock An operational guide to monitorability.
\newblock In Peter~Csaba {\"{O}}lveczky and Gwen Sala{\"{u}}n, editors, {\em
  Software Engineering and Formal Methods - 17th International Conference,
  {SEFM} 2019, Oslo, Norway, September 18-20, 2019, Proceedings}, volume 11724
  of {\em Lecture Notes in Computer Science}, pages 433--453. Springer, 2019.
\newblock \href {https://doi.org/10.1007/978-3-030-30446-1\_23}
  {\path{doi:10.1007/978-3-030-30446-1\_23}}.

\bibitem{AcetoAFIL21:SOSYM}
Luca Aceto, Antonis Achilleos, Adrian Francalanza, Anna
  Ing{\'{o}}lfsd{\'{o}}ttir, and Karoliina Lehtinen.
\newblock {A}n {O}perational {G}uide to {M}onitorability with {A}pplications to
  {R}egular {P}roperties.
\newblock {\em Software and Systems Modeling}, 2021.
\newblock (to appear).

\bibitem{AAttardFI:outlined-choreographed-tech:21}
Luca Aceto, Duncan~Paul Attard, Adrian Francalanza, and Anna
  Ing{\'{o}}lfsd{\'{o}}ttir.
\newblock A choreographed outline instrumentation algorithm for asynchronous
  components.
\newblock {\em CoRR}, abs/2104.09433, 2021.

\bibitem{DBLP:conf/concur/AcetoCFI18}
Luca Aceto, Ian Cassar, Adrian Francalanza, and Anna Ing{\'{o}}lfsd{\'{o}}ttir.
\newblock On runtime enforcement via suppressions.
\newblock In Sven Schewe and Lijun Zhang, editors, {\em 29th International
  Conference on Concurrency Theory, {CONCUR} 2018, September 4-7, 2018,
  Beijing, China}, volume 118 of {\em LIPIcs}, pages 34:1--34:17. Schloss
  Dagstuhl - Leibniz-Zentrum f{\"{u}}r Informatik, 2018.
\newblock \href {https://doi.org/10.4230/LIPIcs.CONCUR.2018.34}
  {\path{doi:10.4230/LIPIcs.CONCUR.2018.34}}.

\bibitem{DBLP:conf/popl/AhmedFSW11}
Amal Ahmed, Robert~Bruce Findler, Jeremy~G. Siek, and Philip Wadler.
\newblock Blame for all.
\newblock In Thomas Ball and Mooly Sagiv, editors, {\em Proceedings of the 38th
  {ACM} {SIGPLAN-SIGACT} Symposium on Principles of Programming Languages,
  {POPL} 2011, Austin, TX, USA, January 26-28, 2011}, pages 201--214. {ACM},
  2011.
\newblock \href {https://doi.org/10.1145/1926385.1926409}
  {\path{doi:10.1145/1926385.1926409}}.

\bibitem{DBLP:journals/ftpl/AnconaBB0CDGGGH16}
Davide Ancona, Viviana Bono, Mario Bravetti, Joana Campos, Giuseppe Castagna,
  Pierre{-}Malo Deni{\'{e}}lou, Simon~J. Gay, Nils Gesbert, Elena Giachino,
  Raymond Hu, Einar~Broch Johnsen, Francisco Martins, Viviana Mascardi,
  Fabrizio Montesi, Rumyana Neykova, Nicholas Ng, Luca Padovani, Vasco~T.
  Vasconcelos, and Nobuko Yoshida.
\newblock Behavioral types in programming languages.
\newblock {\em Found. Trends Program. Lang.}, 3(2-3):95--230, 2016.
\newblock \href {https://doi.org/10.1561/2500000031}
  {\path{doi:10.1561/2500000031}}.

\bibitem{DBLP:series/lncs/BartocciFFR18}
Ezio Bartocci, Yli{\`{e}}s Falcone, Adrian Francalanza, and Giles Reger.
\newblock Introduction to runtime verification.
\newblock In Ezio Bartocci and Yli{\`{e}}s Falcone, editors, {\em Lectures on
  Runtime Verification - Introductory and Advanced Topics}, volume 10457 of
  {\em Lecture Notes in Computer Science}, pages 1--33. Springer, 2018.
\newblock \href {https://doi.org/10.1007/978-3-319-75632-5\_1}
  {\path{doi:10.1007/978-3-319-75632-5\_1}}.

\bibitem{BasinDHKRST:2020:verifiedMon}
David Basin, Thibault Dardinier, Lukas Heimes, Sr{\dj}an Krsti{\'{c}}, Martin
  Raszyk, Joshua Schneider, and Dmitriy Traytel.
\newblock A formally verified, optimized monitor for metric first-order dynamic
  logic.
\newblock In Nicolas Peltier and Viorica Sofronie-Stokkermans, editors, {\em
  Automated Reasoning: {(IJCAR)}}, {LNCS}, pages 432--453, Cham, 2020. Springer
  International Publishing.

\bibitem{BHLN12}
Jeremy Blackburn, Ivory Hernandez, Jay Ligatti, and Michael Nachtigal.
\newblock Completely subtyping iso-recursive types.
\newblock Technical Report CSE-071012, University of South Florida, 2012.

\bibitem{DBLP:journals/tcs/BocchiCDHY17}
Laura Bocchi, Tzu{-}Chun Chen, Romain Demangeon, Kohei Honda, and Nobuko
  Yoshida.
\newblock Monitoring networks through multiparty session types.
\newblock {\em Theor. Comput. Sci.}, 669:33--58, 2017.
\newblock \href {https://doi.org/10.1016/j.tcs.2017.02.009}
  {\path{doi:10.1016/j.tcs.2017.02.009}}.

\bibitem{BroFT:14:API}
Alan Brown, Jerry Fishenden, and Mark Thompson.
\newblock {\em API Economy, Ecosystems and Engagement Models}, pages 225--236.
\newblock Palgrave Macmillan UK, London, 2014.
\newblock \href {https://doi.org/10.1057/9781137443649_13}
  {\path{doi:10.1057/9781137443649_13}}.

\bibitem{erlangCesarini2009}
Francesco Cesarini and Simon Thompson.
\newblock {\em ERLANG Programming}.
\newblock O’Reilly Media, Inc., 1st edition, 2009.

\bibitem{DBLP:conf/tgc/ChenBDHY11}
Tzu{-}Chun Chen, Laura Bocchi, Pierre{-}Malo Deni{\'{e}}lou, Kohei Honda, and
  Nobuko Yoshida.
\newblock Asynchronous distributed monitoring for multiparty session
  enforcement.
\newblock In Roberto Bruni and Vladimiro Sassone, editors, {\em Trustworthy
  Global Computing - 6th International Symposium, {TGC} 2011, Aachen, Germany,
  June 9-10, 2011. Revised Selected Papers}, volume 7173 of {\em Lecture Notes
  in Computer Science}, pages 25--45. Springer, 2011.
\newblock \href {https://doi.org/10.1007/978-3-642-30065-3\_2}
  {\path{doi:10.1007/978-3-642-30065-3\_2}}.

\bibitem{DBLP:journals/lmcs/ChenDSY17}
Tzu{-}Chun Chen, Mariangiola Dezani{-}Ciancaglini, Alceste Scalas, and Nobuko
  Yoshida.
\newblock On the preciseness of subtyping in session types.
\newblock {\em Log. Methods Comput. Sci.}, 13(2), 2017.
\newblock \href {https://doi.org/10.23638/LMCS-13(2:12)2017}
  {\path{doi:10.23638/LMCS-13(2:12)2017}}.

\bibitem{DBLP:conf/sp/AmorimDGHPST15}
Arthur~Azevedo de~Amorim, Maxime D{\'{e}}n{\`{e}}s, Nick Giannarakis, Catalin
  Hritcu, Benjamin~C. Pierce, Antal Spector{-}Zabusky, and Andrew Tolmach.
\newblock Micro-policies: Formally verified, tag-based security monitors.
\newblock In {\em 2015 {IEEE} Symposium on Security and Privacy, {SP} 2015, San
  Jose, CA, USA, May 17-21, 2015}, pages 813--830. {IEEE} Computer Society,
  2015.
\newblock \href {https://doi.org/10.1109/SP.2015.55}
  {\path{doi:10.1109/SP.2015.55}}.

\bibitem{DBLP:journals/fmsd/DemangeonHHNY15}
Romain Demangeon, Kohei Honda, Raymond Hu, Rumyana Neykova, and Nobuko Yoshida.
\newblock Practical interruptible conversations: distributed dynamic
  verification with multiparty session types and python.
\newblock {\em Formal Methods Syst. Des.}, 46(3):197--225, 2015.
\newblock \href {https://doi.org/10.1007/s10703-014-0218-8}
  {\path{doi:10.1007/s10703-014-0218-8}}.

\bibitem{DBLP:conf/fossacs/Francalanza16}
Adrian Francalanza.
\newblock A theory of monitors - (extended abstract).
\newblock In Bart Jacobs and Christof L{\"{o}}ding, editors, {\em Foundations
  of Software Science and Computation Structures - 19th International
  Conference, {FOSSACS} 2016, Held as Part of the European Joint Conferences on
  Theory and Practice of Software, {ETAPS} 2016, Eindhoven, The Netherlands,
  April 2-8, 2016, Proceedings}, volume 9634 of {\em Lecture Notes in Computer
  Science}, pages 145--161. Springer, 2016.
\newblock \href {https://doi.org/10.1007/978-3-662-49630-5\_9}
  {\path{doi:10.1007/978-3-662-49630-5\_9}}.

\bibitem{DBLP:conf/concur/Francalanza17}
Adrian Francalanza.
\newblock Consistently-detecting monitors.
\newblock In Roland Meyer and Uwe Nestmann, editors, {\em 28th International
  Conference on Concurrency Theory, {CONCUR} 2017, September 5-8, 2017, Berlin,
  Germany}, volume~85 of {\em LIPIcs}, pages 8:1--8:19. Schloss Dagstuhl -
  Leibniz-Zentrum f{\"{u}}r Informatik, 2017.
\newblock \href {https://doi.org/10.4230/LIPIcs.CONCUR.2017.8}
  {\path{doi:10.4230/LIPIcs.CONCUR.2017.8}}.

\bibitem{Monitors:21}
Adrian Francalanza.
\newblock A {T}heory of {M}onitors.
\newblock {\em Information and {C}omputation}, page 104704, 2021.
\newblock \href {https://doi.org/https://doi.org/10.1016/j.ic.2021.104704}
  {\path{doi:https://doi.org/10.1016/j.ic.2021.104704}}.

\bibitem{DBLP:conf/rv/FrancalanzaAAAC17}
Adrian Francalanza, Luca Aceto, Antonis Achilleos, Duncan~Paul Attard, Ian
  Cassar, Dario~Della Monica, and Anna Ing{\'{o}}lfsd{\'{o}}ttir.
\newblock A foundation for runtime monitoring.
\newblock In Shuvendu~K. Lahiri and Giles Reger, editors, {\em Runtime
  Verification - 17th International Conference, {RV} 2017, Seattle, WA, USA,
  September 13-16, 2017, Proceedings}, volume 10548 of {\em Lecture Notes in
  Computer Science}, pages 8--29. Springer, 2017.
\newblock \href {https://doi.org/10.1007/978-3-319-67531-2\_2}
  {\path{doi:10.1007/978-3-319-67531-2\_2}}.

\bibitem{DBLP:journals/fmsd/FrancalanzaAI17}
Adrian Francalanza, Luca Aceto, and Anna Ing{\'{o}}lfsd{\'{o}}ttir.
\newblock Monitorability for the hennessy-milner logic with recursion.
\newblock {\em Formal Methods Syst. Des.}, 51(1):87--116, 2017.
\newblock \href {https://doi.org/10.1007/s10703-017-0273-z}
  {\path{doi:10.1007/s10703-017-0273-z}}.

\bibitem{DBLP:conf/coordination/FrancalanzaX20}
Adrian Francalanza and Jasmine Xuereb.
\newblock On implementing symbolic controllability.
\newblock In Simon Bliudze and Laura Bocchi, editors, {\em Coordination Models
  and Languages - 22nd {IFIP} {WG} 6.1 International Conference, {COORDINATION}
  2020, Held as Part of the 15th International Federated Conference on
  Distributed Computing Techniques, DisCoTec 2020, Valletta, Malta, June 15-19,
  2020, Proceedings}, volume 12134 of {\em {LNCS}}, pages 350--369. Springer,
  2020.
\newblock \href {https://doi.org/10.1007/978-3-030-50029-0\_22}
  {\path{doi:10.1007/978-3-030-50029-0\_22}}.

\bibitem{GhilezanJPSY19}
Silvia Ghilezan, Svetlana Jaksic, Jovanka Pantovic, Alceste Scalas, and Nobuko
  Yoshida.
\newblock Precise subtyping for synchronous multiparty sessions.
\newblock {\em J. Log. Algebraic Methods Program.}, 104:127--173, 2019.
\newblock \href {https://doi.org/10.1016/j.jlamp.2018.12.002}
  {\path{doi:10.1016/j.jlamp.2018.12.002}}.

\bibitem{GhilezanPPSY21}
Silvia Ghilezan, Jovanka Pantovi\'{c}, Ivan Proki\'{c}, Alceste Scalas, and
  Nobuko Yoshida.
\newblock Precise subtyping for asynchronous multiparty sessions.
\newblock {\em Proc. ACM Program. Lang.}, 5(POPL), January 2021.
\newblock \href {https://doi.org/10.1145/3434297} {\path{doi:10.1145/3434297}}.

\bibitem{DBLP:conf/esop/GommerstadtJP18}
Hannah Gommerstadt, Limin Jia, and Frank Pfenning.
\newblock Session-typed concurrent contracts.
\newblock In Amal Ahmed, editor, {\em Programming Languages and Systems - 27th
  European Symposium on Programming, {ESOP} 2018, Held as Part of the European
  Joint Conferences on Theory and Practice of Software, {ETAPS} 2018,
  Thessaloniki, Greece, April 14-20, 2018, Proceedings}, volume 10801 of {\em
  {LNCS}}, pages 771--798. Springer, 2018.
\newblock \href {https://doi.org/10.1007/978-3-319-89884-1\_27}
  {\path{doi:10.1007/978-3-319-89884-1\_27}}.

\bibitem{DBLP:conf/tacas/HamersJ20}
Ruben Hamers and Sung{-}Shik Jongmans.
\newblock Discourje: Runtime verification of communication protocols in
  clojure.
\newblock In Armin Biere and David Parker, editors, {\em Tools and Algorithms
  for the Construction and Analysis of Systems - 26th International Conference,
  {TACAS} 2020, Held as Part of the European Joint Conferences on Theory and
  Practice of Software, {ETAPS} 2020, Dublin, Ireland, April 25-30, 2020,
  Proceedings, Part {I}}, volume 12078 of {\em Lecture Notes in Computer
  Science}, pages 266--284. Springer, 2020.
\newblock \href {https://doi.org/10.1007/978-3-030-45190-5\_15}
  {\path{doi:10.1007/978-3-030-45190-5\_15}}.

\bibitem{DBLP:conf/concur/Honda93}
Kohei Honda.
\newblock Types for dyadic interaction.
\newblock In Eike Best, editor, {\em {CONCUR} '93, 4th International Conference
  on Concurrency Theory, Hildesheim, Germany, August 23-26, 1993, Proceedings},
  volume 715 of {\em Lecture Notes in Computer Science}, pages 509--523.
  Springer, 1993.
\newblock \href {https://doi.org/10.1007/3-540-57208-2\_35}
  {\path{doi:10.1007/3-540-57208-2\_35}}.

\bibitem{DBLP:conf/esop/HondaVK98}
Kohei Honda, Vasco~Thudichum Vasconcelos, and Makoto Kubo.
\newblock Language primitives and type discipline for structured
  communication-based programming.
\newblock In Chris Hankin, editor, {\em Programming Languages and Systems -
  ESOP'98, 7th European Symposium on Programming, Held as Part of the European
  Joint Conferences on the Theory and Practice of Software, ETAPS'98, Lisbon,
  Portugal, March 28 - April 4, 1998, Proceedings}, volume 1381 of {\em Lecture
  Notes in Computer Science}, pages 122--138. Springer, 1998.
\newblock \href {https://doi.org/10.1007/BFb0053567}
  {\path{doi:10.1007/BFb0053567}}.

\bibitem{HYC08}
Kohei Honda, Nobuko Yoshida, and Marco Carbone.
\newblock Multiparty asynchronous session types.
\newblock In {\em {POPL}}, 2008.
\newblock Full version in \cite{HYC16}.
\newblock \href {https://doi.org/10.1145/1328438.1328472}
  {\path{doi:10.1145/1328438.1328472}}.

\bibitem{HYC16}
Kohei Honda, Nobuko Yoshida, and Marco Carbone.
\newblock Multiparty asynchronous session types.
\newblock {\em J.~ACM}, 63(1), 2016.
\newblock \href {https://doi.org/10.1145/2827695} {\path{doi:10.1145/2827695}}.

\bibitem{DBLP:conf/rv/HuNYDH13}
Raymond Hu, Rumyana Neykova, Nobuko Yoshida, Romain Demangeon, and Kohei Honda.
\newblock Practical interruptible conversations - distributed dynamic
  verification with session types and python.
\newblock In Axel Legay and Saddek Bensalem, editors, {\em Runtime Verification
  - 4th International Conference, {RV} 2013, Rennes, France, September 24-27,
  2013. Proceedings}, volume 8174 of {\em Lecture Notes in Computer Science},
  pages 130--148. Springer, 2013.
\newblock \href {https://doi.org/10.1007/978-3-642-40787-1\_8}
  {\path{doi:10.1007/978-3-642-40787-1\_8}}.

\bibitem{DBLP:journals/pacmpl/IgarashiTVW17}
Atsushi Igarashi, Peter Thiemann, Vasco~T. Vasconcelos, and Philip Wadler.
\newblock Gradual session types.
\newblock {\em Proc. {ACM} Program. Lang.}, 1({ICFP}):38:1--38:28, 2017.
\newblock \href {https://doi.org/10.1145/3110282} {\path{doi:10.1145/3110282}}.

\bibitem{DBLP:conf/popl/JiaGP16}
Limin Jia, Hannah Gommerstadt, and Frank Pfenning.
\newblock Monitors and blame assignment for higher-order session types.
\newblock In Rastislav Bod{\'{\i}}k and Rupak Majumdar, editors, {\em
  Proceedings of the 43rd Annual {ACM} {SIGPLAN-SIGACT} Symposium on Principles
  of Programming Languages, {POPL} 2016, St. Petersburg, FL, USA, January 20 -
  22, 2016}, pages 582--594. {ACM}, 2016.
\newblock \href {https://doi.org/10.1145/2837614.2837662}
  {\path{doi:10.1145/2837614.2837662}}.

\bibitem{DBLP:journals/jlp/LeuckerS09}
Martin Leucker and Christian Schallhart.
\newblock A brief account of runtime verification.
\newblock {\em J. Log. Algebraic Methods Program.}, 78(5):293--303, 2009.
\newblock \href {https://doi.org/10.1016/j.jlap.2008.08.004}
  {\path{doi:10.1016/j.jlap.2008.08.004}}.

\bibitem{Ligatti05}
Jay Ligatti, Lujo Bauer, and David Walker.
\newblock Edit automata: enforcement mechanisms for run-time security policies.
\newblock {\em Int. J. Inf. Secur.}, 4(1-2):2--16, 2005.
\newblock URL: \url{http://dx.doi.org/10.1007/s10207-004-0046-8}, \href
  {https://doi.org/10.1007/s10207-004-0046-8}
  {\path{doi:10.1007/s10207-004-0046-8}}.

\bibitem{BHLN17}
Jay Ligatti, Jeremy Blackburn, and Michael Nachtigal.
\newblock On subtyping-relation completeness, with an application to
  iso-recursive types.
\newblock {\em ACM Trans. Program. Lang. Syst.}, 39(1):4:1--4:36, March 2017.
\newblock \href {https://doi.org/10.1145/2994596} {\path{doi:10.1145/2994596}}.

\bibitem{Aceto07}
Kim Guldstrand~Larsen Luca~Aceto, Anna~Ing{\'{o}}lfsd{\'{o}}ttir and Jiri Srba.
\newblock {\em Reactive Systems: Modelling, Specification and Verification}.
\newblock Cambridge University Press, 2007.

\bibitem{SMTP}
{Network Working Group}.
\newblock {RFC} 5321: {Simple Mail Transfer Protocol}.
\newblock \url{https://tools.ietf.org/html/rfc5321}, 2008.

\bibitem{Neykova18CC}
Rumyana Neykova, Raymond Hu, Nobuko Yoshida, and Fahd Abdeljallal.
\newblock A session type provider: Compile-time api generation of distributed
  protocols with refinements in f\#.
\newblock In {\em Proceedings of the 27th International Conference on Compiler
  Construction}, CC 2018, page 128–138, New York, NY, USA, 2018. Association
  for Computing Machinery.
\newblock \href {https://doi.org/10.1145/3178372.3179495}
  {\path{doi:10.1145/3178372.3179495}}.

\bibitem{DBLP:conf/cc/NeykovaY17}
Rumyana Neykova and Nobuko Yoshida.
\newblock Let it recover: multiparty protocol-induced recovery.
\newblock In Peng Wu and Sebastian Hack, editors, {\em Proceedings of the 26th
  International Conference on Compiler Construction, Austin, TX, USA, February
  5-6, 2017}, pages 98--108. {ACM}, 2017.
\newblock URL: \url{http://dl.acm.org/citation.cfm?id=3033031}.

\bibitem{DBLP:conf/rv/NeykovaYH13}
Rumyana Neykova, Nobuko Yoshida, and Raymond Hu.
\newblock {SPY:} local verification of global protocols.
\newblock In Axel Legay and Saddek Bensalem, editors, {\em Runtime Verification
  - 4th International Conference, {RV} 2013, Rennes, France, September 24-27,
  2013. Proceedings}, volume 8174 of {\em Lecture Notes in Computer Science},
  pages 358--363. Springer, 2013.
\newblock \href {https://doi.org/10.1007/978-3-642-40787-1\_25}
  {\path{doi:10.1007/978-3-642-40787-1\_25}}.

\bibitem{DBLP:journals/entcs/Peled02}
Doron~A. Peled.
\newblock Specification and verification using message sequence charts.
\newblock {\em Electron. Notes Theor. Comput. Sci.}, 65(7):51--64, 2002.
\newblock \href {https://doi.org/10.1016/S1571-0661(04)80484-5}
  {\path{doi:10.1016/S1571-0661(04)80484-5}}.

\bibitem{Pierce02}
Benjamin~C. Pierce.
\newblock {\em Types and Programming Languages}.
\newblock The MIT Press, 1st edition, 2002.

\bibitem{DBLP:conf/ecoop/ScalasDHY17}
Alceste Scalas, Ornela Dardha, Raymond Hu, and Nobuko Yoshida.
\newblock A linear decomposition of multiparty sessions for safe distributed
  programming.
\newblock In Peter M{\"{u}}ller, editor, {\em 31st European Conference on
  Object-Oriented Programming, {ECOOP} 2017, June 19-23, 2017, Barcelona,
  Spain}, volume~74 of {\em LIPIcs}, pages 24:1--24:31. Schloss Dagstuhl -
  Leibniz-Zentrum f{\"{u}}r Informatik, 2017.
\newblock \href {https://doi.org/10.4230/LIPIcs.ECOOP.2017.24}
  {\path{doi:10.4230/LIPIcs.ECOOP.2017.24}}.

\bibitem{DBLP:journals/darts/ScalasDHY17}
Alceste Scalas, Ornela Dardha, Raymond Hu, and Nobuko Yoshida.
\newblock A linear decomposition of multiparty sessions for safe distributed
  programming (artifact).
\newblock {\em {DARTS}}, 3(2), 2017.
\newblock \href {https://doi.org/10.4230/DARTS.3.2.3}
  {\path{doi:10.4230/DARTS.3.2.3}}.

\bibitem{DBLP:conf/ecoop/ScalasY16}
Alceste Scalas and Nobuko Yoshida.
\newblock Lightweight session programming in scala.
\newblock In Shriram Krishnamurthi and Benjamin~S. Lerner, editors, {\em 30th
  European Conference on Object-Oriented Programming, {ECOOP} 2016, July 18-22,
  2016, Rome, Italy}, volume~56 of {\em LIPIcs}, pages 21:1--21:28. Schloss
  Dagstuhl - Leibniz-Zentrum f{\"{u}}r Informatik, 2016.
\newblock \href {https://doi.org/10.4230/LIPIcs.ECOOP.2016.21}
  {\path{doi:10.4230/LIPIcs.ECOOP.2016.21}}.

\bibitem{Schneider:2000}
Fred~B. Schneider.
\newblock Enforceable security policies.
\newblock {\em ACM Trans. Inf. Syst. Secur.}, 3(1):30--50, February 2000.
\newblock URL: \url{http://doi.acm.org/10.1145/353323.353382}, \href
  {https://doi.org/10.1145/353323.353382} {\path{doi:10.1145/353323.353382}}.

\bibitem{scribble}
{Scribble homepage}, 2020.
\newblock \url{http://www.scribble.org}.

\bibitem{SMTPSurvey}
{SecuritySpace}.
\newblock Mail ({MX}) server survey, 2021.
\newblock URL:
  \url{http://www.securityspace.com/s_survey/data/man.202103/mxsurvey.html}.

\bibitem{SeveriD19}
Paula Severi and Mariangiola Dezani{-}Ciancaglini.
\newblock Observational equivalence for multiparty sessions.
\newblock {\em Fundam. Informaticae}, 170(1-3):267--305, 2019.
\newblock \href {https://doi.org/10.3233/FI-2019-1863}
  {\path{doi:10.3233/FI-2019-1863}}.

\bibitem{Sho:ISSTA:07:AutomataAPI}
Sharon Shoham, Eran Yahav, Stephen Fink, and Marco Pistoia.
\newblock Static specification mining using automata-based abstractions.
\newblock In {\em Proceedings of the 2007 International Symposium on Software
  Testing and Analysis {(ISSTA)}}, ISSTA '07, page 174–184, New York, NY,
  USA, 2007. Association for Computing Machinery.
\newblock \href {https://doi.org/10.1145/1273463.1273487}
  {\path{doi:10.1145/1273463.1273487}}.

\bibitem{SonT:13:MC-API}
Fu~Song and Tayssir Touili.
\newblock Model-checking software library api usage rules.
\newblock In Einar~Broch Johnsen and Luigia Petre, editors, {\em Integrated
  Formal Methods}, pages 192--207, Berlin, Heidelberg, 2013. Springer Berlin
  Heidelberg.

\bibitem{YHNN2013}
Nobuko Yoshida, Raymond Hu, Rumyana Neykova, and Nicholas Ng.
\newblock The {S}cribble protocol language.
\newblock In {\em TGC}, 2013.
\newblock \href {https://doi.org/10.1007/978-3-319-05119-2\_3}
  {\path{doi:10.1007/978-3-319-05119-2\_3}}.

\end{thebibliography}

\iftoggle{techreport}{%
\appendix

\section{Process Calculus and Types --- Addendum to \Cref{sec:calculus-types}}
\label{app:calculus-types}

\subsection{Process Calculus}

We use standard substitution functions both for value variables and process variables. 
They operate on value variables and process variables respectively in the following way:

\medskip
\centerline{
  $x[\nicefrac{v}{y}] \equal 
  \left\{\begin{array}{ll}
    v & \text{if }x\equal y;\\
    x & \text{otherwise}.
  \end{array}\right.\qquad
  X[\nicefrac{\mu_Y.Q}{Y}] \equal
  \left\{\begin{array}{ll}
    \mu_{X}.P & \text{if }X \equal Y; \\
    X & \text{otherwise}.
  \end{array}\right.$
}
\medskip

\noindent and are homomorphic for all the other cases. 
We also use standard functions to determine the \emph{free} value variables and process variables within processes. 
These follow the standard inductive definitions found in the literature such as \cite{Pierce02} which give the set of variables appearing free in a process.

For the process semantics, we use the set $\textsc{Act}$ to contain all \textbf{visible actions} (\ie inputs and outputs);
hence, an action $\alpha \in \textsc{Act}$ %
cannot be $\tau$. 

\subsection{Binary Session Types with Assertions}
\label{sec:session-types-assertions}

In this section we present the full syntax and semantics of session types
augmented with assertions: the are formalised in \Cref{fig:session-types-with-assertions}.

The assertions $A$ have the same syntax of the predicates of the process calculus (\Cref{fig:process-calc-syntax-semantics}).
Unlike typical formulations, our process types have \emph{payload variables} $x$ that enable assertions to place predicates over communicated values. We often omit trivial assertions (\ie \textsf{tt}), and payload variables that are not referenced in any assertion: these omissions produce the syntax shown in \Cref{fig:session-types-noa}.
We assume a standard substitution definition for the session types. 

\begin{figure}
  \small
  \textbf{Syntax}\\
  \centerline{\(
    \begin{array}{r@{\hskip 2mm}r@{\hskip 2mm}c@{\hskip 2mm}l}
      \text{Session types}& R,S & \Coloneqq & \stsel{i}{I} \quad\vert\quad \stbra{i}{I}\\
      &&& \;\vert\; \strec{X}.S \;\vert\; X \;\vert\; \textsf{end}
      \hspace{6mm}
      \text{\footnotesize(with $I \!\neq\! \emptyset$,\;$\texttt{l}_i$ pairwise distinct)} 
    \end{array}
    \)
  }
  \smallskip
  \textbf{Transitions}\\
  \smallskip
  \centerline{
    \inference[\text{[\textsc{sRec}]}]{}{\strec{X}.S \xrightarrow{\tau} S[\nicefrac{\strec{X}.S}{X}]}
    \hspace{0.6cm}
    \inference[\text{[\textsc{sSel}]}]{}{\stsel{i}{I} \xrightarrow{\selactk{j}} S_j} $j \in I$
  }\smallskip\\
  \centerline{
    \inference[\text{[\textsc{sBra}]}]{}{\stbra{i}{I} \xrightarrow{\braactk{j}} S_j} $j \in I$
  }\smallskip\\

  \textbf{Dual types}\\[2mm]
  \centerline{\(
    \begin{array}{r@{\;\;}c@{\;\;}l@{\qquad}c}
    \overline{\stbra{i}{I}} &=& \oplus \big\{!\texttt{l}_{i}(x_{i}:\textsf{B}_{i})[A_{i}].\overline{S_{i}}\big\}_{i\in I} &
    \overline{\textsf{end}} \,=\, \textsf{end} \qquad {\overline{X}} \,=\, X\\[2mm]
    \overline{\stsel{i}{I}} &=& \& \big\{?\texttt{l}_{i}(x_{i}:\textsf{B}_{i})[A_{i}].\overline{S_{i}}\big\}_{i\in I}&
    \overline{\strec{X}.S} \;=\; \strec{X}.\overline{S}
  \end{array}
  \)}
  \caption[Session Types Syntax Augmented with Assertions, their Transitions and Dual Types]{Session Types Augmented with Assertions: Syntax, Transitions and Duality}\label{fig:session-types}\label{fig:session-types-with-assertions}
\end{figure} 
These session types come equipped with transition rules (also in \Cref{fig:session-types}) that show how a type evolves with the particular actions. 
Similarly to the actions for processes in Section \ref{s:process-calculus}, the actions $\delta\in \Act_{ST} \cup \{\tau\}$ range over \emph{output} actions $\selact{l}$, \emph{input} actions $\braact{l}$ and a \emph{silent} action $\tau$. 
The rule \ltsrule{sRec} unfolds the type $\strec{X}.S$ to $S[\nicefrac{\strec{X}.S}{X}]$ by the silent transition $\tau$. 
By \ltsrule{sSel} the selection type proceeds according to the continuation type $S_j$ with the input action $\selactk{j}$, also given that $j \in I$. 
Similarly, by \ltsrule{sBra} the branching type may transition to the continuation type $S_j$ with the input action $\braactk{j}$, given that $j\in I$. 
It is worth noting that within our model, the branching and selection types transition irrespective of the outcome of the assertions. 

\subsection{Session Type System with Assertion Well-Formedness Checks}\label{s:session-type-system-assertions}

The typing rules in \Cref{fig:session-typing-rules} augment those in \Cref{fig:session-typing-rules-noa} with additional checks on the well-formedness of assertions.  When assertions are trivial ($\true$), we obtain the rules in \Cref{fig:session-typing-rules-noa}.
As it may be noted from the augmented rules \ltsrule{tBra} and \ltsrule{tSel}, assertions can only be statically checked for their type and cannot be evaluated since some of the variables only become bound to a value at runtime. 

\begin{figure}
  \small
  \textbf{Process Typing}
  \bigskip\\
  \centerline{
    \inference[\text{[\textsc{tBra}]}]{\forall i \in I, \exists j \in J \cdot \texttt{l}_j \equal \texttt{l}_i,\quad \textsf{B}_j\equal \textsf{B}_i,\quad\Theta\cdot\Gamma, x_j:\textsf{B}_j \rightvdash P_j:S_i,\quad \Gamma, x_j:\textsf{B}_j \rightvdash A_i:\textsf{Bool}}{\Theta\cdot\Gamma \rightvdash \prcv{j}{J}: \stbra{i}{I}}
  }\bigskip\\
  \medskip\centerline{
    \inference[\text{[\textsc{tSel}]}]{\exists i \in I \cdot \texttt{l} \equal \texttt{l}_i,\quad \Gamma\rightvdash a:\textsf{B}_i,\quad\Theta\cdot\Gamma \rightvdash P:S_i,\quad \Gamma,x_i:\basetype_i \rightvdash A_i:\textsf{Bool}}{\Theta\cdot\Gamma \rightvdash \psndk{l}{a}.P: \stsel{i}{I}}
  }\medskip\\
  \centerline{
    \inference[\text{[\textsc{tRec}]}]{\Theta, X:S\cdot\Gamma\rightvdash P:S}{\Theta\cdot\Gamma \rightvdash \pmu{X}.P:S}
    \hspace{1cm}
    \inference[\text{[\textsc{tPVar}]}]{\Theta(X)\equal S}{\Theta\cdot\Gamma \rightvdash X :S}
  }\bigskip\\
  \medskip\centerline{
    \inference[\text{[\textsc{tIf}]}]{\Gamma\rightvdash A:\textsf{Bool}&\Theta\cdot\Gamma\rightvdash P:S & \Theta\cdot\Gamma\rightvdash Q:S}{\Theta\cdot\Gamma \rightvdash \pif{A}{P}{Q}: S}
    \hspace{1cm}
    \inference[\text{[\textsc{tNil}]}]{}{\Theta\cdot\Gamma \rightvdash \pnil: \textsf{end}}
  }
  \caption[Session Typing Rules with Assertion Well-Formedness Checks]{Session Typing Rules with Assertion Well-Formedness Checks.}\label{fig:session-typing-rules}
\end{figure}
 
We assume that predicates can be type-checked with standard rules having the following standard properties:
\begin{enumerate}
  \item \label{assertion-lemma-1}(\textit{Weakening lemma}). If $\Gamma\rightvdash A:\textsf{Bool}$, then $\Gamma,x:\basetype\rightvdash A:\textsf{Bool}$.
  \item \label{assertion-lemma-2}(\textit{Strengthening lemma}). If $\Gamma,x:\basetype\rightvdash A:\textsf{Bool}$ and $x \not\in \fv{A}$, then $\Gamma\rightvdash A:\textsf{Bool}$.
  \item \label{assertion-lemma-3}(\textit{Variable instantiation in assertions lemma}). If $\Gamma,x:\basetype\rightvdash A:\textsf{Bool}$ then $\Gamma\rightvdash A[\nicefrac{v}{x}]:\textsf{Bool}$.
  \item \label{assertion-lemma-4}(\textit{Substitution lemma for variables in assertions}). If $\Gamma \rightvdash v:\basetype$ and $\Gamma,y:\basetype \rightvdash A:\textsf{Bool}$, then $\Gamma \rightvdash A[\nicefrac{v}{y}]:\textsf{Bool}$.
\end{enumerate}

\section{Proof of \Cref{thrm:sound}}
\label{sec:proof-mon-soundness}

\subsection{Properties of the Session Typing System}\label{s:properties-of-type-system}

The proof for subject reduction depends on an important technical property known as the \emph{substitution} lemma. 
This depends on two minor, yet crucial, properties of the type system that are the \emph{weakening} (Lemma \ref{lemma:weakening}) and \emph{strengthening} (Lemma \ref{lemma:strengthening}) lemmata.
These lemmata respectively show that the addition of new mappings and removal of mappings from the typing environments (with some restrictions) do not affect a type-checked process or identifier. 

We start with the weakening lemma which consists of two statements: for the identifiers $a$ and processes $P$. 
The latter consists of two further cases: for the weakening of the environment $\Gamma$ with the addition of variables $x$ and similarly for $\Theta$ with the addition of process variables $X$. 

\begin{lemma}[\textit{Weakening}]~\label{lemma:weakening}
  \begin{enumerate}
    \item\label{weak-i} $\Gamma\rightvdash a:\basetype$ implies $\Gamma,x:\basetype'\rightvdash a:\basetype$
    \item $\Theta\cdot\Gamma\rightvdash P:S$ implies: 
    \begin{enumerate}
      \item\label{weak-ii-a} $\Theta\cdot\Gamma,x:\basetype\rightvdash P:S$
      \item\label{weak-ii-b} $\Theta,X:S'\cdot\Gamma\rightvdash P:S$
    \end{enumerate}
  \end{enumerate} 
\end{lemma}
\begin{proof}
  Case \ref{weak-i} follows by induction on the derivation $\Gamma\rightvdash a:\basetype$. The proofs for Cases \ref{weak-ii-a} and \ref{weak-ii-b} follow by induction on the derivation $\Theta\cdot\Gamma\rightvdash P:S$. 
\end{proof}

Similarly to the weakening lemma, the strengthening lemma stated below consists of three statements. 
The first is for identifiers $a$ and the other two are for processes $P$. 
In all of the cases, one of the type environments $\Gamma$ or $\Theta$ used to type-check $a$ or $P$ is strengthened with the removal of a mapping. 
The variable $x$ or process variable $X$ being removed from the environments must not be equal to $a$ nor in the free variables of $P$.

\begin{lemma}[\textit{Strengthening}]~\label{lemma:strengthening}
  \begin{enumerate}
    \item\label{strength-i} If $\Gamma,x:\basetype' \rightvdash a:\basetype$ and $x \not\equal a$, then $\Gamma \rightvdash a:\basetype$.
    \item\label{strength-ii} If $\Theta\cdot\Gamma,x:\basetype \rightvdash P:S$ where $x$ is not in $\textbf{fv}(P)$, then $\Theta\cdot\Gamma\rightvdash P:S$.
    \item\label{strength-iii} If $\Theta,X:S'\cdot\Gamma \rightvdash P:S$ where $X$ is not in $\fpv{P}$, then $\Theta\cdot\Gamma\rightvdash P:S$.
  \end{enumerate}
\end{lemma}
\begin{proof}
  The proofs for the three cases in Lemma \ref{lemma:strengthening} follow by induction on the derivation $\Gamma,x:\basetype' \rightvdash a:\basetype$ for Case \ref{strength-i}, $\Theta\cdot\Gamma,x:\basetype \rightvdash P:S$ for Case \ref{strength-ii} and $\Theta,X:S'\cdot\Gamma \rightvdash P:S$ for Case \ref{strength-iii}. 
\end{proof}

For a type system to satisfy subject reduction, it must ensure that a typed process remain typed after it transitions. 
As a process transitions, in particular with the rules \ltsrule{pRcv} and \ltsrule{pRec}, variables are substituted with values, and process variables with processes. 
Substitution lemmata show that typing is preserved by these substitutions, provided that the types assigned to the variables coincide with those of the values and processes. 

Similarly to the previous lemmata, we require three statements, two of which are in Lemma \ref{lemma:val-subs} and concern values being substituted with identifiers $a$ or within processes $P$. 
Whereas Lemma \ref{lemma:process-subs} is concerned with the substitution of process variables with processes. 
The strengthening and weakening lemmata are used within these proofs, where it is either required to remove unused mappings from the type environments, or to introduce fresh mappings to allow for the induction to go through. 

\begin{lemma}[\textit{Value Substitution lemma}]~\label{lemma:val-subs}
  \begin{enumerate}
    \item\label{lemma:val-subs-1} If $\Gamma\rightvdash v:\basetype$ and $\Gamma,y:\basetype\rightvdash a:\basetype'$ then $\Gamma\rightvdash a[\nicefrac{v}{y}]:\basetype'$.
    \item\label{lemma:val-subs-2} If $\Gamma\rightvdash v:\basetype$ and $\Theta\cdot\Gamma,y:\basetype\rightvdash P:S$ then $\Theta\cdot\Gamma\rightvdash P[\nicefrac{v}{y}]:S$.
  \end{enumerate}
\end{lemma}
\begin{proof}
  The proofs for Cases \ref{lemma:val-subs-1} and \ref{lemma:val-subs-2} follow by induction on the derivations $\Gamma,y:\basetype\rightvdash a:\basetype'$ and $\Theta\cdot\Gamma,y:\basetype\rightvdash P:S$ respectively. 
\end{proof}

\begin{lemma}[\textit{Process Substitution lemma}]\label{lemma:process-subs}
  If $\Theta\cdot\Gamma\rightvdash \mu_X.P':S'$ and $\Theta,X:S'\cdot\Gamma\rightvdash P:S$ then $\Theta\cdot\Gamma\rightvdash P[\nicefrac{\mu_X.P'}{X}]:S$.
\end{lemma}
\begin{proof}
  The proof follows by induction on the derivation $\Theta,X:S'\cdot\Gamma\rightvdash P:S$.
\end{proof}

Using Lemma \ref{lemma:val-subs} and Lemma \ref{lemma:process-subs} we can prove subject reduction. 
However, for the soundness proof in Section \ref{s:soundness-proof}, we require another property of the type system. 
This is proved in Lemma \ref{lemma:type-substitution} which states that a process that type-checks with a session type $S$ also type-checks with the same $S$ in which variables are substituted for values, \ie $S[\nicefrac{v}{x}]$. 
Consequently, the mappings in the environment $\Theta$ must also be updated accordingly. 
This is achieved through $\Theta[\nicefrac{v}{x}]$ in which substitution is carried out on every session type that the process variables map to, defined as: 
\begin{math}
  \Theta[\nicefrac{v}{x}] \equal \big\{ X:S[\nicefrac{v}{x}]\ \big\vert\ \Theta(X)\equal S\big\}.
\end{math}

\begin{lemma}[\textit{Type Variable Instantiation}]\label{lemma:type-substitution}For all processes $P \in \textsc{Proc}$ and session types $S$, $\Theta\cdot\Gamma \rightvdash P:S$ implies $\Theta[\nicefrac{v}{x}]\cdot\Gamma \rightvdash P:S[\nicefrac{v}{x}]$.
\end{lemma}

\subsection{Transitions and Typing}\label{s:transitions-and-typing}

At this stage, we can proceed with the proof for subject reduction. 
In its simplest form, subject reduction is the property which ensures that a well-typed process remains well-typed when it transitions to a new state. 
This is a very important property of any type-checking system which if not satisfied, would indicate that the semantics are not in harmony with the static analysis. 

We first require some additional definitions to reason about the behaviour of processes with respect to a session type, starting with Definition \ref{def:accepted-actions}. 
This defines the criteria for an action or list of actions (\ie a \emph{trace}) committed either by a process, to be \emph{accepted} by a type. 
In the subject reduction property, this definition is used to restrict the actions of a process to only those that conform with the session type it is expected to follow.

\begin{definition}[\textit{Accepted Actions}]\label{def:accepted-actions}
  An action $\mu$ performed by a process or a monitor is accepted by a session type $S$, denoted as $\textbf{accepted}(\mu,S)$, in the cases:
  \begin{align*}
    \accepted{\tau}{S} &\quad \text{always}\\
    \left.
    \begin{array}{rl}
      \accepted{\triangleleft \texttt{l}(v)}{\stsel{i}{I}}\\
      \accepted{\envsndOP \texttt{l}(v)}{\stsel{i}{I}}\\
      \accepted{\triangleright \texttt{l}(v)}{\stbra{i}{I}}\\
      \accepted{\envrcvOP \texttt{l}(v)}{\stbra{i}{I}}
    \end{array} \right\} &\quad \text{if } \texttt{l}\equal\texttt{l}_j \text{ and } \emptyset \rightvdash v:\textsf{B}_j \text{ s.t. } j\in I
  \end{align*}
\end{definition}

As it may be noted, Definition \ref{def:accepted-actions} also considers monitor actions, this is because this definition will be used later on in the soundness proof.

The second definition that is required \emph{matches} actions by processes with the transitions of session types. 
In particular, it ensures that the label of the message in the input and output transitions of processes and session types are the same. 
This allows us to ensure that the process and the type that it follows, transition with the same label.

\begin{definition}[\textit{Matching Processes and Session Type Transitions}]
  An action $\mu$ performed by a process $P$ matches a transition $\delta$ of a session type $S$, denoted as $\match{\mu}{\delta}$, in the cases:
  \vspace{-3mm}
  \begin{align*}
    \match{\tau}{\tau}& \quad\text{always}\\
    \match{\triangleleft \texttt{l}(v)}{\triangleleft \texttt{l}}& \quad\text{for any }v\\
    \match{\triangleright \texttt{l}(v)}{\triangleright \texttt{l}}& \quad\text{for any }v
  \end{align*} %
\end{definition}

The subject reduction property is stated in Proposition \ref{prop:sub-red}. 
Intuitively, this proposition states that when a \emph{closed} process $P$ type-checks with a session type $S$ and transitions to $P'$ with an action \emph{accepted} by $S$, either $S$ also transitions with a \emph{matching} action that types $P'$, or $P$ performed a silent action and it still type-checks with $S$. 

\begin{proposition}[\textit{Subject Reduction}]\label{prop:sub-red}
  For any process $P$ and session type $S$, $\emptyset\cdot\emptyset\rightvdash P:S$ and $P\xrightarrow{\mu}P'$ where $\textbf{accepted}(\mu, S)$, implies:
  \begin{enumerate}
    \item\label{sub-red:implication-1} $S\xrightarrow{\delta}S'$ s.t.~$\textbf{match}(\mu,\delta)$ and $\emptyset\cdot\emptyset\rightvdash P':S'$; or
    \item\label{sub-red:implication-2} $\mu \equal \tau$ and $\emptyset\cdot\emptyset\rightvdash P':S$. \hfill\envqed
  \end{enumerate} 
\end{proposition}
\bigskip

We also assume that the action $\mu$ performed by the process is accepted by the session type $S$ process $P$ type-checks with and unless it is a silent transition, $\mu$ is required to match the transition of $S$ (\ie $\delta$). 
The following example highlights why we require such assumptions and conditions for our subject reduction property.  

\begin{example}\label{eg:subject-reduction}
  \begin{subequations}
    Consider the session type $\sab$ and process $\pabc$ defined as:
    \begin{align}
      \sab&\ \equal\ \&\big\{?\labk{A}().S_a,\ ?\labk{B}().S_b\big\}\\
      \pabc&\ \equal\ \triangleright\big\{\labk{A}().P_a,\ \labk{B}().P_b,\ \labk{C}().P_c \big\}
    \end{align}

    \noindent If we assume that $P_a$ and $P_b$ type-check with the types $S_a$ and $S_b$ respectively, it can be shown that $\pabc$ type-checks with the session type $\sab$ by applying the rule \ltsrule{tBra} (from Section \ref{s:session-type-system}).
    This holds since every possible branch stated in $\sab$ is present as a choice in $\pabc$.

    If $\pabc$ performs the transition $\pabc \xrightarrow{\mu}P_a$ where $\mu \equal \triangleright\labk{A}()$ with the rule \ltsrule{pRcv}, it no longer type-checks with the type $\sab$. 
    Rather, it now type-checks with the type $S_a$. 
    As stated in Case \ref{sub-red:implication-1} of Proposition \ref{prop:sub-red}, $\sab$ can perform the transition $\sab \xrightarrow{\delta} S_a$ where $\delta\equal \triangleright\labk{A}$. 
    If $\sab$ performs $\sab \xrightarrow{\triangleright\labk{B}} S_b$ rather than $\sab \xrightarrow{\delta} S_a$, the process $P_a$ does not type-check with the type $S_b$ and consequently, $\pabc$ would not remain well-typed as it transitions.
    Therefore, contradicting the subject reduction property in which a well-typed process must remain well-typed after it transitions. 
    This is the reason why we require the actions to match and in this case, $\match{\mu}{\delta}$ holds. 

    Moreover, if $\pabc$ performs the transition $\pabc \xrightarrow{\triangleright\labk{C}()} P_c$ with the rule \ltsrule{pRcv}, it will end up in a state which $\sab$ cannot transition to and hence $\pabc$ is no longer well-typed. 
    Once again, this goes against the subject reduction property, and therefore we require that the action performed by the process is accepted by the type.
    In this case, by Definition \ref{def:accepted-actions}, $\accepted{\triangleright\labk{C}()}{\sab}$ does not hold. 
    To rule out such actions, in Proposition \ref{prop:sub-red}, we assume that processes only perform accepted actions.
  
    With Case \ref{sub-red:implication-2} we handle those cases in which processes perform silent actions ($\tau$ transitions) and remain well-typed according to the session type they were type-checked with before they transitioned.
  \end{subequations}
\end{example}

The following is the proof for Proposition \ref{prop:sub-red}. 

\begin{proof}
  The proof for Proposition \ref{prop:sub-red} follows by rule induction on the derivation $P\xrightarrow{\mu}P'$: 
  \begin{description}
    \item[\ltsrule{pSnd}] \begin{subequations} From the rule we know $P\equal\triangleleft \lab_j(v_j).P_j$, $P'\equal P_j$ and $\mu \equal \triangleleft \lab_j(v_j)$. We know $\emptyset\cdot\emptyset\rightvdash\triangleleft \lab_j(v_j).P_j:S$ could have been derived using only the rule \ltsrule{tSel} which means that $S\equal\stsel{i}{I}$ and: 
    \begin{gather}
      \label{psnd-1}\exists i \in I \cdot \lab_j \equal \lab_i,\quad\emptyset\rightvdash v_j:\basetype_i,\quad\Theta\cdot\Gamma \rightvdash P_j:S_i,\quad \emptyset\cdot\emptyset \rightvdash A_i:\textsf{Bool}
    \end{gather}
    Since by \ref{psnd-1} we know that $j \in I$, by the rule \ltsrule{sSel}, $S \xrightarrow{\triangleleft \lab_j} S_i$ and by definition, $\textbf{match}(\mu, \triangleleft \lab_j)$ holds. Moreover by \ref{psnd-1} we know that $\emptyset\cdot\emptyset\rightvdash P_j:S_i$ as required in the implication \ref{sub-red:implication-1}.
    \end{subequations}

    \item[\ltsrule{pRcv}] \begin{subequations} From the rule we know $P \equal \prcv{j}{J}$, $\mu\equal \triangleright\,l_k(v_k)$ for a $k \in J$ and $P'\equal P_k[\nicefrac{v_k}{x_k}]$. We know $\emptyset\cdot\emptyset\rightvdash P:S$ could have been derived using only the rule \ltsrule{tBra} which means that $S \equal \stbra{i}{I}$ and:
      \begin{gather}
        \label{prcv-1}\forall i \in I, \exists j \in J \cdot \lab_j \equal \lab_i,\quad \basetype_j\equal \basetype_i,\quad \emptyset\cdot x_j:\basetype_j \rightvdash P_j:S_i,\quad x_j:\basetype_j \rightvdash A_i:\textsf{Bool}
      \end{gather}
      Since from the assumption we know that $\textbf{accepted}(\triangleright\,l_k(v_k), S)$ then by definition:\begin{gather}
        \label{sub-red-prcv-1}k \in I\\
        \label{sub-red-prcv-2}\emptyset\rightvdash v_k:B_k
      \end{gather}
      From \ref{sub-red-prcv-1} we can apply the rule \ltsrule{sBra} to obtain $S \xrightarrow{\triangleright\,l_k}S_k$ and by definition we know that $\textbf{match}(\mu,\triangleright\,l_k)$ holds. From \ref{prcv-1} we know $\emptyset\cdot x_k:B_k\rightvdash P_k:S_k$ and since we know \ref{sub-red-prcv-2} we can apply the value substitution lemma (Lemma \ref{lemma:val-subs}) to obtain $\emptyset\cdot\emptyset\rightvdash P_k[\nicefrac{v_k}{x_k}]:S_k$ as required in the implication \ref{sub-red:implication-1}.
    \end{subequations}

    \item[\ltsrule{pRec}] \begin{subequations} From the rule we know $P\equal \mu_X.Q$, $\mu\equal\tau$ and $P'\equal Q[\nicefrac{\mu_X.Q}{X}]$. We know $\emptyset\cdot\emptyset\rightvdash \mu_X.Q:S$ could have been derived using only the rule \ltsrule{tRec} which means that the process $P$ type-checks with any session type $S$ such that $X:S\cdot\emptyset\rightvdash Q:S$. Since we know $\emptyset\cdot\emptyset\rightvdash \mu_X.Q:S$ and $X:S\cdot\emptyset\rightvdash Q:S$ by applying the process substitution lemma (Lemma \ref{lemma:process-subs}) we obtain $\emptyset\cdot\emptyset\rightvdash Q[\nicefrac{\mu_X.Q}{X}]:S$ as required in the implication \ref{sub-red:implication-2} since $\mu\equal\tau$.
    \end{subequations}

    \item[\ltsrule{pTru}] From the rule we know $P \equal \textsf{if } A \textsf{ then }P_1\textsf{ else }Q_1$, $P'\equal P_1$ and $\mu\equal\tau$. We know $\emptyset\cdot\emptyset\rightvdash P:S$ could have been derived using only the rule \ltsrule{tIf} which means that it type-checks with any session type $S$ and $\emptyset\rightvdash A:\textsf{Bool}$, $\emptyset\cdot\emptyset\rightvdash P_1:S$ and $\emptyset\cdot\emptyset\rightvdash Q_1:S$. Implication \ref{sub-red:implication-2} holds since $\mu\equal\tau$ and $\emptyset\cdot\emptyset\rightvdash P_1:S$.
      
    \item[\ltsrule{pFls}] This case is analogous to the previous case for \ltsrule{pTru}.
  \end{description}
\end{proof}

\subsection{Synthesis Soundness}\label{s:soundness-proof}

In the following we distinguish two variants of the violation $\textsf{no}_P$:
\begin{itemize}
  \item $\nops$ when the violation is caused by a label violation (\Cref{fig:monitor-calculus});
  \item $\nopd$ when the violation is flagged by the synthesised monitor (\Cref{def:synth-function}).
\end{itemize}

We also distinguish two variants of the violation $\textsf{no}_E$:
\begin{itemize}
  \item $\noes$ when the violation is caused by a label violation (\Cref{fig:monitor-calculus});
  \item $\noed$ when the violation is flagged by the synthesised monitor (\Cref{def:synth-function}).
\end{itemize}

Correspondingly, we use the following version of the monitor synthesis function (\Cref{def:synth-function}), that supports the session type assertions in \Cref{fig:session-types}. This makes our soundness proof more general.

\begin{definition}\label{def:app:synth-function}
  The \textit{synthesis function} $\lsem \minus \rsem:S\mapsto M$ takes as input a session type $S$ and returns a monitor $M$. It is defined inductively, on the structure of the session type $S$ as follows: 
  
  \smallskip\centerline{
    $\lsem\stselAi{i}{I}\rsem\,\triangleq\,\triangleright\big\{\texttt{l}_i(x_i:\textsf{B}_i).\pif{(\isValueB{i}{x_i} \mathbin{\textsf{\&\&}}A'_i)}{\envsndOP \texttt{l}_i(x_i).\lsem S_i\rsem}{\nopd}\big\}_{i\in I}$
  }\medskip\centerline{
    $\lsem\stbraAi{i}{I}\rsem\,\triangleq\,\envrcvOP\big\{\texttt{l}_i(x_i:\textsf{B}_i).\pif{(\isValueB{i}{x_i} \mathbin{\textsf{\&\&}}A'_i)}{\triangleleft \texttt{l}_i(x_i).\lsem S_i\rsem}{\noed}\big\}_{i\in I}$
  }\medskip\centerline{
    $\lsem\textsf{rec } X.S\rsem\triangleq \mu_X.\lsem S\rsem \qquad\quad \lsem X\rsem \triangleq X \qquad\quad \lsem\textsf{end}\rsem\triangleq \textbf{0}$}
\end{definition}

Before we prove soundness, we require another lemma, stated in Lemma \ref{lemma:synthesis-subs}.
This is mostly used in conjunction with Lemma \ref{lemma:type-substitution} in the soundness proof. 
Lemma \ref{lemma:synthesis-subs} states that the synthesis of a session type with substitution is identical to the synthesis of the session type with the same substitution applied after the synthesis. 

\begin{lemma}[\textit{Synthesis Substitution Commutativity}]\label{lemma:synthesis-subs}~
  \begin{enumerate}
    \item\label{synth-subs-i} (\textit{Variable}) For all session types $S$, $\lsem S[\nicefrac{v}{x}]\rsem$ is equivalent to $\lsem S \rsem[\nicefrac{v}{x}]$.
    \item\label{synth-subs-ii} (\textit{Process Variable}) For all session types $S$, $\lsem S[\nicefrac{\textsf{rec }X.S'}{X}]\rsem$ is equivalent to $\lsem S \rsem[\nicefrac{\mu_X.\lsem S'\rsem}{X}]$.
  \end{enumerate}
\end{lemma}
\begin{proof}
  The proof for Cases \ref{synth-subs-i} and \ref{synth-subs-ii} follow by induction on the structure of $S$. 
\end{proof}

As the composite system transitions, the monitor M does not always transition in harmony with the session type S it was synthesised from. 
This is due to the way monitors are synthesised by the function in \Cref{def:app:synth-function}. 
Consequently, we require additional statements that weaken the contrapositive defined in \Cref{sec:monitor-soundness-maintext} for these particular cases when the monitor is out of synchrony with the session type. 
To better understand the necessity of these statements consider the examples
below which highlight these instances. 

\begin{example}\label{eg:comp-system-sound-cases-ii-iii}
  \begin{subequations}
    Consider the session type $\sef$, process $\pef$ and monitor $\mef$ defined as:
    \begin{align}
      \sef&\ \equal\ ?\labk{E}().S_f\qquad\text{where }S_f \equal\ !\labk{F}().\textsf{end}\\
      \pef&\ \equal\ \triangleright\labk{E}().\triangleleft\labk{F}().\bm{0}\\
      \mef&\ \equal\ \synth{\sef}\ \equal\ \envrcvOP \labk{E}().\pif{\true}{\triangleleft \labk{E}().\synth{S_f}}{\noed}
    \end{align}

    \noindent When the process $\pef$ is instrumented together with the monitor $\mef$ to form the composite system $\instr{\pef}{\synth{\sef}}$, we can observe the following behaviour:
    \begin{align}
      \instr{\pef}{\mef} & &&\\
      \label{transitions-ii-iii-1}&\xrightarrow{\envrcvOP \labk{E}()} \instr{\pef}{\pif{\true}{\triangleleft \labk{E}().\synth{S_f}}{\noed}} \qquad && \text{using \ltsrule{iIn}}\\
      \label{transitions-ii-iii-2}&\quad \xrightarrow{\tau} \instr{\pef}{\triangleleft \labk{E}().\synth{S_f}} \qquad && \text{using \ltsrule{iMon}}\\
      \label{pab-mab-cont}&\quad\quad \xrightarrow{\tau} \instr{\triangleleft\labk{F}().\bm{0}}{\synth{S_f}} \qquad && \text{using \ltsrule{iRcv}}
    \end{align}
    If we consider the sequence of transitions taken by the monitor as the system transitions in \ref{transitions-ii-iii-1}, \ref{transitions-ii-iii-2} and \ref{pab-mab-cont} respectively, it would be as shown below: 
    \begin{align}
      \label{mon-out-of-sync-0}\mef\equal\synth{\sef}& &&\\
      \label{mon-out-of-sync-1}&\xrightarrow{\envrcvOP \labk{E}()} \pif{\true}{\triangleleft \labk{E}().\synth{S_f}}{\noed}\qquad &&\text{using \ltsrule{mIn}}\\
      \label{mon-out-of-sync-2}&\quad\xrightarrow{\tau} \triangleleft \labk{E}().\synth{S_f}\qquad &&\text{using \ltsrule{mTru}}\\
      \label{mon-out-of-sync-3}&\quad\quad\xrightarrow{\triangleleft \labk{E}()} \synth{S_f}\qquad && \text{using \ltsrule{mSnd}} 
    \end{align}
    In \ref{mon-out-of-sync-0}, the monitor $\mef$ starts from the synthesis of the session type $\sef$ (\ie $\synth{\sef}$). 
    However, once it performs the first transition in \ref{mon-out-of-sync-1} (using the rule \ltsrule{mIn}), it comes in a form which does not directly correspond to the synthesis of any session type. 
    Consequently, the composite system is no longer of the form $\instr{P}{\synth{S}}$, rather it is in the state shown in \ref{transitions-ii-iii-1}. 
    Similarly, once it performs the second transition in \ref{mon-out-of-sync-2} (using the rule \ltsrule{mTru}), it is not in such a form either. 
    It only comes in a synthesised form after it performs the third transition in \ref{mon-out-of-sync-3} (using the rule \ltsrule{mSnd}). \exqed
  \end{subequations}
\end{example}

As shown in Example \ref{eg:comp-system-sound-cases-ii-iii}, when the composite system transitions, the monitor is not always in a form which directly corresponds to a synthesis of a session type. 
Consequently, once we get to prove the inductive step of the contrapositive, where the system $\instr{P}{\synth{S}}$ performs a transition, we will end up with a system of the form in \ref{transitions-ii-iii-1}, on which we cannot apply the inductive hypothesis. 
Therefore, we can strengthen the contrapositive for this intermediary state by adding another statement to it which would be similar to: $\lsem S \rsem \xrightarrow{\envrcvOP \lab(v)} M$ implies $M' \not\equal \textsf{no}^S_P$. 
However, this is still not enough since once we get to prove this statement, we will end up in a similar scenario (such as in \ref{pab-mab-cont}). 
Consequently, to cater for these two instances we strengthen the contrapositive with the two cases as shown below. 

\begin{proposition}\label{stmt:contrapstv-1}
  For all processes $P \in \textsc{Proc}$ and session types $S$, $\emptyset\cdot\emptyset\rightvdash P:S$ and $\instr{P}{M} \xRightarrow{e} \instr{P'}{M'}$ where:
  \begin{enumerate}
    \item $M\equal\synth{S}$ implies $M'\not\equal \textsf{no}^S_P$; and
    \item $\lsem S \rsem \xrightarrow{\envrcvOP \lab(v)} M$ s.t.~$\textbf{accepted}(\envrcvOP \lab(v), S)$ implies $M' \not\equal \textsf{no}^S_P$; and
    \item $\lsem S \rsem \xrightarrow{\envrcvOP \lab(v)} M_1 \xrightarrow{\tau} M$ s.t.~$\textbf{accepted}(\envrcvOP \lab(v), S)$ implies $M' \not\equal \textsf{no}^S_P$. 
  \end{enumerate}
\end{proposition}

Note that we add the requirement for the transitions to be accepted by the session type. 
This is because we assume that the monitor communicates with a statically type-checked component, and therefore all messages received from it (via $\envrcvOP$), would adhere to the session type.

Similarly to the previous cases, we once again strengthen Proposition \ref{stmt:contrapstv-1} with two additional cases as shown in Proposition \ref{prop:soundness} (Case \ref{sound-case-4} and Case \ref{sound-case-5}). 
Case \ref{sound-case-4} caters for the instance when the session type $S$ in Proposition \ref{stmt:contrapstv-1} is a selection type and the synthesised monitor performs an internal receive (via $\triangleright$) rather than an external input (via $\envrcvOP$). 
Similarly to Case \ref{sound-case-3}, we also add Case \ref{sound-case-5} to cater for the inductive step of Case \ref{sound-case-4} to make the inductive hypothesis to go through. 
The requirement for the monitor actions to be accepted by the dual of the session type (\ie $\textbf{accepted}(\triangleright\lab(v), \overline{S'})$), ensures that the monitor only performs accepted actions as it transitions independently of the process. 

\begin{proposition}\label{prop:soundness}
  For all $P \in$ Processes, and session types $S$, $\emptyset\cdot\emptyset\rightvdash P:S$ and $\instr{P}{M} \xRightarrow{e} \instr{P'}{M'}$ then:
  \begin{enumerate}
    \item\label{sound-case-1} $M\equal\lsem S\rsem$ implies $M' \not\equal \textsf{no}^S_P$; and

    \item\label{sound-case-2} $\lsem S \rsem \xrightarrow{\envrcvOP \lab(v)} M$ s.t.~$\textbf{accepted}(\envrcvOP \lab(v), S)$ implies $M' \not\equal \textsf{no}^S_P$; and

    \item\label{sound-case-3} $\lsem S \rsem \xrightarrow{\envrcvOP \lab(v)} M_1 \xrightarrow{\tau} M$ s.t.~$\textbf{accepted}(\envrcvOP \lab(v), S)$ implies $M' \not\equal \textsf{no}^S_P$; and

    \item\label{sound-case-4} $\exists S' \cdot S' \xrightarrow{\triangleleft \lab} S$ and $\lsem S' \rsem \xrightarrow{\triangleright\lab(v)} M$ s.t.~$\textbf{accepted}(\triangleright\lab(v), \overline{S'})$ implies $M' \not\equal \textsf{no}^S_P$; and

    \item\label{sound-case-5} $\exists S' \cdot S' \xrightarrow{\triangleleft \lab} S$ and $\lsem S' \rsem \xrightarrow{\triangleright\lab(v)} M_1 \xrightarrow{\tau} M$ s.t.~$\textbf{accepted}(\triangleright\lab(v), \overline{S'})$ implies $M' \not\equal \textsf{no}^S_P$.
  \end{enumerate}
\end{proposition}

\begin{proof} 
  We can state $\instr{P}{M} \xRightarrow{e} \instr{P'}{M'}$ as the statement $\instr{P}{M} \xtwoheadrightarrow{t} \instr{P'}{M'}$ where $t$ is equal to $e$ when we filter out all the $\tau$-transitions from $t$.
  The proof proceeds by lexicographical induction on (\textit{1}) the derivation $\emptyset\cdot\emptyset\rightvdash P:S$ and (\textit{2}) numerical induction on the length $n$ of $t$ (\ie $n\equal\vert\; t\;\vert$) in $\instr{P}{M} \xtwoheadrightarrow{t} \instr{P'}{M'}$.
  When the derivation is by the rule:
  \begin{description}
    \item[\ltsrule{tBra}] \begin{subequations} Therefore from the rule we know: 
    \begin{gather}
      \label{tbra-sound-p}P \equal \prcv{j}{J}\\
      \label{tbra-sound-s}S \equal \stbra{i}{I}\\
      \label{tbra-sound-premise}\forall i \in I, \exists j \in J \cdot \lab_j \equal \lab_i,\ \basetype_j\equal \basetype_i,\ \emptyset\cdot x_j:\basetype_j \rightvdash P_j:S_i,\ x_j:\basetype_j\rightvdash A_i:\textsf{Bool}\\
      \label{tbra-sound-synthesis}\lsem S \rsem \equal \synthbra
    \end{gather}
    \textbf{[$n\equal0$]} Therefore $P\equal P'$ and $M\equal M'$.
    \begin{enumerate}[label=\roman*.,ref=\roman*]
      \item If $M\equal\lsem S\rsem$, since $M\equal M'$, by \ref{tbra-sound-synthesis} we know that $M' \not\equal \textsf{no}^S_P$.
        
      \item\label{sound-tbra-bc-ii}\begin{subequations} 
        We assume that:
        \begin{gather}
          \label{sound-tbra-bc-1}\synth{S} \xrightarrow{\envrcvOP \lab(v)} M \text{ s.t. }\textbf{accepted}(\envrcvOP \lab(v), S)
        \end{gather}
        Since \ref{sound-tbra-bc-1} could only be derived with \ltsrule{mIn} we know from \ref{tbra-sound-synthesis} that $\lab(v)\equal\lab_m(v_m)$ for some $m \in I$ where: 
        \begin{gather}
          M\equal\condstmtinternal{m}
        \end{gather}
        Since $M\equal M'$ we know that $M' \not\equal \textsf{no}^S_P$. 
      \end{subequations}

      \item\label{sound-tbra-bc-iii}\begin{subequations} We assume that:
        \begin{gather}
          \synth{S} \xrightarrow{\envrcvOP \lab(v)} M_1 \xrightarrow{\tau} M  \text{ s.t. } \textbf{accepted}(\envrcvOP \lab(v), S)
        \end{gather}
        Similar to the previous case, since $\lsem S \rsem \xrightarrow{\envrcvOP \lab(v)} M_1$ could only be derived with \ltsrule{mIn} we know from \ref{tbra-sound-synthesis} that $\lab(v)\equal\lab_m(v_m)$ for some $m\in I$ where:
        \begin{gather}
          M_1 \equal \condstmtinternal{m}
        \end{gather}
        There are two possible transitions for $M_1 \xrightarrow{\tau} M$ with the rules:
        \begin{enumerate}[label=\arabic*.]
          \item \ltsrule{mTru} where $M \equal \triangleleft \lab_m(x_m).\lsem S_m \rsem[\nicefrac{v_m}{x_m}]$ and since $M\equal M'$ then $M' \not\equal \textsf{no}^S_P$.

          \item \ltsrule{MFls} where $M \equal \textsf{no}^D_E$ and since $M\equal M'$ then $M' \not\equal \textsf{no}^S_P$.
        \end{enumerate}
      \end{subequations}

      \item\label{sound-tbra-bc-iv}\begin{subequations} 
        We assume that: 
        \begin{gather}
          \label{sound-tbra-bc-iv-1}\exists S' \cdot S' \xrightarrow{\triangleleft \lab} S\\
          \label{sound-tbra-bc-iv-4}\lsem S' \rsem \xrightarrow{\triangleright\lab(v)} M\\
          \label{sound-tbra-bc-iv-7}\textbf{accepted}(\triangleright\lab(v), \overline{S'})
        \end{gather}
        Since \ref{sound-tbra-bc-iv-1} can only be derived with \ltsrule{sSel}, we know: 
        \begin{gather}
          \label{sound-tbra-bc-iv-2}S' \equal \stsel{k}{K}\\
          \label{sound-tbra-bc-iv-5}\lab\equal \lab_m \text{ for some } m \in K
        \end{gather}
        From \ref{sound-tbra-bc-iv-2} and by Definition \ref{def:synth-function}, we know that: 
        \begin{gather}
          \label{sound-tbra-bc-iv-3}\lsem S'\rsem\equal \synthsel{k}{K}
        \end{gather}
        Since we know \ref{sound-tbra-bc-iv-3}, \ref{sound-tbra-bc-iv-5} and \ref{sound-tbra-bc-iv-7}, we know that \ref{sound-tbra-bc-iv-4} can only be derived from the rule \ltsrule{mRcv} and therefore:
        \begin{gather}
          \label{sound-tbra-bc-iv-6}M \equal \pif{A_m[\nicefrac{v_m}{x_m}]}{\triangleleft \lab_m(v_m).\lsem S_m\rsem[\nicefrac{v_m}{x_m}]}{\noed}\\
          \emptyset\rightvdash v_m:B_m
        \end{gather}
        Since $M\equal M'$ and we know \ref{sound-tbra-bc-iv-6} then $M' \not\equal \textsf{no}^S_P$.
      \end{subequations}
  
      \item\label{sound-tbra-bc-v}\begin{subequations} 
        We assume that:
        \begin{gather}
          \label{sound-tbra-bc-v-1}\exists S' \cdot S' \xrightarrow{\triangleleft \lab} S\\
          \label{sound-tbra-bc-v-2}\lsem S' \rsem \xrightarrow{\triangleright\lab(v)} M_1 \xrightarrow{\tau} M \\
          \label{sound-tbra-bc-v-8}\textbf{accepted}(\triangleright\lab(v), \overline{S'})
        \end{gather}
        Similar to the previous case, since \ref{sound-tbra-bc-v-1} can only be derived with \ltsrule{sSel}, we know:
        \begin{gather}
          \label{sound-tbra-bc-v-3}S' \equal \stsel{k}{K}\\
          \label{sound-tbra-bc-v-4}\lab \equal \lab_m \text{ for some } m \in K
        \end{gather}
        From \ref{sound-tbra-bc-v-4} and by Definition \ref{def:synth-function}, we know that:
        \begin{gather}
          \label{sound-tbra-bc-v-5}\lsem S'\rsem\equal \synthsel{k}{K}
        \end{gather}
        Since we know \ref{sound-tbra-bc-v-5}, \ref{sound-tbra-bc-v-4} and \ref{sound-tbra-bc-v-8} we know that $\lsem S' \rsem \xrightarrow{\triangleright\lab(v)} M_1$ from \ref{sound-tbra-bc-v-2} could only be derived using \ltsrule{mRcv} and therefore:
        \begin{gather}
          \label{sound-tbra-bc-v-6}M_1 \equal \pif{A_m[\nicefrac{v_m}{x_m}]}{\triangleleft \lab_m(v_m).\lsem S_m\rsem[\nicefrac{v_m}{x_m}]}{\noed}\\
          \emptyset\rightvdash v_m:B_m
        \end{gather}
        From \ref{sound-tbra-bc-v-6}, we know that $M_1$ can only perform either of two actions by the rules:
        \begin{enumerate}[label=\alph*.]
            \item \ltsrule{mTru} where $M \equal \envsndOP \lab_m(v_m).\lsem S_m \rsem[\nicefrac{v_m}{x_m}]$ and since $M\equal M'$ then $M' \not\equal \textsf{no}^S_P$.

            \item \ltsrule{MFls} where $M \equal \textsf{no}^D_P$ and since $M\equal M'$ then $M' \not\equal \textsf{no}^S_P$.
        \end{enumerate}
      \end{subequations}
    \end{enumerate}

    \textbf{[$n\equal k+1$]} Therefore $\instr{P}{M} \xrightarrow{\mu} \instr{P_1}{M_1}\xtwoheadrightarrow{t'}\instr{P'}{M'}$ where $\vert\; t'\;\vert\equal k$.
      \begin{enumerate}[label=\roman*.]
        \item \begin{subequations}
          When $M\equal\lsem S\rsem$, from \ref{tbra-sound-p} and \ref{tbra-sound-synthesis} we know that the only possible action of $\instr{P}{M}$ is by the rule \ltsrule{iIn}, where:
            \begin{gather}
              \label{1-ic-i-1}M\equal\lsem S\rsem\xrightarrow{\envrcvOP \lab_m(v_m)}M_1\\
              \label{1-ic-i-2}\instr{P}{M} \xrightarrow{\envrcvOP \lab_m(v_m)}\instr{P}{M_1}
            \end{gather}
            Depending on $m$, \ref{1-ic-i-1} may be derived with two rules:
            \begin{enumerate}[label=\alph*.]
              \item \ltsrule{mIn} when $m \in I$ such that:
              \begin{gather}
                \label{sound-tbra-i-ic-1}M_1 \equal \condstmtinternal{m}\\
                \label{sound-tbra-i-ic-2}\emptyset\rightvdash v_m:B_m
              \end{gather}
              Since we know \ref{tbra-sound-s}, \ref{sound-tbra-i-ic-2} and $m \in I$, by Definition \ref{def:accepted-actions}, we know that: 
              \begin{gather}
                \label{1-ic-i-3}\textbf{accepted}(\envrcvOP \lab_m(v_m), S)
              \end{gather} 
              Therefore, since from the assumption we know that $\emptyset\cdot\emptyset\rightvdash P:S$ and we know \ref{1-ic-i-1}, \ref{1-ic-i-3}, and $\instr{P}{M_1} \xtwoheadrightarrow{t'} \instr{P'}{M'}$, using the inductive hypothesis for Case \ref{sound-case-2} we obtain $M'\not\equal \textsf{no}^S_P$.

            \item \ltsrule{mEV} when $m \not\in I$ or $\emptyset\not\rightvdash v_k:B_k$ such that $M_1 \equal \textsf{no}^S_E$. However, $\instr{P}{\textsf{no}^S_E} \not{\xtwoheadrightarrow{t'}}$ and therefore $M'$ must be equal to $M_1$, which means that $M'\not\equal \textsf{no}^S_P$.
          \end{enumerate}
        \end{subequations}

        \item \begin{subequations} 
          We assume that: 
          \begin{gather}
            \label{sound-tbra-ic-ii-1}\synth{S} \xrightarrow{\envrcvOP \lab(v)} M \text{ s.t. }\textbf{accepted}(\envrcvOP \lab(v), S)
          \end{gather}
          As we show in the base case of this case (Case \ref{sound-tbra-bc-ii}) we know that \ref{sound-tbra-ic-ii-1} can only be derived by the rule \ltsrule{mIn} where for some $m \in I$: 
          \begin{gather}
            \label{tbra-labl-equal-ic-ii}\lab(v) \equal \lab_m(v_m)\\
            \label{tbra-m-structure-1}M\equal \condstmtinternal{m}\\
            \label{tbra-vm-bm}\emptyset\cdot\emptyset\rightvdash v_m:B_m
          \end{gather}
          Due to the structure of $P$ in \ref{tbra-sound-p} and the structure of $M$ in \ref{tbra-m-structure-1}, the only possible action that $\instr{P}{M}$ can perform is by the rule \ltsrule{iMon}, where:
          \begin{gather}
            P\equal P_1\\
            \label{tbra-m-tau-transition-ic-ii}M \xrightarrow{\tau} M_1\\
            \label{tbra-pm-transition-ic-7}\langle P;M\rangle \xrightarrow{\tau} \langle P;M_1 \rangle
          \end{gather}
          Due to the structure of $M$ in \ref{tbra-m-structure-1}, the transition \ref{tbra-m-tau-transition-ic-ii} could only be derived with either of the following two rules: 
          \begin{enumerate}[label=\alph*.]
            \item \ltsrule{mTru} and therefore we know that: 
            \begin{gather}
              \label{tbra-ic-ii-6} M_1 \equal \triangleleft \lab_m(v_m).\lsem S_m \rsem[\nicefrac{v_m}{x_m}]
            \end{gather}
            Since we know \ref{sound-tbra-ic-ii-1} and \ref{tbra-m-tau-transition-ic-ii}, we therefore know that:
            \begin{gather}
              \label{1-ic-ii-1}\lsem S \rsem \xrightarrow{\envrcvOP \lab_m(v_m)} M \xrightarrow{\tau} M_1
            \end{gather}
            Moreover, from \ref{tbra-pm-transition-ic-7} we know:
            \begin{gather}
              \label{1-ic-ii-2}\langle P;M_1 \rangle \xtwoheadrightarrow{t'} \langle P';M'\rangle
            \end{gather}
            From \ref{sound-tbra-ic-ii-1} and \ref{tbra-labl-equal-ic-ii} we know that:
            \begin{gather}
              \label{tbra-accepted-ic-ii}\textbf{accepted}(\envrcvOP \lab_m(v_m), S)
            \end{gather}
            Therefore, since from the assumption we know $\emptyset\cdot\emptyset\rightvdash P:S$ and we also know \ref{1-ic-ii-2}, \ref{tbra-accepted-ic-ii} and \ref{1-ic-ii-1} we can apply the inductive hypothesis for Case \ref{sound-case-3} to obtain $M' \not\equal \textsf{no}^S_P$ as required.

            \item \ltsrule{mFls} and therefore we know that $M_1 \equal \textsf{no}^D_E$. Since we know that $\langle P;M \rangle \xrightarrow{\tau} \langle P;\textsf{no}^D_E \rangle$ and since $\langle P;\textsf{no}^D_E \rangle \not{\xtwoheadrightarrow{}}$ the implication holds trivially.
          \end{enumerate}
        \end{subequations}

        \item \begin{subequations} 
          We assume that:
          \begin{gather}
            \label{sound-tbra-ic-iii-1}\lsem S \rsem \xrightarrow{\envrcvOP \lab(v)} M'_1 \xrightarrow{\tau} M \text{ s.t. }\textbf{accepted}(\envrcvOP \lab(v), S)
          \end{gather}
          As shown in the previous case, $\lsem S \rsem \xrightarrow{\envrcvOP \lab(v)} M'_1$ could only be derived by rule \ltsrule{mIn} (\ref{tbra-labl-equal-ic-ii}, \ref{tbra-m-structure-1}, \ref{tbra-vm-bm}) where for some $m \in I$: 
          \begin{gather}
            \lab(v)\equal\lab_m(v_m)\\
            \label{tba-ic-iii-m1prime}M'_1\equal \condstmtinternal{m}\\
            \label{tbra-ic-iii-vmbm}\emptyset\cdot\emptyset\rightvdash v_m:B_m
          \end{gather}
          Due to the structure of $M'_1$ in \ref{tba-ic-iii-m1prime}, the only actions it can perform to match \ref{sound-tbra-ic-iii-1} are by the rules:
          \begin{enumerate}[label=\alph*.]
            \item \ltsrule{mTru} where $M'_1 \xrightarrow{\tau} M$ and $M\equal \triangleleft \lab_m(v_m).\lsem S_k \rsem[\nicefrac{v_m}{x_m}]$. Due to the structure of $P$ in \ref{tbra-sound-p} and the structure of $M$, the only action $\langle P;M\rangle$ can perform is by the rule \ltsrule{iRcv}, where:
            \begin{gather}
              \label{1-ic-iii-1}P\xrightarrow{\triangleright\,\lab_m(v_m)} P_1 \text{ by \ltsrule{pRcv} where }P_1\equal P_m[\nicefrac{v_m}{x_m}]\\
              \label{1-ic-iii-2}M\xrightarrow{\triangleleft \lab_m(v_m)} M_1 \text{ by \ltsrule{mSnd} where } M_1\equal\lsem S_m\rsem[\nicefrac{v_m}{x_m}]\\
              \label{1-ic-iii-3}\langle P;M \rangle \xrightarrow{\tau} \langle P_1;M_1 \rangle
            \end{gather}
            Since we know that $m \in I$, $\emptyset\cdot\emptyset\rightvdash v_m:B_m$ from \ref{tbra-ic-iii-vmbm} and the structure of $S$ in \ref{tbra-sound-s}, through Definition \ref{def:accepted-actions}, we know that:
            \begin{gather}
              \label{tbra-accepted-iii-ic}\accepted{\triangleright\,\lab_m(v_m)}{S}
            \end{gather}
            From our assumption we know that $\emptyset\cdot\emptyset\rightvdash P:S$ and together with \ref{1-ic-iii-1} and \ref{tbra-accepted-iii-ic} by applying subject reduction (Proposition \ref{prop:sub-red}) we obtain:
            \begin{gather}
              S \xrightarrow{\triangleright\,\lab_m}S_m \text{ s.t. } \match{\triangleright\lab_m()}{\triangleright\lab_m}\\
              \label{1-ic-iii-4}\emptyset\cdot\emptyset\rightvdash P_1:S_m
            \end{gather}
            We can apply the lemma for type variable instantiation (Lemma \ref{lemma:type-substitution}) on \ref{1-ic-iii-4} to obtain:
            \begin{gather}
              \label{1-ic-iii-5}\emptyset\cdot\emptyset\rightvdash P_1:S_k[\nicefrac{v_k}{x_k}] \ (\text{since } \emptyset[\nicefrac{v_k}{x_k}])\equal\emptyset
            \end{gather}
            Moreover, from \ref{1-ic-iii-2} we know $M_1\equal\lsem S_k\rsem[\nicefrac{v_k}{x_k}]$ and by the synthesis substitution lemma (Lemma \ref{lemma:synthesis-subs}) we obtain that:
            \begin{gather}
              \label{1-ic-iii-6}M_1\equal\lsem S_k\rsem[\nicefrac{v_k}{x_k}]\equal\lsem S_k[\nicefrac{v_k}{x_k}] \rsem
            \end{gather}
            Therefore, since we know \ref{1-ic-iii-5} and \ref{1-ic-iii-6} and we know that $\langle P_1;M_1 \rangle \xtwoheadrightarrow{t'} \langle P';M'\rangle$, we can apply the inductive hypothesis for Case \ref{sound-case-1} to obtain $M' \not\equal \textsf{no}^S_P$.

            \item \ltsrule{mFls} where $M'_1 \xrightarrow{\tau} \textsf{no}^D_E$. However, $\langle P;\textsf{no}^D_E \rangle \not{\rightarrow}$ and implication holds trivially.
          \end{enumerate}
        \end{subequations}

        \item\begin{subequations}\label{sound-tbra-ic-iv} 
          We assume that:
          \begin{gather}
            \label{sound-tbra-ic-iv-1}\exists S' \cdot S' \xrightarrow{\triangleleft \lab} S\\
            \label{sound-tbra-ic-iv-4}\lsem S' \rsem \xrightarrow{\triangleright\lab(v)} M\\
            \label{sound-tbra-ic-iv-6}\textbf{accepted}(\triangleright\lab(v), \overline{S'})
          \end{gather}
          As we show in the base case of this case (Case \ref{sound-tbra-bc-iv}) we know that \ref{sound-tbra-ic-iv-1} could only be derived using \ltsrule{sSel} and therefore: 
          \begin{gather}
            \label{sound-tbra-ic-iv-7}\lab\equal \lab_m \text{ for some } m \in K\\
            \label{sound-tbra-ic-iv-5}S' \equal \stsel{k}{K}
          \end{gather}
          From \ref{sound-tbra-ic-iv-5} and Definition \ref{def:synth-function}, we know that:
          \begin{gather}
            \label{sound-tbra-ic-iv-8}\lsem S'\rsem\equal \synthsel{k}{K}
          \end{gather}
          Since we know \ref{sound-tbra-ic-iv-8}, \ref{sound-tbra-ic-iv-6} and \ref{sound-tbra-ic-iv-7}, we know that \ref{sound-tbra-ic-iv-4} could only be derived using \ltsrule{mRcv} and therefore:
          \begin{gather}
            \label{sound-tbra-ic-iv-9}M \equal \pif{A_m[\nicefrac{v_m}{x_m}]}{\triangleleft \lab_m(v_m).\lsem S_m\rsem[\nicefrac{v_m}{x_m}]}{\noed}\\
            \emptyset\rightvdash v_m:B_m
          \end{gather}
          From the structure of $P$ in \ref{tbra-sound-p} and the structure of $M$ in \ref{sound-tbra-ic-iv-9} we know that $\langle P;M \rangle \xrightarrow{\mu} \langle P_1;M_1\rangle$ can only be derived with the rule \ltsrule{iMon}, where $\mu \equal \tau$, $P \equal P_1$ and with the rule:
          \begin{enumerate}[label=\bfseries\alph*.]
            \item \ltsrule{mTru} $M_1\equal \envsndOP \lab_m(v_m).\lsem S_m \rsem[\nicefrac{v_m}{x_m}]$. Since from the assumption we know that $\emptyset\cdot\emptyset\rightvdash P:S$ and we know \ref{sound-tbra-ic-iv-1}, \ref{sound-tbra-ic-iv-7}, \ref{sound-tbra-ic-iv-6}, $\langle P;M_1\rangle \xtwoheadrightarrow{t'} \langle P';M'\rangle$, $\lsem S' \rsem \xrightarrow{\triangleright\,\lab_m(v_m)} M \xrightarrow{\tau} M_1$, we can apply the inductive hypothesis for Case \ref{sound-case-5} to obtain $M' \not\equal \textsf{no}^S_P$.

            \item \ltsrule{mFls} $M_1\equal \textsf{no}^D_P$. However, $\langle P;M_1\rangle \not{\rightarrow}$ and implication holds trivially.
          \end{enumerate}
        \end{subequations}

        \item\label{sound-tbra-ic-v}
        \begin{subequations}  
          We assume that:
          \begin{gather}
            \label{sound-tbra-ic-v-1}\exists S' \cdot S' \xrightarrow{\triangleleft \lab} S\\
            \label{sound-tbra-ic-v-2}\lsem S' \rsem \xrightarrow{\triangleright\lab(v)} M_1 \xrightarrow{\tau} M \\
            \label{sound-tbra-ic-v-3}\textbf{accepted}(\triangleright\lab(v), \overline{S'})
          \end{gather}
          As shown in the previous case, we can infer that:
          \begin{gather}
            \label{sound-tbra-ic-v-4}\lab\equal \lab_m \text{ for some } m \in K\\
            \label{sound-tbra-ic-v-5}S' \equal \stsel{k}{K}\\
            \label{sound-tbra-ic-v-6}\lsem S'\rsem\equal \synthsel{k}{K}\\
            \label{sound-tbra-ic-v-7}M_1 \equal \pif{A_m[\nicefrac{v_m}{x_m}]}{\triangleleft \lab_m(v_m).\lsem S_m\rsem[\nicefrac{v_m}{x_m}]}{\noed}\\
            \label{sound-tbra-ic-v-8}\emptyset\rightvdash v_m:B_m
          \end{gather}
          $M_1$ in \ref{sound-tbra-ic-v-7} can perform either of the following actions with the rules: 
          \begin{enumerate}[label=\alph*.]
            \item \ltsrule{mTru} where $M \equal \envsndOP \lab_m(v_m).\lsem S_m \rsem[\nicefrac{v_m}{x_m}]$ for an $m \in K$. For $\langle P;M \rangle \xrightarrow{\mu} \langle P_1;M_1\rangle$, only the rule \ltsrule{iOut} may be applied, where $P\equal P_1$ and:
            \begin{gather}
              \label{1-ic-v-1}M_1\equal \lsem S_m \rsem[\nicefrac{v_m}{x_m}]\\
              \label{1-ic-v-2}\langle P;M \rangle \xrightarrow{\envsndOP \lab_m(v_m)} \langle P;M_1\rangle \xtwoheadrightarrow{t'} \langle P';M'\rangle
            \end{gather}
            Since we know \ref{sound-tbra-ic-v-1} and \ref{sound-tbra-ic-v-4} we know that $S \equal S_m$ by the rule \ltsrule{sSel} and therefore from \ref{1-ic-v-1} we know that:
            \begin{gather}
              \label{1-ic-v-3}M_1\equal \lsem S \rsem[\nicefrac{v_m}{x_m}]
            \end{gather}
            From the assumption we know $\emptyset\cdot\emptyset\rightvdash P:S$ to which we can apply the type instantiation lemma (Lemma \ref{lemma:type-substitution}) to obtain:
            \begin{gather}
              \label{1-ic-v-4}\emptyset\cdot\emptyset\rightvdash P:S[\nicefrac{v_m}{x_m}]\text{  since } \emptyset[\nicefrac{v_k}{x_k}]\equal\emptyset
            \end{gather}
            Moreover, from \ref{1-ic-v-3} we know $M_1\equal \lsem S \rsem[\nicefrac{v_m}{x_m}]$ and by the synthesis substitution lemma (Lemma \ref{lemma:synthesis-subs}) we know that:
            \begin{gather}
              \label{1-ic-v-5}M_1\equal \lsem S\rsem[\nicefrac{v_m}{x_m}] \equal \lsem S[\nicefrac{v_m}{x_m}]\rsem
            \end{gather}
            Therefore, since we know \ref{1-ic-v-4} and \ref{1-ic-v-5} and from \ref{1-ic-v-2} we know $\langle P;M_1\rangle \xtwoheadrightarrow{t'} \langle P';M'\rangle$, we can apply the inductive hypothesis for Case \ref{sound-case-1} to obtain $M' \not\equal \textsf{no}^S_P$.

            \item \ltsrule{mFls} we know $M \equal \textsf{no}^D_P$ and since $\langle P;M\rangle \not{\rightarrow}$, implication holds trivially.
          \end{enumerate}
        \end{subequations}
      \end{enumerate}
    \end{subequations}

    \item[\ltsrule{tSel}] \begin{subequations} Therefore we know that $P \equal \pwrtk{j}.P_j$ and $S \equal \stsel{i}{I}$. Moreover:
      \begin{gather}
          \label{tsel-1}\exists i \in I \cdot \lab_j \equal \lab_i,\quad \emptyset\rightvdash v_j:\basetype_i,\quad\emptyset\cdot\emptyset \rightvdash P_j:S_i,\quad x_i:\basetype_i \rightvdash A_i:\textsf{Bool}\\
          \label{tsel-2}\lsem S \rsem \equal \synthsel{i}{I}
      \end{gather}
      \textbf{[$n\equal0$]} Therefore $P\equal P'$ and $M\equal M'$.
      \begin{enumerate}[label=\roman*.]
          \item We know $M\equal\lsem S\rsem$, and since $M\equal M'$, by \ref{tsel-2} $M' \not\equal \textsf{no}^S_P$.
      
          \item Since $\lsem S \rsem \not{\xrightarrow{\envrcvOP \lab(v)}}$ the implication holds trivially.

          \item Analogous to the previous case since $\lsem S \rsem \not{\xrightarrow{\envrcvOP \lab(v)}}$.

          \item The base case for $S' \xrightarrow{\triangleleft \lab_m} S$ and $\lsem S' \rsem \xrightarrow{\triangleright\,\lab_m(v_m)} M$ is identical to the base case \ref{sound-tbra-bc-iv} for the rule \ltsrule{tBra}.

          \item Similarly, the base case for $S' \xrightarrow{\triangleleft \lab_m} S$ and $\lsem S' \rsem \xrightarrow{\triangleright\,\lab_m(v_m)} M_1 \xrightarrow{\tau} M$ is identical to the base case \ref{sound-tbra-bc-v} for the rule \ltsrule{tBra}.
      \end{enumerate}
      \textbf{[$n\equal k+1$]} Therefore $\instr{P}{M} \xrightarrow{\mu} \instr{P_1}{M_1}\xtwoheadrightarrow{t'}\instr{P'}{M'}$ where $\vert\; t'\;\vert\equal k$.
      \begin{enumerate}[label=\roman*.,ref=\roman*]
        \item \label{tsel-ic-i} When $M\equal\lsem S\rsem$, the only possible action the composite system $\langle P;M\rangle$ can perform is by the rule \ltsrule{iSnd}, where:
        \begin{gather}
          \label{2-ic-1-1}P \xrightarrow{\triangleleft \lab_j(v_j)} P_1 \quad\text{where }P_1\equal P_j \text{ by \ltsrule{pSnd}}\\
          \label{2-ic-1-2}M \equal \lsem S \rsem\xrightarrow{\triangleright\,\lab_j(v_j)} M_1\\
          \label{2-ic-1-6}\langle P;M\rangle \xrightarrow{\tau} \langle P_1;M_1\rangle \xtwoheadrightarrow{t'} \langle P';M'\rangle
        \end{gather}
        where \ref{2-ic-1-2} can only be derived by the rule \ltsrule{mRcv} since by \ref{tsel-1}, $j \in I$ and $\emptyset\rightvdash v_j:\basetype_j$ and therefore $M_1 \equal\condstmtexternal{j}$. 
        Moreover, since $\overline{S} \equal \&\{?\lab_i(x_i:\basetype_i)[A_i].\overline{S_i}\}_{i \in I}$ and we know that $j \in I$ and $\emptyset\rightvdash v_j:\basetype_j$, by Definition \ref{def:accepted-actions}, we know that:
        \begin{gather}
          \label{2-ic-1-5}\textbf{accepted}(\triangleright\,\lab_j(v_j), \overline{S})
        \end{gather}
        Since we know \ref{2-ic-1-5}, by Definition \ref{def:accepted-actions} and the dual of the session type $\overline{S}$ we know that:
        \begin{gather}
          \label{3-ic-accepted-8}\textbf{accepted}(\triangleleft \lab_j(v_j), S)
        \end{gather}
        Since from the assumption we know that $\emptyset\cdot\emptyset\rightvdash P:S$, and we know \ref{2-ic-1-1} such that \ref{3-ic-accepted-8}, we can apply subject reduction (Proposition \ref{prop:sub-red}), to obtain:
        \begin{gather}
          \label{2-ic-1-3}S \xrightarrow{\triangleleft \lab_j} S_j \text{ for } j \in I\\
          \label{2-ic-1-4}\emptyset\cdot\emptyset\rightvdash P_j:S_j
        \end{gather}
        Since we know \ref{2-ic-1-4}, \ref{2-ic-1-3} and by \ref{2-ic-1-2} we know that $\lsem S \rsem \xrightarrow{\triangleright\,\lab_j(v_j)} M_1$ such that \ref{2-ic-1-5}, together with $\langle P_1;M_1\rangle \xtwoheadrightarrow{t'} \langle P';M'\rangle$ from \ref{2-ic-1-6}, we can apply the inductive hypothesis of Case \ref{sound-case-4} to obtain $M'\not\equal \textsf{no}^S_P$.

        \item Since $\lsem S \rsem \not{\xrightarrow{\envrcvOP \lab(v)}}$ the implication holds trivially.

        \item Analogous to the previous case since $\lsem S \rsem \not{\xrightarrow{\envrcvOP \lab(v)}}$.
        
        \item The inductive case for $S' \xrightarrow{\triangleleft \lab_m} S$ and $\lsem S' \rsem \xrightarrow{\triangleright\,\lab_m(v_m)} M$ is identical to the inductive case of Case \ref{sound-tbra-ic-iv} for the rule \ltsrule{tBra}.

        \item Similarly, the inductive case for $S' \xrightarrow{\triangleleft \lab_m} S$ and $\lsem S' \rsem \xrightarrow{\triangleright\,\lab_m(v_m)} M_1 \xrightarrow{\tau} M$ is identical to the inductive case of Case \ref{sound-tbra-ic-v} for the rule \ltsrule{tBra}.
      \end{enumerate}
  \end{subequations}

  \item[\ltsrule{tRec}] \begin{subequations} Therefore we know that $P \equal \mu_X.Q$ that type-checks with any $S$. Moreover:
      \begin{gather}
          \label{tRec-1}X:S\cdot\emptyset\rightvdash Q:S
      \end{gather}
      \textbf{[$n\equal0$]} Therefore $P\equal P'$ and $M\equal M'$.\\
      All base cases for \ltsrule{tRec} are analogous to the base cases for the rule \ltsrule{tBra}.

      \textbf{[$n\equal k+1$]} Therefore $\instr{P}{M} \xrightarrow{\mu} \instr{P_1}{M_1}\xtwoheadrightarrow{t'}\instr{P'}{M'}$ where $\vert\; t'\;\vert\equal k$.
      \begin{enumerate}[label=\roman*.]
        \item When $M\equal\lsem S\rsem$, the only possible actions for $\langle P;M\rangle \xrightarrow{\mu} \langle P_1;M_1\rangle$ due to the structure of $P$ and the strucutre of $M$, are by the rules:
        \begin{enumerate}[label=\alph*.]
          \item \ltsrule{iProc} where $P \xrightarrow{\tau} P_1$ s.t.~$P_1\equal Q[\nicefrac{\mu_X.Q}{X}]$ by \ltsrule{pRec} and $M \equal M_1 \equal \lsem S \rsem$. From the assumption we know that $\emptyset\cdot\emptyset\rightvdash P:S$ and since $\textbf{accepted}(\tau,S)$ always holds, by applying subject reduction (Proposition \ref{prop:sub-red}) we obtain $\emptyset\cdot\emptyset\rightvdash P_1:S$. We now have $\langle P_1;\lsem S\rsem \rangle \xtwoheadrightarrow{t'} \langle P';M' \rangle$ and by the inductive hypothesis of Case \ref{sound-case-1} we obtain that $M'\not\equal \textsf{no}^S_P$.

          \item \ltsrule{iMon} where $S\equal \textsf{rec }Y.S_1$ since:
          \begin{gather}
            \label{3-ic-1-3}M \equal \lsem S \rsem \equal \mu_Y.\lsem S_1\rsem \text{ and }P\equal P_1\\
            \label{3-ic-1-4}M \equal \lsem S \rsem \xrightarrow{\tau}M_1 \text{ by \ltsrule{mRec} where }  M_1\equal\lsem S_1 \rsem[\nicefrac{\mu_Y.\lsem S_1\rsem}{Y}]\\
            \label{3-ic-1-5}\langle P;M \rangle \xrightarrow{\tau} \langle P;M_1 \rangle \xtwoheadrightarrow{t'} \langle P';M'\rangle
          \end{gather}
          From the assumption we know $\emptyset\cdot\emptyset\rightvdash P:\textsf{rec }Y.S_1$, which since we consider $\textsf{rec }Y.S_1$ to be definitionally equal to $S_1[\nicefrac{\textsf{rec }Y.S_1}{Y}]$ we know that:
          \begin{gather}
            \label{3-ic-1-6}\emptyset\cdot\emptyset\rightvdash P:S_1[\nicefrac{\textsf{rec }Y.S_1}{Y}]
          \end{gather}
          From \ref{3-ic-1-4} and the synthesis substitution lemma (Lemma \ref{lemma:synthesis-subs}) we know that:
          \begin{gather}
            \label{3-ic-1-6-1}M_1\equal\lsem S_1 \rsem[\nicefrac{\mu_Y.\lsem S_1\rsem}{Y}]\equal\lsem S_1[\nicefrac{\textsf{rec }Y.S_1}{Y}] \rsem
          \end{gather}
          By \ref{3-ic-1-6}, \ref{3-ic-1-6-1} and since from \ref{3-ic-1-5} we know $\langle P;M_1 \rangle \xtwoheadrightarrow{t'} \langle P';M'\rangle$ by applying the inductive hypothesis of Case \ref{sound-case-1} we obtain $M' \not\equal \textsf{no}^S_P$.

          \item \ltsrule{iIn} where $S \equal \stbra{i}{I}$ since: 
          \begin{gather}
            \label{3-ic-1-7}\lsem S \rsem \equal \synthbra\text{ and }P\equal P_1\\
            \label{3-ic-1-9}M\equal\lsem S \rsem \xrightarrow{\envrcvOP \lab_j(v_j)}M_1\\
            \label{3-ic-1-8}\langle P;M \rangle \xrightarrow{\envrcvOP \lab_j(v_j)} \langle P;M_1 \rangle \xtwoheadrightarrow{t'} \langle P';M'\rangle
          \end{gather}
          Depending on $j$, \ref{3-ic-1-9} can be derived with the rules:
          \begin{enumerate}[label=\arabic*.]
            \item \ltsrule{mIn} where $M_1 \equal \condstmtinternal{j}$ for $j \in I$ and $\emptyset\rightvdash v_j:\basetype_j$ and hence from the definition we can infer:
            \begin{gather}
              \label{3-ic-1-10}\textbf{accepted}(\envrcvOP \lab_j(v_j), S)
            \end{gather}
            Since we know \ref{3-ic-1-9}, \ref{3-ic-1-10}, from the assumption we know that $\emptyset\cdot\emptyset\rightvdash P:S$ and from \ref{3-ic-1-8} we know $\langle P;M_1 \rangle \xtwoheadrightarrow{t'} \langle P';M'\rangle$, we can apply the inductive hypothesis of Case \ref{sound-case-2} to obtain $M' \not\equal \textsf{no}^S_P$.

            \item \ltsrule{mEV} where $M_1 \equal \textsf{no}^S_E$ for $j \not\in I$ or $\emptyset\cdot\emptyset\not\rightvdash v_j:\basetype_j$. Since we know that $M_1 \not{\rightarrow}$, then $M'\equal M_1$ and hence $M' \not\equal \textsf{no}^S_P$.
          \end{enumerate}
        \end{enumerate}

        \item For the derivation $\lsem S \rsem \xrightarrow{\envrcvOP \lab(v)} M$ s.t.~$\textbf{accepted}(\envrcvOP \lab(v), S)$ we know that:
        \begin{gather}
          \label{trec-ic-ii-1}\lsem S \rsem \equal \synthbra \text{ where:}\\
          \label{trec-ic-ii-2}S \equal \stbra{i}{I}
        \end{gather}
        Moreover since $\lsem S \rsem \xrightarrow{\envrcvOP \lab(v)} M$ s.t.~$\textbf{accepted}(\envrcvOP \lab(v), S)$ can only be derived with the rule \ltsrule{mIn} we know:
        \begin{gather}
          \label{trec-ic-ii-4}M \equal \condstmtinternal{j} \text{ for some } j \in I \text{ and } \emptyset\rightvdash v_j:\basetype_j
        \end{gather}
        To derive $\langle P;M\rangle \xrightarrow{\mu} \langle P_1;M_1 \rangle$, due to the structures of $P$ and $M$, either of these rules must be applied:
        \begin{enumerate}[label=\alph*.]
          \item\label{trec-ic-ii-a} \ltsrule{iProc} where $M \equal M_1$ and:
          \begin{gather}
            \label{trec-ic-ii-5}P \xrightarrow{\tau} P_1 \text{ by \ltsrule{pRec} where }P_1 \equal Q[\nicefrac{\mu_X.Q}{X}]\\
            \label{trec-ic-ii-6}\langle P;M\rangle \xrightarrow{\mu}\langle P_1;M \rangle \xtwoheadrightarrow{t'} \langle P';M'\rangle
          \end{gather}
          From the assumption we know that $\emptyset\cdot\emptyset\rightvdash P:S$ and since we know \ref{trec-ic-ii-5} and $\textbf{accepted}(\tau, S)$ always holds, by applying subject reduction (Proposition \ref{prop:sub-red}) we obtain:
          \begin{gather}
            \label{thrm:soundness:trec-ic-iii-1}\emptyset\cdot\emptyset\rightvdash P_1:S \text{  since  } \mu\equal\tau
          \end{gather}
          Since we know \ref{thrm:soundness:trec-ic-iii-1} and from the assumption we know $\lsem S \rsem \xrightarrow{\envrcvOP \lab(v)} M$ s.t.~$\textbf{accepted}(\envrcvOP \lab(v), S)$, and also from \ref{trec-ic-ii-6} we know $\langle P_1;M \rangle \xtwoheadrightarrow{t'} \langle P';M'\rangle$, we can apply the inductive hypothesis of Case \ref{sound-case-2} to obtain $M' \not\equal \textsf{no}^S_P$.\label{iproc-1}
          
          \item \ltsrule{iMon} where $M \xrightarrow{\tau} M_1$ which can be derived from the rules \ltsrule{mTru} or \ltsrule{mFls}, and $P\equal P_1$. In both cases, since from the assumption we know that $\emptyset\cdot\emptyset\rightvdash P:S$, and we also know that $\langle P;M_1 \rangle \xtwoheadrightarrow{t'} \langle P';M' \rangle$ where $\lsem S \rsem \xrightarrow{\envrcvOP \lab(v)} M \xrightarrow{\tau} M_1$, the inductive hypothesis for Case \textbf{iii} may be applied, which implies that $M'\not\equal \textsf{no}^S_P$.
        \end{enumerate}

        \item For the derivation $\lsem S \rsem \xrightarrow{\envrcvOP \lab(v)} M'_1 \xrightarrow{\tau} M$ s.t.~$\textbf{accepted}(\envrcvOP \lab(v), S)$, we know, from the previous case, that it can only be derived by rule \ltsrule{mIn} where $M'_1\equal \condstmtinternal{j}$ for some $j \in I$ and $\emptyset\rightvdash v_k:B_k$ (\ref{trec-ic-ii-4}). The actions $M'_1$ can perform are by the rules:
        \begin{enumerate}[label=\alph*.]
            \item \ltsrule{mTru} where $M'_1 \xrightarrow{\tau} M$ and $M\equal \triangleleft \lab_j(v_j).\lsem S_j \rsem[\nicefrac{v_j}{x_j}]$. However, since $M \not{\xrightarrow{\triangleleft \lab_k(v_j)}}$ due to $P\not{\xrightarrow{\triangleright\,\lab_j(v_j)}}$, the only rule that can be applied for the derivation $\langle P;M\rangle \xrightarrow{\mu} \langle P_1;M_1 \rangle$ is \ltsrule{iProc} (similarly to Case ii.\ref{trec-ic-ii-a} of the rule \ltsrule{tRec}), where $M\equal M_1$ and:
            \begin{gather}
              \label{trec-ic-iii-1}P \xrightarrow{\tau} P_1 \text{ by \ltsrule{pRec} where } P_1\equal Q[\nicefrac{\mu_X.Q}{X}]\\
              \label{trec-ic-iii-3}\langle P;M\rangle \xrightarrow{\tau} \langle P_1;M \rangle \xtwoheadrightarrow{t'} \langle P';M'\rangle
            \end{gather}
            Since we know that \ref{trec-ic-iii-1} and $\textbf{accepted}(\tau,S)$ always holds, and from the assumption we know $\emptyset\cdot\emptyset\rightvdash P:S$, we can apply subject reduction (Proposition \ref{prop:sub-red}) to obtain:
            \begin{gather}
              \label{trec-ic-iii-4}\emptyset\cdot\emptyset\rightvdash P_1:S \text{ since } P \xrightarrow{\tau} P_1
            \end{gather}
            Since from the assumption we know $\lsem S \rsem \xrightarrow{\envrcvOP \lab(v)} M'_1 \xrightarrow{\tau} M$ s.t.~$\textbf{accepted}(\envrcvOP \lab(v), S)$ and we also know \ref{trec-ic-iii-4} and $\langle P_1;M\rangle \xtwoheadrightarrow{t'} \langle P';M'\rangle$ from \ref{trec-ic-iii-3} we can apply the inductive hypothesis of Case \ref{sound-case-3} to obtain $M' \not\equal \textsf{no}^S_P$.

            \item \ltsrule{mFls} where $M'_1 \xrightarrow{\tau} \textsf{no}^D_E$. However, $\textsf{no}^D_E \not{\rightarrow}$ and only \ltsrule{iProc} can be applied, which becomes identical to the previous case.
        \end{enumerate}

        \item For the derivation $S' \xrightarrow{\triangleleft \lab_m} S$ and $\lsem S' \rsem \xrightarrow{\triangleright\,\lab_m(v_m)} M$ s.t.~$\textbf{accepted}(\triangleright\,\lab_m(v_m), \overline{S'})$, from the base case of the case \textbf{iv}, $S'\equal \stsel{k}{K}$ and:
        \begin{gather}
          \label{trec-ic-iv-1}M \equal \condstmtexternal{m} \text{ for some } m \in K
        \end{gather}
        Due to the structures of $P$ and $M$, $\langle P;M \rangle \xrightarrow{\mu} \langle P_1;M_1\rangle$ can only be derived with the rules:
        \begin{enumerate}[label=\alph*.]
          \item \ltsrule{iProc} where $M\equal M_1$ and:
          \begin{gather}
            \label{trec-ic-iv-2}P \xrightarrow{\tau} P_1 \text{ by \ltsrule{pRec} where } P_1\equal Q[\nicefrac{\mu_X.Q}{X}]\\
            \label{trec-ic-iv-3}\langle P;M \rangle \xrightarrow{\mu} \langle P_1;M\rangle \xtwoheadrightarrow{t'} \langle P';M'\rangle
          \end{gather}
          Since we know \ref{trec-ic-iv-2} and $\textbf{accepted}(\tau,S)$ always holds, and from the assumption we know $\emptyset\cdot\emptyset\rightvdash P:S$ we can apply subject reduction (Proposition \ref{prop:sub-red}) to obtain:
          \begin{gather}
            \label{trec-ic-iv-4}\emptyset\cdot\emptyset\rightvdash P_1:S \text{ since } P \xrightarrow{\tau} P_1
          \end{gather}
          Since we know \ref{trec-ic-iv-4} and from the assumption we know $S' \xrightarrow{\triangleleft \lab_m} S$, $\lsem S' \rsem \xrightarrow{\triangleright\,\lab_m(v_m)} M$ s.t.~$\textbf{accepted}(\triangleright\,\lab_m(v_m), \overline{S'})$ and from \ref{trec-ic-iv-3} we know $\langle P_1;M\rangle \xtwoheadrightarrow{t'} \langle P';M'\rangle$, therefore we can apply the inductive hypothesis for Case \textbf{iv} to obtain that $M' \not\equal \textsf{no}^S_P$.

          \item\label{sound-trec-iv-b} \ltsrule{iMon}, where $P \equal P_1$ and with the rule:
          \begin{enumerate}[label=\arabic*.]
            \item \ltsrule{mTru} we know that:
            \begin{gather}
              \label{trec-ic-iv-5}\lsem S' \rsem \xrightarrow{\triangleright\,\lab_m(v_m)} M \xrightarrow{\tau} M_1 \text{ where } M_1\equal \envsndOP \lab_m(v_m).\lsem S_m \rsem[\nicefrac{v_m}{x_m}]\\
              \label{trec-ic-iv-6}\langle P;M \rangle \xrightarrow{\mu} \langle P;M_1\rangle \xtwoheadrightarrow{t'} \langle P';M'\rangle
            \end{gather}
            From the assumption we know that $\emptyset\cdot\emptyset\rightvdash P:S$, $S' \xrightarrow{\triangleleft \lab_m} S$ and $\textbf{accepted}(\triangleright\,\lab_m(v_m), \overline{S'})$ and we also know \ref{trec-ic-iv-5} and $\langle P;M_1\rangle \xtwoheadrightarrow{t'} \langle P';M'\rangle$ from \ref{trec-ic-iv-6}, therefore we can apply the inductive hypothesis for Case \textbf{v} to obtain that $M' \not\equal \textsf{no}^S_P$.

            \item \ltsrule{mFls} $M_1\equal \textsf{no}^D_P$. Since we know that $M_1 \not{\rightarrow}$, then $M' \equal M_1$ and hence $M'\not\equal \textsf{no}^S_P$
          \end{enumerate}
        \end{enumerate}
          
        \item\begin{subequations}For the derivation $S' \xrightarrow{\triangleleft \lab_m} S$ and $\lsem S' \rsem \xrightarrow{\triangleright\,\lab_m(v_m)} M'_1 \xrightarrow{\tau} M$ s.t.~$\textbf{accepted}(\triangleright\,\lab_m(v_m), \overline{S'})$, from the previous case we know that $M'_1 \equal \condstmtexternal{m}$ for some $m \in K$, and can transition with the rules:
        \begin{enumerate}[label=\alph*.]
          \item \ltsrule{mTru} where $M'_1 \xrightarrow{\tau} M$ s.t.~$M\equal\envsndOP \lab_m(v_m).\lsem S_m \rsem[\nicefrac{v_m}{x_m}]$. $\langle P;M\rangle \xrightarrow{\mu} \langle P_1;M_1\rangle$ can be derived with the rules:
          \begin{enumerate}[label=\arabic*.]
            \item \ltsrule{iProc} where $M\equal M_1$ and:
            \begin{gather}
              \label{trec-ic-v-1}P \xrightarrow{\tau} P_1 \text{ by \ltsrule{pRec} where }P_1\equal Q[\nicefrac{\mu_X.Q}{X}]\\
              \label{trec-ic-v-2}\langle P;M\rangle \xrightarrow{\tau} \langle P_1;M\rangle \xtwoheadrightarrow{t'} \langle P';M'\rangle
            \end{gather}
            Since we know \ref{trec-ic-v-1} and $\textbf{accepted}(\tau,S)$ always holds and from the assumption we know $\emptyset\cdot\emptyset\rightvdash P:S$ we can apply subject reduction (Proposition \ref{prop:sub-red}) to obtain that:
            \begin{gather}
              \label{trec-ic-v-3}\emptyset\cdot\emptyset\rightvdash P_1:S \text{ since } P \xrightarrow{\tau} P_1
            \end{gather}
            From the assumption we know that $S' \xrightarrow{\triangleleft \lab_m} S$ and $\lsem S' \rsem \xrightarrow{\triangleright\,\lab_m(v_m)} M'_1 \xrightarrow{\tau} M$  and $\textbf{accepted}(\triangleright\,\lab_m(v_m), \overline{S'})$ and we also know \ref{trec-ic-v-3} and $\langle P_1;M \rangle \xtwoheadrightarrow{t'} \langle P';M'\rangle$ from \ref{trec-ic-v-2}, therefore we can apply the inductive hypothesis for Case \textbf{v} to obtain that $M' \not\equal \textsf{no}^S_P$.

            \item\label{sound-trec-v-a-2} \ltsrule{iOut} where $P\equal P_1$ and:
            \begin{gather}
              \label{trec-ic-v-4}M \xrightarrow{\envsndOP \lab_m(v_m)} M_1 \text{ by \ltsrule{mOut} where }M_1\equal \lsem S_m \rsem[\nicefrac{v_m}{x_m}]\\
              \label{trec-ic-v-5}\langle P;M\rangle \xrightarrow{\envsndOP \lab_m(v_m)} \langle P;M_1\rangle \xtwoheadrightarrow{t'} \langle P';M'\rangle
            \end{gather}
            From the assumption we know $S' \xrightarrow{\triangleleft \lab_m} S$ and since from \ref{trec-ic-iv-1} we know $m \in K$, we can infer that $S \equal S_m$ by the rule \ltsrule{sSel} and therefore from \ref{trec-ic-v-4} we know that:
            \begin{gather}
              \label{trec-ic-v-6}M_1\equal \lsem S_m \rsem[\nicefrac{v_m}{x_m}]\equal\lsem S \rsem[\nicefrac{v_m}{x_m}]
            \end{gather}
            From the assumption we know $\emptyset\cdot\emptyset\rightvdash P:S$ to which we can apply the type variable instantiation lemma (Lemma \ref{lemma:type-substitution}) to obtain:
            \begin{gather}
              \label{trec-ic-v-7}\emptyset\cdot\emptyset\rightvdash P:S[\nicefrac{v_m}{x_m}]\text{ since }\emptyset[\nicefrac{v_m}{x_m}]\equal\emptyset
            \end{gather}
            Moreover from \ref{trec-ic-v-6} we know $M_1\equal\lsem S \rsem[\nicefrac{v_m}{x_m}]$ and by the synthesis substitution lemma (Lemma \ref{lemma:synthesis-subs}) we know that:
            \begin{gather}
              \label{trec-ic-v-8}M_1\equal\lsem S \rsem[\nicefrac{v_m}{x_m}]\equal\lsem S [\nicefrac{v_m}{x_m}]\rsem
            \end{gather}
            Therefore, since we know \ref{trec-ic-v-7} and \ref{trec-ic-v-8} and from \ref{trec-ic-v-5} we know $\langle P;M_1\rangle \xtwoheadrightarrow{t'} \langle P';M'\rangle$, we can apply the inductive hypothesis of Case \ref{sound-case-1} to obtain $M'\not\equal \textsf{no}^S_P$.
          \end{enumerate}

          \item\label{sound-trec-v-b} \ltsrule{mFls} where $M'_1 \xrightarrow{\tau} M$ where $M\equal\textsf{no}^D_P$. Since we know that $M\not{\rightarrow}$, then $M\equal M'$ and hence $M'\not\equal \textsf{no}^S_P$.
        \end{enumerate}\end{subequations}
      \end{enumerate}
    \end{subequations}

    \item[\ltsrule{tIf}]\begin{subequations} Therefore we know that $P\equal\textsf{if }A\textsf{ then }Q_1\textsf{ else }Q_2$ and $S \equal S'$, moreover:
      \begin{gather}
        \label{tIf-sound-premise-1}\Theta\cdot\Gamma\rightvdash A:\textsf{Bool}\\
        \label{tIf-sound-premise-2}\Theta\cdot\Gamma\rightvdash Q_1:S \text{ and } \Theta\cdot\Gamma\rightvdash Q_2:S
      \end{gather}
      This case is analogous to the case for the rule \ltsrule{tRec}, where the process type-checks with any session type $S$ and it can perform a $\tau$-transition either with the rule \ltsrule{pTru} or \ltsrule{pFls}. 
    \end{subequations}

    \item[\ltsrule{tPVar}]\begin{subequations} Therefore we know that $P \equal X$ and type-checks with any $S$. Moreover, from the rule we know that $\Theta(X)\equal S$. This case is analogous to the case for the rule \ltsrule{tRec}, however with less cases since $P \not{\xrightarrow{\mu}}$. \end{subequations}

    \item[\ltsrule{tNil}]\begin{subequations}Therefore we know that $P \equal \bm{0}$ and $S \equal \textsf{end}$ and $\lsem S \rsem \equal \bm{0}$. This case holds trivially since $P \not{\xrightarrow{\mu}}$. \end{subequations}
  \end{description}
\end{proof} 
\noindent Recall that originally we wanted to prove Theorem \ref{thrm:sound}. 
To prove this, we use Proposition \ref{prop:soundness}. 
In particular, since we are only interested in Case \ref{sound-case-1}, we can take its contrapositive to obtain Theorem \ref{thrm:sound} as a corollary. 
\section{Proof of \Cref{lem:partial-monitor-completeness}}
\label{sec:proof-partial-mon-completeness}

We follow an approach based on \emph{failing derivations}, inspired by \cite{BHLN12,BHLN17}.
\begin{enumerate}
\item We define the \emph{rule function $\Phi$} (\Cref{def:typing-rule-function})
  that maps a typing judgement
  of the form\; $J = \Theta\cdot\Gamma \rightvdash P:S$ \;to the set containing
  either all judgements in $J$'s premises, or $\true$;
\item we formalise a \emph{failing derivation} of a session typing judgement\;
  $J = \Theta\cdot\Gamma \rightvdash P:S$,
  \;(\Cref{def:failing-derivation}), showing that it exists
  if and only if $J$ is \emph{not} derivable (\Cref{lem:failing-deriv-not-deriv});
\item we formalise a negated typing judgement\;
  $\Theta\cdot\Gamma \nrightvdash P:S$
  \;(\Cref{def:negated-typing})
  and prove that it holds if and only if there is a corresponding 
  failing typing derivation of\; $\Theta\cdot\Gamma \rightvdash P:S$
  \;(\Cref{lem:typing-deriv-non-deriv});
\item finally, we use all ingredients above to prove
  \Cref{lem:partial-monitor-completeness}.
\end{enumerate}

\begin{definition}
  \label{def:typing-rule-function}%
  The \emph{rule function $\Phi$} is defined in \Cref{fig:typing-rule-function}.
\end{definition}

\begin{figure}
  \[
    \Phi\!\left(\Theta\cdot\Gamma \rightvdash P:S\right) \;=\;
    \left\{
      \begin{array}{@{}ll@{}}
        \big\{
          \Theta\cdot\Gamma, x_i:\textsf{B}_i \rightvdash P_i:S_i
        \big\}_{i \in I}
        &
        \text{if }
        \left\{
        \begin{array}{@{}l@{}}
          \exists I,J:\; P = \prcv{j}{I \cup J} \text{ and }\\
          S = \stbraNA{i}{I}
        \end{array}
        \right.
        \\
        \big\{
          \Theta\cdot\Gamma \rightvdash P':S_i
        \big\}
        &
        \text{if }
        \left\{
        \begin{array}{@{}l@{}}
          P = \psndk{l}{a}.P' \text{ and }\\
          S = \stselNA{i}{I} \text{ and }\\
          \exists i \in I:\;
          \texttt{l} \equal \texttt{l}_i \text{ and }
          \Gamma \rightvdash a:\textsf{B}_i
        \end{array}
        \right.
        \\
        \big\{
          \Theta\cdot\Gamma\rightvdash P:S \;,\; \Theta\cdot\Gamma\rightvdash Q:S
        \big\}
        &
        \text{if }
        \left\{
        \begin{array}{@{}l@{}}
          P = \pif{A}{P}{Q} \text{ and }\\
          \Gamma\rightvdash A:\textsf{Bool}
        \end{array}
        \right.
        \\
        \big\{
          \Theta, X:S\cdot\Gamma\rightvdash P':S
        \big\}
        &
        \text{if }
          P = \pmu{X}.P'
        \\
        \big\{
          \true
        \big\}
        &
        \text{if }
          P = \pnil \text{ and } S = \textsf{end}
        \\
        \big\{
          \true
        \big\}
        &
        \text{if }
          \Theta(X)\equal S
        \\
        \emptyset
        &
        \text{in all other cases}
      \end{array}
    \right.
  \]
  \caption[Typing Rule Function]{Rule Function for the Typing Judgement in \Cref{fig:session-typing-rules}}\label{fig:typing-rule-function}
\end{figure}

\begin{definition}[Failing typing derivation]
  \label{def:failing-derivation}
  A \emph{failing derivation of the typing judgement $\Theta\cdot\Gamma \rightvdash P:S$}
  is a finite sequence of judgements $(J_0, J_1, \ldots, J_n)$ such that:
  \begin{enumerate}
    \item for all $i \in 0..n$, $J_i$ is a judgement of the form\;
      $\Theta_i\cdot\Gamma_i \rightvdash P_i:S_i$;
    \item $J_0 = \Theta\cdot\Gamma \rightvdash P:S$;
    \item for all $i \in 1..n$, $J_i \in \Phi(J_{i-1})$;
    \item $\Phi(J_n) = \emptyset$.
  \end{enumerate}
\end{definition}
  
\begin{proposition}
  \label{lem:failing-deriv-not-deriv}
  \begin{enumerate}
  \item\label{lem:non-deriv-implies-failing}
    If\: $\Theta\cdot\Gamma \rightvdash P:S$ is \emph{not} derivable,
    then there is a failing derivation of $\Theta\cdot\Gamma \rightvdash P:S$;
  \item\label{lem:failing-implies-non-deriv}
    if there is a failing derivation of $\Theta\cdot\Gamma \rightvdash P:S$,  then $\Theta\cdot\Gamma \rightvdash P:S$ is \emph{not} derivable.
  \end{enumerate}
\end{proposition}
\begin{proof}
  Item~\ref{lem:non-deriv-implies-failing} is proved by considering
  a tentative derivation of $\Theta\cdot\Gamma \rightvdash P:S$,
  that in at least one branch does not reach an axiom
  (at least one such tentative derivations exists by hypothesis).
  Then, by induction on that branch of the tentative derivation,
  we construct a failing derivation,
  proceding by cases and applying the definition of the rule function $\Phi$
  (\Cref{fig:typing-rule-function}).
  \fxASwarning{Is it clear enough? Expand the proof?}%
  
  Item~\ref{lem:failing-implies-non-deriv} is proved by induction
  on the length $n$ of the failing derivation $(J_0, J_1, \ldots, J_n)$
  with\; $J_0 = \Theta\cdot\Gamma \rightvdash P:S$,
  \;proceeding in decreasing order toward $0$,
  and showing that each $J_i$ ($i \in 1..n$) is not derivable:
  \begin{itemize}
  \item for the base case $n$, we know that $\Phi(J_n) = \emptyset$,
    i.e., we are in the last clause of \Cref{fig:typing-rule-function}.
    By expanding all possible shapes of the judgement\;
    $J_n = \Theta_n\cdot\Gamma_n \rightvdash P_n:S_n$,
    \;we verify that no rule in \Cref{fig:session-typing-rules}
    can support it, hence we conclude that $J_n$ is not derivable;
    \fxASwarning{Is it clear enough? Expand the proof?}%
  \item for the inductive case $m < n$,
    we know that $\Phi(J_m) \neq \emptyset$,
    and there is a judgement\;
    $J_{m+1} = \Theta_{m+1}\cdot\Gamma_{m+1} \rightvdash P_{m+1}:S_{m+1} \in \Phi(J_m)$
    \;that is part of the failing derivation,
    and (if true) would be a premise to derive $J_m$.
    However, by the induction hypothesis, we have that $J_{m+1}$ is not derivable.
    Now, observe that the rules in \Cref{fig:session-typing-rules} are deterministic:
    i.e., there is no alternative rule applicable to derive $J_m$,
    besides the one having the (non-derivable) premise $J_{m+1}$.
    Hence, we conclude that $J_m$ is not derivable.
    \fxASwarning{Is it clear enough? Expand the proof?}%
  \end{itemize}
\end{proof}

\begin{definition}[Negated typing judgement]
  \label{def:negated-typing}
  The judgement\; $\Theta\cdot\Gamma \nrightvdash P:S$
  \;is inductively defined by the rules in \Cref{fig:negated-session-typing-rules}
  --- where, for all boolean predicates $A$,\;
  $\Gamma \nrightvdash A : \textsf{Bool}$ \;holds if and only if\;
  $\Gamma \rightvdash A : \textsf{Bool}$ \;is \emph{not} derivable.
\end{definition}

\begin{figure}
  \small
  \textbf{Negated Identifier Typing}
  \bigskip\\
  \centerline{
    \inference[\text{[\textsc{nVar}]}]{\Gamma(x)\equal \basetypek{1} \qquad \basetypek{1} \neq \basetypek{2}}{\Gamma \nrightvdash x : \basetypek{2}}
    \hspace{1cm}
    \inference[\text{[\textsc{nVal}]}]{v \not\in \basetype}{\Gamma \nrightvdash v : \basetype}
  }\bigskip\\
  \textbf{Negated Process Typing}
  \medskip\\
  \centerline{
    \inference[\text{[\textsc{nBra0}]}]{\forall I, \{\texttt{l}_i, \textsf{B}_i, S_i\}_{i \in I} \quad S \neq \stbraNA{i}{I}}{\Theta\cdot\Gamma \nrightvdash \prcv{j}{J}: S}
  }\medskip\\
  \centerline{
    \inference[\text{[\textsc{nBra1}]}]{\exists k \in I \cdot \forall j \in J \qquad \texttt{l}_k \not\equal \texttt{l}_j}{\Theta\cdot\Gamma \nrightvdash \prcv{j}{J}: \stbraNA{i}{I}}
  }\medskip\\
  \centerline{
    \inference[\text{[\textsc{nBra2}]}]{\exists i \in I \quad \Theta\cdot\Gamma, x_i:\textsf{B}_i \nrightvdash P_i:S_i}{\Theta\cdot\Gamma \nrightvdash \prcv{i}{I \cup J}: \stbraNA{i}{I}}
  }\medskip\\
  \medskip\centerline{
    \inference[\text{[\textsc{nSel0}]}]{\forall I, \{\texttt{l}_i, \textsf{B}_i, S_i\}_{i \in I} \quad S \neq \stselNA{i}{I}}{\Theta\cdot\Gamma \nrightvdash \psndk{l}{a}.P: S}
  }\medskip\\
  \medskip\centerline{
    \inference[\text{[\textsc{nSel1}]}]{\forall i \in I \qquad \texttt{l} \not\equal \texttt{l}_i}{\Theta\cdot\Gamma \nrightvdash \psndk{l}{a}.P: \stselNA{i}{I}}
  }\medskip\\
  \medskip\centerline{
    \inference[\text{[\textsc{nSel2}]}]{\exists k \in I \qquad \texttt{l} \equal \texttt{l}_k \qquad \Gamma\nrightvdash a:\textsf{B}_k}{\Theta\cdot\Gamma \nrightvdash \psndk{l}{a}.P: \stselNA{i}{I}}
  }\medskip\\
  \medskip\centerline{
    \inference[\text{[\textsc{nSel3}]}]{\exists i \in I \qquad \texttt{l} \equal \texttt{l}_i \qquad \Gamma\rightvdash a:\textsf{B}_i \qquad \Theta\cdot\Gamma \nrightvdash P:S_i}{\Theta\cdot\Gamma \nrightvdash \psndk{l}{a}.P: \stselNA{i}{I}}
  }\medskip\\
  \centerline{
    \inference[\text{[\textsc{nRec}]}]{\Theta, X:S\cdot\Gamma\nrightvdash P:S}{\Theta\cdot\Gamma \nrightvdash \pmu{X}.P:S}
    \hspace{1cm}
    \inference[\text{[\textsc{nPVar}]}]{\Theta(X) \neq S}{\Theta\cdot\Gamma \nrightvdash X :S}
  }\bigskip\\
  \centerline{
    \inference[\text{[\textsc{nIf1}]}]{\Gamma\nrightvdash A:\textsf{Bool}}{\Theta\cdot\Gamma \nrightvdash \pif{A}{P}{Q}: S}
  }\bigskip\\
  \centerline{
    \inference[\text{[\textsc{nIf2}]}]{\Gamma\rightvdash A:\textsf{Bool}&\Theta\cdot\Gamma\nrightvdash P:S & \Theta\cdot\Gamma\rightvdash Q:S}{\Theta\cdot\Gamma \nrightvdash \pif{A}{P}{Q}: S}
    \qquad
    \inference[\text{[\textsc{nIf3}]}]{\Gamma\rightvdash A:\textsf{Bool}&\Theta\cdot\Gamma\rightvdash P:S & \Theta\cdot\Gamma\nrightvdash Q:S}{\Theta\cdot\Gamma \nrightvdash \pif{A}{P}{Q}: S}
  }\bigskip\\
  \centerline{
    \inference[\text{[\textsc{nIf4}]}]{\Gamma\rightvdash A:\textsf{Bool}&\Theta\cdot\Gamma\nrightvdash P:S & \Theta\cdot\Gamma\nrightvdash Q:S}{\Theta\cdot\Gamma \nrightvdash \pif{A}{P}{Q}: S}
  }\bigskip\\
  \centerline{
    \inference[\text{[\textsc{nNil}]}]{S \neq \textsf{end}}{\Theta\cdot\Gamma \nrightvdash \pnil: S}
  }
  \caption[Negated Session Typing Rules]{Negated Session Typing Rules}\label{fig:negated-session-typing-rules}
\end{figure}

\begin{lemma}
  \label{lem:typing-deriv-non-deriv}%
  For all $\Theta, \Gamma, P, S$,\;
  there is a failing derivation of $\Theta\cdot\Gamma \rightvdash P:S$
  \;if and only if\;
  $\Theta\cdot\Gamma \nrightvdash P:S$ is derivable.
\end{lemma}
\begin{proof}
  ($\implies$)
  Assume a failing derivation $(J_0, J_1, \ldots, J_n)$
  where $J_0 = \Theta\cdot\Gamma \rightvdash P:S$.
  We now prove that, for all $i \in 1..n$,
  the judgement\; $\Theta_i\cdot\Gamma_i \nrightvdash P_i:S_i$ is derivable
  by the rules in \Cref{fig:negated-session-typing-rules}.
  We proceed by induction on $n$, in decreasing order toward $0$.
  \begin{itemize}
  \item Base case $n$.\quad
    Since $\Phi(J_m) = \emptyset$,
    by  have the following sub-cases for $P_m$.
    \begin{itemize}
    \item $P_n$ is an input process $\prcv{j}{J}$.\quad
      Since $\Phi(J_n) = \emptyset$, 
      by \Cref{fig:typing-rule-function} we have the following
      (non mutually exclusive) possibilities:
      \begin{itemize}
      \item $S_n$ is \emph{not} an external choice.
        Then, we obtain\; $\Theta_n\cdot\Gamma_i \nrightvdash P_n:S_n$ \;by rule [\textsc{nBra0}];
      \item $S_n = \stbraNA{i}{I}$ and
        $\forall i \in I$, $\not\exists j \in J$ such that $\texttt{l}_j \equal \texttt{l}_i$.
        Then, we obtain\; $\Theta_n\cdot\Gamma_i \nrightvdash P_n:S_n$ \;by rule [\textsc{nBra1}];
      \end{itemize}
    \item $P_n$ is an output process $\psndk{l}{a}.P'$.\quad
      Since $\Phi(J_n) = \emptyset$, 
      by \Cref{fig:typing-rule-function} we have the following
      (non mutually exclusive) possibilities:
      \begin{itemize}
      \item $S_n$ is \emph{not} an internal choice.
        Then, we obtain\; $\Theta_n\cdot\Gamma_n \nrightvdash P_n:S_n$ \;by rule [\textsc{nSel0}];
      \item $S_n = \stselNA{i}{I}$ and
        $\not\exists i \in I$ such that $\texttt{l}_i \equal \texttt{l}$.
        Then, we obtain\; $\Theta_n\cdot\Gamma_n \nrightvdash P_n:S_n$ \;by rule [\textsc{nSel1}];
      \item $S_n = \stselNA{i}{I}$ and
        $\exists i \in I$ such that $\texttt{l}_i \equal \texttt{l}$,
        but $\Gamma\nrightvdash a:\textsf{B}_i$.
        Then, we obtain\; $\Theta_n\cdot\Gamma_n \nrightvdash P_n:S_n$ \;by rule [\textsc{nSel2}];
      \end{itemize}
    \item $P_n = \pif{A}{P'}{Q'}$.\quad
      Since $\Phi(J_n) = \emptyset$, 
      by \Cref{fig:typing-rule-function} we have that
      $\Gamma \rightvdash A:\textsf{Bool}$ is not derivable:
      this implies $\Gamma \nrightvdash A:\textsf{Bool}$
      Thus, we obtain\; $\Theta_n\cdot\Gamma_n \nrightvdash P_n:S_n$ \;by rule [\textsc{nIf1}];
    \item $P_n = X$.\quad
      Since $\Phi(J_n) = \emptyset$, 
      by \Cref{fig:typing-rule-function} we know that $\Theta(X) \not\equal S$.
      Thus, we obtain\; $\Theta_n\cdot\Gamma_n \nrightvdash P_n:S_n$ \;by rule [\textsc{nPVar}];
    \item $P_n = \pnil$.\quad
      Since $\Phi(J_n) = \emptyset$, 
      by \Cref{fig:typing-rule-function} we know that $S \not\equal \textsf{end}$.
      Thus, we obtain\; $\Theta_n\cdot\Gamma_n \nrightvdash P_n:S_n$ \;by rule [\textsc{nNil}].
    \end{itemize}
  \item Inductive case $m < n$.\quad
    This implies $\Phi(J_m) \neq \emptyset$, and
    we have the following sub-cases for $P_m$.
    \begin{itemize}
    \item $P_m$ is an input process $\prcv{j}{J}$.\quad
      Since $\Phi(J_m) \neq \emptyset$,
      by \Cref{fig:typing-rule-function}
      we know that $S_n = \stbraNA{i}{I}$ with $J \subseteq I$, and
      $\exists i \in I$ such that
      $J_{m+1} = \Theta_m\cdot\Gamma_m, x_i:\textsf{B}_i \rightvdash P_i:S_i$.
      Therefore, by the induction hypothesis, we have
      $\Theta_m\cdot\Gamma_m, x_i:\textsf{B}_i \nrightvdash P_i:S_i$.
      Thus, we obtain\; $\Theta_m\cdot\Gamma_m \nrightvdash P_m:S_m$ \;by rule [\textsc{nBra2}];
    \item $P_m$ is an output process $\psndk{l}{a}.P'$.\quad
      Since $\Phi(J_m) \neq \emptyset$,
      by \Cref{fig:typing-rule-function}
      we know that $S_n = \stselNA{i}{I}$, and
      $\exists i \in I$ such that $\texttt{l}_i \equal \texttt{l}$,
      and $\Gamma \rightvdash a:\textsf{B}_i$,
      and $J_{m+1} = \Theta_m\cdot\Gamma_m \rightvdash P':S_i$.
      Therefore, by the induction hypothesis, we have
      $\Theta_m\cdot\Gamma_m \nrightvdash P':S_i$.
      Thus, we obtain\; $\Theta_m\cdot\Gamma_m \nrightvdash P_m:S_m$ \;by rule [\textsc{nSel3}];
    \item $P_m = \pif{A}{P'}{Q'}$.\quad
      Since $\Phi(J_m) \neq \emptyset$,
      by \Cref{fig:typing-rule-function}
      we know that $J_{m+1} = \Theta_m\cdot\Gamma_m \rightvdash R:S_m$,
      where $R$ is either $P'$ or $Q'$:
      \begin{itemize}
      \item if $R = P'$, 
        by the induction hypothesis we have
        $\Theta_m\cdot\Gamma_m \nrightvdash P':S_m$.
        Thus, we obtain\; $\Theta_m\cdot\Gamma_m \nrightvdash P_m:S_m$ \;by rule [\textsc{nIf2}] or [\textsc{nIf4}];
      \item if $R = Q'$, 
        the proof is analogous, except that we conclude by rule [\textsc{nIf3}] or [\textsc{nIf4}];
      \end{itemize}
    \item $P_m = \pmu{X}.P'$.\quad
      Since $\Phi(J_m) \neq \emptyset$,
      by \Cref{fig:typing-rule-function}
      we know that $J_{m+1} = \Theta_m,X:S:\cdot\Gamma_m \rightvdash P':S_m$;
      hence, by the induction hypothesis we have
      $\Theta_m,X:S_m\cdot\Gamma_m \nrightvdash P':S_m$.
      Thus, we obtain $\Theta_m\cdot\Gamma_m \nrightvdash P_m:S_m$ \;by rule [\textsc{nRec}].
    \end{itemize}
  \end{itemize}

  ($\impliedby$)
  Assume that\; $\Theta\cdot\Gamma \nrightvdash P:S$ \;is derivable.
  We prove the statement by induction on the derivation of
  \; $\Theta\cdot\Gamma \nrightvdash P:S$,
  \;by constructing a failing derivation $(J_0, J_1, \ldots, J_n)$
  where\; $J_0 = \Theta\cdot\Gamma \rightvdash P:S$.
  \begin{itemize}
  \item Base cases [\textsc{nBra0}], [\textsc{nBra1}], [\textsc{nSel0}], [\textsc{nSel1}], [\textsc{nSel2}], [\textsc{nPVar}], [\textsc{nIf1}], [\textsc{nNil}].\quad
    In all these cases, by the last clause in \Cref{fig:typing-rule-function},
    we have\; $\Phi\!\left(\Theta\cdot\Gamma \rightvdash P:S\right) = \emptyset$,
    \;and thus, there is a failing derivation of\; $\Theta\cdot\Gamma \rightvdash P:S$ \;with just one element.
  \item Inductive cases [\textsc{nBra2}], [\textsc{nSel3}], [\textsc{nRec}], [\textsc{nIf2}], [\textsc{nIf3}], [\textsc{nIf4}].\quad
    Let $R$ be the derivation rule under consideration:
    in all these cases for $R$, one of the first 4 cases of \Cref{fig:typing-rule-function} applies, hence we have $\Phi(J_0) = \Phi\!\left(\Theta\cdot\Gamma \rightvdash P:S\right) =
    \left\{\Theta_i\cdot\Gamma_i \rightvdash P_i:S_i\right\}_{i \in I}$
    (for some $I \neq \emptyset$).
    Observe that at least one of the judgements contained in $\Phi(J_0)$
    corresponds to a premise of\; $\Theta\cdot\Gamma \nrightvdash P:S$
    \;by rule $R$: let such a judgement be\;
    $\Theta_k\cdot\Gamma_k \rightvdash P_k:S_k$
    \;(for some $k \in I$), and the corresponding premise for $R$ be\;
    $\Theta_k\cdot\Gamma_k \nrightvdash P_k:S_k$.
    \;Therefore, by the induction hypothesis,
    there is a failing derivation $(J'_0, J'_1, \ldots, J'_m)$
    where\; 
    $J'_0 = \Theta_k\cdot\Gamma_k \rightvdash P_k:S_k$.
    \;Hence, we can construct a failing derivation for\;
    $J_0 = \Theta\cdot\Gamma \rightvdash P:S$ \;as
    $(J_0, J'_0, J'_1, \ldots, J'_m)$.    
  \end{itemize}
\end{proof}

\begin{proposition}
  \label{lem:neg-typing-unfolding}%
  Assume that $\Theta\cdot\Gamma \nrightvdash \mu_{X}.P : S$
  \;is derived with one or more instances of rule [\textsc{nPVar}] for variable $X$.
  Then, there is a process $\mu_{X}.P'$ where $P'$ is an unfolding of $P$,
  \fxASwarning{To be pedantic, this is not an strictly speaking an unfolding}
  such that
  \(
    \Theta\cdot\Gamma \nrightvdash
    \mu_{X}.P':S
  \)
  is derived by \emph{not} using rule [\textsc{nPVar}] on $X$.
\end{proposition}
\begin{proof}
  Consider the derivation of $\Theta\cdot\Gamma \nrightvdash \mu_{X}.P:S$,
  and assume it contains a judgement\; $\Theta''\cdot\Gamma'' \nrightvdash X:S''$
  \;(instance of rule $[\textsc{nPVar}]$).
  If we replace $X$ with $P$ (thus effectively unfolding $P$),
  we obtain a derivation ending on a tentative judgement\;  $\Theta''\cdot\Gamma'' \nrightvdash P:S''$,
  with three possibilities:
  \begin{itemize}
    \item the judgement\; $\Theta''\cdot\Gamma'' \nrightvdash P:S''$ \;is not derivable.
      This is impossible, because it would imply %
      that $\Theta''\cdot\Gamma'' \rightvdash P:S''$ is derivable,
      which would imply that\; $\Theta\cdot\Gamma \rightvdash \mu_{X}.P:S$ \;is derivable,
      hence we would obtain that the hypothesis\; $\Theta\cdot\Gamma \nrightvdash \mu_{X}.P:S$ \;is false (contradiction);
    \item the judgement\; $\Theta''\cdot\Gamma'' \nrightvdash P:S''$ is derived by reaching an axiom different from $[\textsc{nPVar}]$.
      In this case, we graft this derivation in the original derivation, and we are done;
    \item the judgement\; $\Theta''\cdot\Gamma'' \nrightvdash P':S''$ is derived by reaching an instance of $[\textsc{nPVar}]$ for variable $X$.
      In this case, we graft this derivation in the original derivation, and we repeat the substitution above again, by replacing $X$ with $P$.
  \end{itemize}
  After a finite number of substitutions of $X$ with $P$,
  and repetitions of the above procedure,
  we obtain a derivation for the judgement
  $\Theta\cdot\Gamma \nrightvdash \mu_{X}.P':S$
  where $P'$ is an unfolding of $P$,
  \fxASwarning{This is not quite an unfolding --- we are duplicating the body of $P$ inside the recursion}
  whithout any instance of rule [\textsc{nPVar}] applied to $X$.%
  \fxASwarning{More details would not hurt!}
\end{proof}

\lemMonitorCompleteness*
\begin{proof}
  By the hypothesis and \Cref{lem:failing-deriv-not-deriv},
  we know that there is a failing derivation
  for\; $\emptyset\cdot\emptyset \rightvdash P:S$.
  \;Moreover, by \Cref{lem:typing-deriv-non-deriv}, we know that from such a failing derivation
  we can construct a derivation of\; $\emptyset\cdot\emptyset \nrightvdash P:S$.

  Now, let us consider the process $P_\mu$ obtained
  by unfolding $P$ until\; $\emptyset\cdot\emptyset \nrightvdash P_\mu:S$
  \;is derived \emph{without} using the axiom [\textsc{nPVar}]
  in \Cref{fig:negated-session-typing-rules}:
  such $P_\mu$ exists, and is obtained by repeating \Cref{lem:neg-typing-unfolding}
  until all occurrences of [\textsc{nPVar}] in the derivation of $P$ are removed.
  Observe that $P$ and $P_\mu$ are weakly bisimilar:
  the only behavioural difference is that $P_\mu$ may perform less
  $\tau$-transitions to unfold recursions (by rule [\textsc{pRec}]).

  Since\; $\emptyset\cdot\emptyset \nrightvdash P_\mu:S$ \;is derivable,
  by \Cref{lem:typing-deriv-non-deriv} there is a failing derivation
  $(J_0, J_1, \ldots, J_{n})$,
  with\; $J_0 = \emptyset\cdot\emptyset \rightvdash P_\mu:S$.
  Furthermore, by observing the proof of \Cref{lem:typing-deriv-non-deriv},
  we can see that:
  \begin{enumerate}[label={(o\arabic*)}]
  \item\label{item:part-mon-comp:ind} $\forall i \in 1..n-1$,
    $J_i$ corresponds to the conclusion of an inductive rule $R_i$ in \Cref{fig:negated-session-typing-rules},
    and one of the judgements $\Phi(J_i)$ corresponds to a judgement in $R_i$'s premises;
  \item\label{item:part-mon-comp:axiom} $J_n$ is the conclusion of an axiom $R_n$ in \Cref{fig:negated-session-typing-rules}.
  \end{enumerate}

  We now use the failing derivation $(J_0, J_1, \ldots, J_{n})$
  to construct a trace $t = t_1 t_2 \cdots t_n e'$
  leading from $\instr{P_\mu}{M}$ to some $\instr{P'}{M'}$, such that:
  \begin{enumerate}[label=(\alph*)]
    \item\label{item:partial-mon:trace}
      $\instr{P_\mu}{M} = \instr{P_0\sigma_0}{\lsem{S_0}\rsem} \xRightarrow{t_1} \instr{P_1\sigma_1}{\lsem{S_1}\rsem} \xRightarrow{t_2} \instr{P_2\sigma_2}{\lsem{S_2}\rsem} \xRightarrow{t_3}\cdots \xRightarrow{t_n} \instr{P_n\sigma_n}{\lsem{S_n}\rsem} \xRightarrow{t'} \instr{P'}{M'}$,
      where, $\forall i \in 0..n$:
      \begin{itemize}
        \item $P_i$ and $S_i$ are the process and type occurring in judgement $J_i$;
        \item there is a non-empty set of substitutions $\Sigma_i$ such that $\forall \sigma \in \Sigma_i: \fv{P_i} \subseteq \dom{\sigma}$; and
        \item $\sigma_i$ is any substitution such that $\sigma_i \in \Sigma_i$;
          \fxASwarning{This quantification is confusing, it may be better to keep a non-empty set of distinct processes $\{P_i\sigma\}_{\sigma \in \Sigma_i}$}
        \item $\fv{P_i} \subseteq \dom{\sigma_i}$; and
        \item $\forall j \in j..n-1$,
          $\dom{\sigma_j} \subseteq \dom{\sigma_{j+1}}$.
      \end{itemize}
      Notice that $P_0 = P_\mu$,
      and $P_0\sigma_0 = P_0 = P_\mu$
      (for any $\sigma_0$, because $\fv{P_\mu} = \fv{P_0} = \emptyset$
      by the hypothesis ``$P$ closed''),
      and $S_0 = S$;
    \item\label{item:partial-mon:stuck}
      $\instr{P'}{M'} \mathrel{\not\rightarrow}{}$ with $P' \not\equal \pnil$ or $M' \not\equal \pnil$;
      moreover, $M' \not\equal \noes$ and $M' \not\equal \noed$.
  \end{enumerate}
  We proceed by induction on $n$, in decreasing order toward $0$:
  \fxASwarning{Here we are skimming over the fact that
    types may be under recursion,
    hence the synthesised monitors may perform a finite number of $\tau$-transitions
    before reaching the configurations described below%
  }
  \begin{itemize}
  \item base case $n$.\quad
    We know that $J_n$ falls under item~\ref{item:part-mon-comp:axiom} above,
    so we proceed by cases on the axiom $R_n$
    with conclusion\; $\Theta_n\cdot\Gamma_n \nrightvdash P_n: S_n$:
    \begin{itemize}
    \item{}[\textsc{nBra0}]\quad
      We have $P_n = \prcv{j}{J}$, and $S_n$ is not a branching type.
      We have the following sub-cases:
      \begin{itemize}
      \item $S_n = \textsf{end}$,
        and thus (by \Cref{def:synth-function})
        $\lsem{S_n}\rsem = \pnil$.
        In this case, we conclude by taking $P' = P_n\sigma_n$,
        $\Sigma_n$ as the set of all substitutions with domain $\fv{P'}$,
        $\sigma_n$ as any substitution from $\Sigma_n$ (hence, $\fv{P_n} \subseteq \dom{\sigma_n}$),
        $t'$ empty,
        and $M' = \lsem{S_n}\rsem$;
      \item $S_n = \stselNA{i}{I}$
        and thus (by \Cref{def:synth-function})\\
        $\lsem{S_n}\rsem = \triangleright\big\{\texttt{l}_i(x_i:\textsf{B}_i).\pif{\isValueB{i}{x_i}}{\envsndOP \texttt{l}_i(x_i).\lsem S_i\rsem}{\nopd}\big\}_{i\in I}$.
        In this case, we conclude by taking $P' = P_n\sigma_n$,
        $\Sigma_n$ as the set of all substitutions with domain $\fv{P'}$,
        $\sigma_n$ as any substitution from $\Sigma_n$ (hence, $\fv{P_n} \subseteq \dom{\sigma_n}$),
        $t'$ empty,
        and $M' = \lsem{S_n}\rsem$;
      \end{itemize}
    \item{}[\textsc{nSel0}]\quad
      We have $P_n = \psndk{l}{a}$, and $S_n$ is not a selection type.
      We have the following sub-cases:
      \begin{itemize}
      \item $S_n = \textsf{end}$,
        and thus (by \Cref{def:synth-function})
        $\lsem{S_n}\rsem = \pnil$.
        In this case, we conclude by taking $P' = P_n\sigma_n$,
        $\Sigma_n$ as the set of all substitutions with domain $\fv{P'}$,
        $\sigma_n$ as any substitution from $\Sigma_n$ (hence, $\fv{P_n} \subseteq \dom{\sigma_n}$),
        $t'$ empty,
        and $M' = \lsem{S_n}\rsem$;
      \item $S_n = \stbraNA{i}{I}$
        and thus (by \Cref{def:synth-function})\\
        $\lsem{S_n}\rsem = \envrcvOP\big\{\texttt{l}_i(x_i:\textsf{B}_i).\pif{\isValueB{i}{x_i}}{\triangleleft \texttt{l}_i(x_i).\lsem S_i\rsem}{\noed}\big\}_{i\in I}$.
        In this case, we conclude by taking $P' = P_n\sigma_n$,
        $\Sigma_n$ as the set of all substitutions with domain $\fv{P'}$,
        $\sigma_n$ as any substitution from $\Sigma_n$ (hence, $\fv{P_n} \subseteq \dom{\sigma_n}$),
        $t' = \envrcvOP\texttt{l}_k(v)$ for any $k \in I$ and
        $v$ such that $\isValueB{k}{v}$ holds,
        and $M' = \triangleleft \texttt{l}_k(v).\lsem S_k\rsem$;
      \end{itemize}
    \item{}[\textsc{nIf1}]\quad
      We have $P_n = \pif{A}{Q}{Q'}$, and (from the rule premises)
      $\Gamma_n\nrightvdash A:\textsf{Bool}$: thus, $A$ cannot evaluate to a boolean
      (by item~\ref{assumption:ill-typed-no-eval} on page~\pageref{assumption:ill-typed-no-eval}),
      hence $P_n$ cannot reduce by the rules in \Cref{fig:process-calc-syntax-semantics}.
      Hence, in this case we conclude by taking $P' = P_n\sigma_n$,
      $\Sigma_n$ as the set of all substitutions with domain $\fv{P'}$,
      $\sigma_n$ as any substitution from $\Sigma_n$ (hence, $\fv{P_n} \subseteq \dom{\sigma_n}$),
      $t'$ empty,
      and $M' = \lsem{S_n}\rsem$;
    \item{}[\textsc{nNil}]\quad
      We have $P_n = \pnil$, and $S_n \not\equal \textsf{end}$.
      We have the following sub-cases:
      \begin{itemize}
      \item $S_n = \stbraNA{i}{I}$
        and thus (by \Cref{def:synth-function})\\
        $\lsem{S_n}\rsem = \envrcvOP\big\{\texttt{l}_i(x_i:\textsf{B}_i).\pif{\isValueB{i}{x_i}}{\triangleleft \texttt{l}_i(x_i).\lsem S_i\rsem}{\noed}\big\}_{i\in I}$.
        In this case, we conclude by taking $P' = P_n\sigma_n$,
        $\Sigma_n$ as the set of all substitutions with domain $\fv{P'}$,
        $\sigma_n$ as any substitution from $\Sigma_n$ (hence, $\fv{P_n} \subseteq \dom{\sigma_n}$),
        $t' = \envrcvOP\texttt{l}_k(v)$ for any $k \in I$ and
        $v$ such that $\isValueB{k}{v}$ holds,
        and $M' = \triangleleft \texttt{l}_k(v).\lsem S_k\rsem$;
      \item $S_n = \stselNA{i}{I}$
        and thus (by \Cref{def:synth-function})\\
        $\lsem{S_n}\rsem = \triangleright\big\{\texttt{l}_i(x_i:\textsf{B}_i).\pif{\isValueB{i}{x_i}}{\envsndOP \texttt{l}_i(x_i).\lsem S_i\rsem}{\nopd}\big\}_{i\in I}$.
        In this case, we conclude by taking $P' = P_n\sigma_n$,
        $\Sigma_n$ as the set of all substitutions with domain $\fv{P'}$,
        $\sigma_n$ as any substitution from $\Sigma_n$ (hence, $\fv{P_n} \subseteq \dom{\sigma_n}$),
        $t'$ empty,
        and $M' = \lsem{S_n}\rsem$;
      \end{itemize}
    \item{}[\textsc{nBra1}]\quad
      We have $P_n = \prcv{j}{J}$, and $S_n = \stbraNA{i}{I}$,
      and $\exists k \in I: \forall j \in J :\texttt{l}_k \not\equal \texttt{l}_j$.
      By \Cref{def:synth-function}, we have\\
      $\lsem{S_n}\rsem = \triangleright\big\{\texttt{l}_i(x_i:\textsf{B}_i).\pif{\isValueB{i}{x_i}}{\envsndOP \texttt{l}_i(x_i).\lsem S_i\rsem}{\nopd}\big\}_{i\in I}$.
      In this case, we conclude by taking $P' = P_n\sigma_n$,
      $\Sigma_n$ as the set of all substitutions with domain $\fv{P'}$,
      $\sigma_n$ as any substitution from $\Sigma_n$ (hence, $\fv{P_n} \subseteq \dom{\sigma_n}$),
      $t' = \envrcvOP\texttt{l}_k(v)$ for any
      $v$ such that $\isValueB{k}{v}$ holds,
      and $M' = \triangleleft \texttt{l}_k(v).\lsem S_k\rsem$;
    \item{}[\textsc{nSel1}]\quad
      We have $P_n = \psndk{l}{a}.P''$, and $S_n = \stselNA{i}{I}$,
      and $\forall i \in I: \texttt{l} \not\equal \texttt{l}_i$.
      By \Cref{def:synth-function}, we have\\
      $\lsem{S_n}\rsem = \triangleright\big\{\texttt{l}_i(x_i:\textsf{B}_i).\pif{\isValueB{i}{x_i}}{\envsndOP \texttt{l}_i(x_i).\lsem S_i\rsem}{\nopd}\big\}_{i\in I}$.
      In this case, we have two possibilities:
      \begin{itemize}
      \item there exists a non-empty set of substitutions $\Sigma''$ such that,
        $\forall \sigma'' \in \Sigma''$,
        $a\sigma'' \Downarrow v$ (for some $v$).
        In this case, we conclude by taking
        $\Sigma_n = \Sigma''$,
        $\sigma_n$ as any $\sigma'' \in \Sigma_n$,
        $P' = P''\sigma_n$,
        $t' = \tau$ (for the communication between $P_n$ and $\lsem{S_n}\rsem$),
        and $M' = \nops$
        (by rule [\textsc{mIV}] in \Cref{fig:monitor-calculus});
      \item there is no substitution $\sigma''$ such that 
        $a\sigma'' \Downarrow v$ (for any $v$),
        i.e., no substitution $\sigma''$ allowing $P_n\sigma''$ to perform
        an output transition, by the rules in
        \Cref{fig:process-calc-syntax-semantics}.
        In this case, we conclude by taking $P' = P_n\sigma_n$,
        $\Sigma_n$ as the set of all substitutions with domain $\fv{P'}$,
        $\sigma_n$ as any substitution from $\Sigma_n$ (hence, $\fv{P_n} \subseteq \dom{\sigma_n}$),
        $t'$ empty,
        and $M' = \lsem{S_n}\rsem$;
      \end{itemize}
    \item{}[\textsc{nSel2}]\quad
      We have $P_n = \psndk{l}{a}.P''$, and $S_n = \stselNA{i}{I}$,
      and $\exists k \in I: \texttt{l} \equal \texttt{l}_k$ and
      $\Gamma\nrightvdash a:\textsf{B}_k$.
      By \Cref{def:synth-function}, we have\\
      $\lsem{S_n}\rsem = \triangleright\big\{\texttt{l}_i(x_i:\textsf{B}_i).\pif{\isValueB{i}{x_i}}{\envsndOP \texttt{l}_i(x_i).\lsem S_i\rsem}{\nopd}\big\}_{i\in I}$.
      In this case, we have two possibilities:
      \begin{itemize}
      \item there exists a non-empty set of substitutions $\Sigma''$ such that,
        $\forall \sigma'' \in \Sigma''$,
        $a\sigma'' \Downarrow v$ (for some $v$).
        In this case, we conclude by taking
        $\Sigma_n = \Sigma''$,
        $\sigma_n$ as any $\sigma'' \in \Sigma_n$,
        $P' = P''\sigma_n$,
        $t' = \tau\tau$
        (the first $\tau$ is for the communication between $P_n$ and $\lsem{S_n}\rsem$, while the second $\tau$ is for evaluating the ``if'' in the monitor, which picks the $\false$ branch),
        and $M' = \nopd$
      \item there is no substitution $\sigma''$ such that 
        $a\sigma'' \Downarrow v$ (for any $v$),
        i.e., no substitution $\sigma''$ allowing $P_n\sigma''$ to perform
        an output transition, by the rules in
        \Cref{fig:process-calc-syntax-semantics}.
        In this case, we conclude by taking $P' = P_n\sigma_n$,
        $\Sigma_n$ as the set of all substitutions with domain $\fv{P'}$,
        $\sigma_n$ as any substitution from $\Sigma_n$ (hence, $\fv{P_n} \subseteq \dom{\sigma_n}$),
        $t'$ empty,
        and $M' = \lsem{S_n}\rsem$;
      \end{itemize}
    \item{}[\textsc{nPVar}].\quad
      This case is impossible: $P\mu$ does not contain instances of
      [\textsc{nPVar}], as we unfolded $P$ into $P_\mu$
      by using \Cref{lem:neg-typing-unfolding} (see above);
    \end{itemize}
  \item inductive case $m < n$.\quad
    We know that $J_m$ falls under item~\ref{item:part-mon-comp:ind} above,
    so we proceed by cases on the inductive rule $R_m$
    with conclusion\; $\Theta_m\cdot\Gamma_m \nrightvdash P_m: S_m$:
    \begin{itemize}
    \item{}[\textsc{nBra2}].\quad
      We have $S_m = \stbraNA{i}{I}$ and $P_m = \prcv{j}{I \cup J}$, and
      $\exists k \in I: \Theta\cdot\Gamma, x_k:\textsf{B}_k \nrightvdash P_k:S_k$.
      Correspondingly, since $J_{m+1} \in \Phi(J_m)$, we have
      $J_{m+1} = \Theta_{m+1}\cdot\Gamma_{m+1} \rightvdash P_{m+1}: S_{m+1}$
      with $\Theta_{m+1} = \Theta_m$, $\Gamma_{m+1} = \Gamma_m, x_k:\textsf{B}_k$, $P_{m+1} = P_k$ and $S_{m+1} = S_k$.
      By \Cref{def:synth-function}, we have\\
      $\lsem{S_n}\rsem = \envrcvOP\big\{\texttt{l}_i(x_i:\textsf{B}_i).\pif{\isValueB{i}{x_i}}{\triangleleft \texttt{l}_i(x_i).\lsem S_i\rsem}{\noed}\big\}_{i\in I}$.\\
      By the induction hypothesis,
      there is a set $\Sigma_{m+1} \neq \emptyset$
      containing substitutions with domain including $\fv{P_{m+1}} = \fv{P_k}$,
      and any substitution $\sigma_{m+1} \in \Sigma_{m+1}$
      can be used to build a trace $t''$ such that
      $\instr{P_{m+1}\sigma_{m+1}}{\lsem{S_{m+1}}\rsem} \xRightarrow{t''} \instr{P'}{M'} \mathrel{\not\rightarrow}{}$ with $P' \not\equal \pnil$ or $M' \not\equal \pnil$;
      moreover, $M' \not\equal \noes$ and $M' \not\equal \noed$.
      Now, we take:

      \smallskip\centerline{\(
      \Sigma'' \;=\; \left\{\sigma'' \in \Sigma_{m+1} \;\middle|\;
      \begin{array}{@{}l@{}}
        x_k \in \dom{\sigma''} \;\land\;
        \Gamma_{m+1} \rightvdash x_k\sigma'' : \textsf{B}_k
      \end{array}
      \right\}
      \)}\smallskip

      \ie all substitutions from $\Sigma_{m+1}$ that instantiate $x_k$ as a well-typed;
      at least one exists, by $\Gamma_{m+1} \rightvdash a:\textsf{B}_k$ and the substitution lemma.
      \fxASwarning{Is would be better to add this invariant to the i.h.}
      Thus, we conclude by taking
      $\Sigma_m = \Sigma_{m+1} \setminus \{x_i\}_{i \in I}$
      (\ie, we remove the variables bound by $P_m$'s external choice),
      $\sigma_m$ as any substitution from $\Sigma_m$ (hence, $\fv{P_m} \subseteq \dom{\sigma_m}$),
      $t' = \envrcvOP \texttt{l}_k(x_k\sigma'') \tau \tau t''$ 
      (where $\sigma'' \in \Sigma''$, the first action is the monitor receiving
      a valid message from the environment, the first $\tau$ is is the monitor validating the such input,
      and the second $\tau$ is the monitor and process synchronising)
      and $M_m = \lsem{S_m}\rsem$;

    \item{}[\textsc{nSel3}].\quad
      We have $P_m = \psndk{l}{a}.P''$, and $S_m = \stselNA{i}{I}$,
      and $\exists k \in I: \texttt{l} \equal \texttt{l}_k$ and
      $\Gamma_m \rightvdash a:\textsf{B}_k$
      and $\Theta_m\cdot\Gamma_m \nrightvdash P'':S_k$.
      Correspondingly, since $J_{m+1} \in \Phi(J_m)$, we have
      $J_{m+1} = \Theta_{m+1}\cdot\Gamma_{m+1} \rightvdash P_{m+1}: S_{m+1}$
      with $\Theta_{m+1} = \Theta_m$, $\Gamma_{m+1} = \Gamma_m$, $P_{m+1} = P''$ and $S_{m+1} = S_k$.
      By \Cref{def:synth-function}, we have\\
      $\lsem{S_n}\rsem = \triangleright\big\{\texttt{l}_i(x_i:\textsf{B}_i).\pif{\isValueB{i}{x_i}}{\envsndOP \texttt{l}_i(x_i).\lsem S_i\rsem}{\nopd}\big\}_{i\in I}$.\\
      By the induction hypothesis,
      there is a set $\Sigma_{m+1} \neq \emptyset$
      containing substitutions with domain including $\fv{P_{m+1}} = \fv{P''}$,
      and any substitution $\sigma_{m+1} \in \Sigma_{m+1}$
      can be used to build a trace $t''$ such that
      $\instr{P_{m+1}\sigma_{m+1}}{\lsem{S_{m+1}}\rsem} \xRightarrow{t''} \instr{P'}{M'} \mathrel{\not\rightarrow}{}$ with $P' \not\equal \pnil$ or $M' \not\equal \pnil$;
      moreover, $M' \not\equal \noes$ and $M' \not\equal \noed$.
      Now, we take:

      \smallskip\centerline{\(
      \Sigma_m \;=\; \left\{\sigma'' \in \Sigma_{m+1} \;\middle|\;
      \begin{array}{@{}l@{}}
        \fv{a} \subseteq \dom{\sigma''} \;\land\;
        \Gamma_m \rightvdash a\sigma'' : \textsf{B}_k
      \end{array}
      \right\}
      \)}\smallskip

      \ie all substitutions from $\Sigma_{m+1}$ that instantiate $a$ as a well-typed value;
      at least one exists, by $\Gamma_m \rightvdash a:\textsf{B}_k$ and the substitution lemma.
      \fxASwarning{Is would be better to add this invariant to the i.h.}
      Thus, we conclude by taking 
      $\sigma_m$ as any substitution from $\Sigma_m$ (hence, $\fv{P_m} \subseteq \dom{\sigma_m}$),
      $t' = \tau \tau \envsndOP \texttt{l}_k(a\sigma_m) t''$ (where the first $\tau$ is synchronisation between the process and monitor,
      the second $\tau$ is the monitor validating the process output,
      and then the monitor forwarding the message sent by the process)
      and $M_m = \lsem{S_m}\rsem$;

    \item{}[\textsc{nIf2}].\quad
      We have $P_m = \pif{A}{Q}{Q'}$, with $\Theta_m\cdot\Gamma_m \nrightvdash Q: S_m$;
      correspondingly, since $J_{m+1} \in \Phi(J_m)$, we have
      $J_{m+1} = \Theta_{m+1}\cdot\Gamma_{m+1} \rightvdash P_{m+1}: S_{m+1}$
      with $\Theta_{m+1} = \Theta_m$, $\Gamma_{m+1} = \Gamma_m$, $P_{m+1} = Q$ and $S_{m+1} = S_m$.
      By the induction hypothesis,
      there is a set $\Sigma_{m+1} \neq \emptyset$
      containing substitutions with domain including $\fv{P_{m+1}} = \fv{Q}$,
      and any substitution $\sigma_{m+1} \in \Sigma_{m+1}$
      can be used to build a trace $t''$ such that
      $\instr{P_{m+1}\sigma_{m+1}}{\lsem{S_{m+1}}\rsem} \xRightarrow{t''} \instr{P'}{M'} \mathrel{\not\rightarrow}{}$ with $P' \not\equal \pnil$ or $M' \not\equal \pnil$;
      moreover, $M' \not\equal \noes$ and $M' \not\equal \noed$.
      Now, we take:

      \smallskip\centerline{\(
      \Sigma_m \;=\; \left\{\sigma'' \cup \sigma''' \;\middle|\;
      \begin{array}{@{}l@{}}
        \sigma'' \in \Sigma_{m+1} \;\land\; A\sigma'' \Downarrow \true
        \\
        \dom{\sigma'''} = \fv{Q'} \setminus \dom{\sigma''}
      \end{array}
      \right\}
      \)}\smallskip

      Observe that $\Sigma_m \neq \emptyset$;
      otherwise, there would be no trace $t^{*}$ and substitution $\sigma^{*}$
      such that $P_\mu \xRightarrow{t^{*}} P_m\sigma^{*} \xrightarrow{\tau} Q\sigma^{*}$, \ie $P_\mu$ would contain dead code (by \Cref{def:no-dead-code}),
      hence $P$ would also contain dead code, which would contradict the hypothesis in the statement. \fxASwarning{Is it clear?}
      Thus, we conclude by taking 
      $\sigma_m$ as any substitution from $\Sigma_m$ (hence, $\fv{P_m} \subseteq \dom{\sigma_m}$),
      $t' = \tau t''$ (where the $\tau$ is the ``if'' statement in $P_m\sigma_m$ reducing to the ``true'' branch),
      and $M_m = \lsem{S_m}\rsem$;
      
    \item{}[\textsc{nIf3}].\quad
      Similar to case [\textsc{nIf2}] above,
      except that we use the inductive premise on $Q'$ (\ie the ``false'' branch of the ``if'' statement);
    \item{}[\textsc{nIf4}].\quad
      The failing derivation follows one of the two inductive premises of the rule:
      the proof is similar to either case [\textsc{nIf2}] or case [\textsc{nIf3}] above;
    \item{}[\textsc{nRec}].\quad
      We have $P_m = \pmu{X}.P''$, with $\Theta_m,X:S_m\cdot\Gamma_m \nrightvdash P'': S_m$;
      correspondingly, since $J_{m+1} \in \Phi(J_m)$, we have
      $J_{m+1} = \Theta_{m+1}\cdot\Gamma_{m+1} \rightvdash P_{m+1}: S_{m+1}$
      with $\Theta_{m+1} = \Theta_m,X:S_m$, $\Gamma_{m+1} = \Gamma_m$, $P_{m+1} = P''$ and $S_{m+1} = S_m$.
      By the induction hypothesis,
      there is a set of substitutions $\Sigma_{m+1} \neq \emptyset$
      such that any substitution $\sigma_{m+1} \in \Sigma_{m+1}$
      can be used to build a trace $t''$ such that
      $\instr{P_{m+1}\sigma_{m+1}}{\lsem{S_{m+1}}\rsem} \xRightarrow{t''} \instr{P'}{M'} \mathrel{\not\rightarrow}{}$ with $P' \not\equal \pnil$ or $M' \not\equal \pnil$;
      moreover, $M' \not\equal \noes$ and $M' \not\equal \noed$.
      Thus, we conclude by taking
      $\Sigma_m = \{ \sigma'' \in \Sigma_{m+1} \mid \sigma''(X) = P''\}$
      \fxASnote{At least one such $\sigma''$ exists because we have unfolded $P_\mu$
        so that its negated typing derivation never uses [\textsc{nPVar}] (\Cref{lem:neg-typing-unfolding})
        hence the failing derivation from $P''$ ends before reaching $X$,
        and the trace $t''$ ends before unfolding $\pmu{X}.P''$ again%
      }
      $\sigma_m$ as any substitution from $\Sigma_m$ (hence, $\fv{P_m} \subseteq \dom{\sigma_m}$),
      $t' = \tau t''$ (where the $\tau$ is the unfolding of the top-level recursion of $P_m$),
      and $M_m = \lsem{S_m}\rsem$.
    \end{itemize}
  \end{itemize}

  We have thus proved that if\; $\emptyset\cdot\emptyset \nrightvdash P_\mu:S$,
  \;then there is a trace $t$ such that (for some $P', M'$)\;
  $\instr{P_\mu}{M} \xRightarrow{t} \instr{P'}{M'} \mathrel{\not\rightarrow}{}$,
  with $P' \not\equal \pnil$ or $M' \not\equal \pnil$.
  Since $P$ and $P_\mu$ are weakly bisimilar (see above),
  we also have (for some $P', M'$)\;
  $\instr{P}{M} \xRightarrow{t} \instr{P'}{M'} \mathrel{\not\rightarrow}{}$,
  with $P' \not\equal \pnil$ or $M' \not\equal \pnil$:
  this is the thesis.
\end{proof}
 }{}%

\listoffixmes

\end{document}